\documentclass[sn-mathphys,Numbered]{sn-jnl}

\usepackage{subcaption}
\usepackage{graphicx}%
\usepackage{multirow}%
\usepackage{amsmath,amssymb,amsfonts}%
\usepackage{amsthm}%
\usepackage{mathrsfs}%
\usepackage[title]{appendix}%
\usepackage{xcolor}%
\usepackage{textcomp}%
\usepackage{manyfoot}%
\usepackage{booktabs}%
\usepackage{enumitem}
\usepackage{bbm}
\usepackage{footnote}
\usepackage{ragged2e,booktabs,tabularx,adjustbox}
\newcommand{\scrF}{\mathscr{F}}

\newtheorem{theorem}{Theorem}
\newtheorem{corollary}[theorem]{Corollary}

\begin{document}

\title[Article Title]{Fitting the seven-parameter Generalized Tempered Stable distribution to the financial data}


\author*[]{\fnm{Aubain} \sur{Nzokem}}\email{hilaire77@gmail.com}

\author[]{\fnm{Daniel} \sur{Maposa}}




\abstract{The paper proposes and implements a methodology to fit a seven-parameter Generalized Tempered Stable (GTS) distribution to financial data. The nonexistence of the mathematical expression of the GTS probability density function makes the maximum likelihood estimation (MLE) inadequate for providing parameter estimations. Based on the function characteristic and the fractional Fourier transform (FRFT), we provide a comprehensive approach to circumvent the problem and yield a good parameter estimation of the GTS probability. The methodology was applied to fit two heavily tailed data (Bitcoin and Ethereum returns) and two peaked data (S\&P 500 and SPY ETF returns). For each index, the estimation results show that the six-parameter estimations are statistically significant except for the local parameter, $\mu$. The goodness-of-fit was assessed through Kolmogorov-Smirnov, Anderson-Darling, and Pearson's chi-squared statistics.  While the two-parameter geometric Brownian motion (GBM) hypothesis is always rejected, the GTS distribution fits significantly with a very high p-value; and outperforms the Kobol, Carr-Geman-Madan-Yor, and Bilateral Gamma distributions.}




\maketitle
\section{Introduction}\label{sec1}
\noindent
Modeling the high-frequency asset return with the normal distribution is the underlying assumption in many financial tools, such as the Black-Scholes-Merton option pricing model and the risk metric variance-covariance technique to Value-at-Risk (VAR). However, substantial empirical evidence rejects the normal distribution for various asset classes and financial markets. The symmetric and rapidly decreasing tail properties of the normal distribution cannot describe the skewed and fat-tailed properties of the asset return distribution.\\

\noindent
The $\alpha$-stable distribution has been proposed \cite{ken1999levy,nolan2020univariate} as an alternative to the normal distribution for modeling asset return and many types of physical and economic systems. The theoretical and empirical argument is that the stable distribution generalizes the Central Limit Theorem regardless of the variance nature (finite or infinite)\cite{rachev2011financial,nzokem2024self}. There are two major drawbacks \cite{nolan2020univariate,borak2005stable}: firstly, the lack of closed formulas for densities and distribution functions, except for the normal distribution($\alpha=2$), Cauchy distribution ($\alpha=1$) and L\'evy distribution ($\alpha=\frac{1}{2}$)\cite{Tsallis1997}; secondly, most of the moments of the stable distribution are infinite. An infinite variance of the asset return leads to an infinite price for derivative instruments such as options.\\

\noindent
The Generalized Tempered Stable (GTS) distribution was developed to overcome the shortcomings of the two distributions, and the tails of the GTS distribution are heavier than the normal distribution but thinner than the stable distribution\cite{grabchak2010financial, kim2009new}. The general form of the GTS distribution can be defined by the following L\'evy measure ($V(dx)$) (\ref{eq:l1}):
 \begin{align}
V(dx) =\left(\frac{\alpha_{+}e^{-\lambda_{+}x}}{x^{1+\beta_{+}}} \boldsymbol{1}_{x>0} + \frac{\alpha_{-}e^{-\lambda_{-}|x|}}{|x|^{1+\beta_{-}}} \boldsymbol{1}_{x<0}\right) dx \label{eq:l1}
 \end{align}
\noindent
where $0\leq \beta_{+}\leq 1$, $0\leq \beta_{-}\leq 1$, $\alpha_{+}\geq 0$, $\alpha_{-}\geq 0$, $\lambda_{+}\geq 0$ and  $\lambda_{-}\geq 0$. More details on Tempered Stable distribution are provided \cite{kuchler2013tempered,rachev2011financial}.\\

\noindent
 The rich class of GTS distribution  (\ref{eq:l1}) has a myriad of applications ranging from financial to mathematical physics and economic systems. However, few studies \cite{massing2024, nzokem2022fitting, fallahgoul2021modelling} have covered the methods and techniques to estimate the parameters of the GTS distribution. This study aims to contribute to the literature by providing a methodology for fitting seven-parameter GTS distribution. As illustrations, the study used four historical prices: two heavily tailed data (Bitcoin and Ethereum returns) and two peaked data (S\&P 500 and SPY ETF returns). The GTS distribution is fitted to the underlying distribution of each data index and the goodness-of-fit analysis is carried out.  The main disadvantage of the GTS distribution is the lack of the closed form of the density, cumulative, and derivative functions. We use a computational algorithm, called the enhanced fast FRFT scheme \cite{nzokem2023enhanced}, to circumvent the problem.\\

\noindent
The rest of the paper is organized as follows: Section 2 provides some theoretical framework of the GTS distribution. Section 3 presents the multivariate maximum likelihood (ML) method and the analytic version of the two-parameter normal distribution. Section 4 presents the results of the GTS parameter estimations along with the associated statistical tests for the heavily-tailed Bitcoin and Ethereum returns. Section 5 fits the GTS distribution to the traditional indices S\&P 500 and SPY ETF returns, while Section 6 presents the results of the goodness-of-fit test. Section 7 provides the concluding remarks.

\section{Generalized Tempered Stable (GTS) Distribution}
\noindent
The L\'evy measure of the GTS distribution ($V(dx)$) is defined in (\ref{eq:l2}) as a product of a tempering function $q(x)$ and a L\'evy measure of the $\alpha$-stable distribution $V_{stable}(dx)$:

\begin{equation}
 \begin{aligned}
q(x) &= e^{-\lambda_{+}x} \boldsymbol{1}_{x>0} + e^{-\lambda_{-}|x|} \boldsymbol{1}_{x<0} \\
V_{stable}(dx) &=\left(\alpha_{+}\frac{1}{x^{1+\beta_{+}}} \boldsymbol{1}_{x>0} + \alpha_{-}\frac{1}{|x|^{1+\beta_{-}}} \boldsymbol{1}_{x<0} \right) dx \\ \label{eq:l2}
V(dx) =q(x)V_{stable}(dx)&=\left(\alpha_{+}\frac{e^{-\lambda_{+}x}}{x^{1+\beta_{+}}} \boldsymbol{1}_{x>0} + \alpha_{-}\frac{e^{-\lambda_{-}|x|}}{|x|^{1+\beta_{-}}} \boldsymbol{1}_{x<0}\right) dx
 \end{aligned}
 \end{equation}

\noindent
where $0\leq \beta_{+}\leq 1$, $0\leq \beta_{-}\leq 1$, $\alpha_{+}\geq 0$, $\alpha_{-}\geq 0$, $\lambda_{+}\geq 0$ and  $\lambda_{-}\geq 0$.\\
The six parameters that appear have important interpretations. $\beta_{+}$ and $\beta_{-}$ are the indexes of stability bounded below by 0 and above by 2 \cite{borak2005stable}. They capture the peakedness of the distribution similarly to the $\beta$-stable distribution, but the distribution tails are tempered. If $\beta$ increases (decreases), then the peakedness decreases (increases). $\alpha_{+}$ and $\alpha_{-}$ are the scale parameters, also called the process intensity \cite{boyarchenko2002non}; they determine the arrival rate of jumps for a given size. $\lambda_{+}$ and $\lambda_{-}$ control the decay rate on the positive and negative tails. Additionally, $\lambda_{+}$ and $\lambda_{-}$ are also skewness parameters. If $\lambda_{+}>\lambda_{-}$ ($\lambda_{+}<\lambda_{-}$), then the distribution is skewed to the left (right), and if $\lambda_{+}=\lambda_{-}$, then it is symmetric \cite{rachev2011financial, fallahgoul2019quantile}. $\alpha$ and $\lambda$ are related to the degree of peakedness and thickness of the distribution. If $\alpha$ increases (decreases), the peakedness and the thickness decrease (increase). Similarly, if $\lambda$ increases (decreases), then the peakedness increases (decreases) and the thickness decreases (increases) \cite{bianchi2019handbook}.  For more details on tempering function and L\'evy measure of tempered stable distribution, refer to \cite{kuchler2013tempered,rachev2011financial}.\\

\noindent
The activity process of the GTS distribution can be studied from the integral (\ref{eq:l3}) of the L\'evy measure (\ref{eq:l2}):
\begin{equation}
\begin{aligned}
 \int_{-\infty}^{+\infty} V(dx) =\begin{cases}
    +\infty   \quad \quad \quad \quad \quad \quad \quad \quad \quad \quad \quad  \quad \text{if }{0\leq \beta_{+} < 1 \wedge 0\leq\beta_{-} <1}    \\
  \alpha_{+}{\lambda_{+}}^{\beta_{+}}\Gamma(-\beta_{+}) +  \alpha_{-}{\lambda_{-}}^{\beta_{-}}\Gamma(-\beta_{-}) \quad \text{if }{\beta_{+} < 0\wedge\beta_{-} < 0}. \end{cases} \label{eq:l3}
 \end{aligned}
 \end{equation}

\noindent
As shown in (\ref{eq:l3}), if $\beta_{+}<0$ and $\beta_{-}<0$, $GTS(\beta_{+}, \beta_{-}, \alpha_{+},\alpha_{-}, \lambda_{+}, \lambda_{-})$ is of finite activity process and can be written as a compound Poisson \cite{barndorff2002financial}. When $0\leq \beta_{+}<1$ and $0\leq \beta_{-}< 1$, this L\'evy density ($V(dx)$) is not integrable as it goes off to infinity too rapidly as $x$ goes to zero \cite{barndorff2002financial}, which means in practice that there will be a large number of very small jumps. As shown in (\ref{eq:l3}), $GTS(\beta_{+}, \beta_{-}, \alpha_{+},\alpha_{-}, \lambda_{+}, \lambda_{-})$ is an infinite activity process with infinite jumps in any given time interval. \\

\noindent
In addition to the infinite activities process,  the variation of the process can be studied through the following integral:
\begin{equation*}
 \begin{aligned}
 \int_{-1}^{1} |x|V(dx)&= \int_{-1}^{0}|x|V(dx) + \int_{0}^{1} |x|V(dx) \\
  &=\alpha_{-}\lambda_{-}^{\beta_{-}-1}\gamma(1-\beta_{-},\lambda_{-}) + \alpha_{+}\lambda_{+}^{\beta_{+}-1}\gamma(1-\beta_{+},\lambda_{+})
  \end{aligned}
\end{equation*}
\noindent
where $\gamma(s,x)=\int_{0}^{x}y^{s-1}e^{-y}dy$ is the lower incomplete gamma function.\\

\indent { And we have:}
\begin{align}
  \int_{-1}^{1} |x|V(dx) <+ \infty  \hspace{5mm}
   \hbox{ if $0< \beta_{-}\leq 1$ \& $0< \beta_{+}\leq 1$}.\label{eq:l4}
\end{align}
\noindent
As shown in (\ref{eq:l4}), $GTS (\beta_{+}, \beta_{-}, \alpha_{+}, \alpha_{-}, \lambda_{+}, \lambda_{-})$ generates a finite variance process, which is contrary to the Brownian motion process. $GTS (\beta_{+}, \beta_{-}, \alpha_{+}, \alpha_{-}, \lambda_{+}, \lambda_{-})$ generates a type B L\'evy process \cite{barndorff2001levy}, which is a purely non-Gaussian infinite activity L\'evy process of finite variation whose sample paths have an infinite number of small jumps and a finite number of large jumps in any finite time interval.\\

\noindent
The GTS distribution can be denoted by $X\sim GTS(\beta_{+}, \beta_{-}, \alpha_{+},\alpha_{-}, \lambda_{+}, \lambda_{-})$ and $X=X_{+} - X_{-}$ with $X_{+} \geq 0$, $X_{-}\geq 0$. $X_{+}\sim TS(\beta_{+}, \alpha_{+},\lambda_{+})$ and $X_{-}\sim TS(\beta_{-}, \alpha_{-},\lambda_{-})$. By  adding the location parameter, the GTS distribution becomes $GTS(\mu, \beta_{+}, \beta_{-}, \alpha_{+}, \alpha_{-}, \lambda_{+}, \lambda_{-})$, and we have (\ref{eq:l5}):
 \begin{align}
Y=\mu + X=\mu + X_{+} - X_{-}, \quad \quad  Y\sim GTS(\mu, \beta_{+}, \beta_{-}, \alpha_{+}, \alpha_{-}, \lambda_{+}, \lambda_{-}). \label {eq:l5}
  \end{align}

\subsection{ GTS Distribution and Characteristic Exponent}
\begin{theorem}\label{lem5} \ \\
Consider a variable $Y \sim GTS(\mu, \beta_{+}, \beta_{-}, \alpha_{+},\alpha_{-}, \lambda_{+}, \lambda_{-})$. The characteristic exponent can be written as:
 \begin{equation}
\begin{aligned}
\resizebox{0.935\hsize}{!}{$\Psi(\xi)=\mu\xi i+\alpha_{+}\Gamma(-\beta_{+})\left((\lambda_{+} - i\xi)^{\beta_{+}}-{\lambda_{+}}^{\beta_{+}}\right)+\alpha_{-}\Gamma(-\beta_{-})\left((\lambda_{-}+i\xi)^{\beta_{-}}-{\lambda_{-}}^{\beta_{-}}\right)$}. \label{eq:l6}
 \end{aligned}
 \end{equation}
\end{theorem}
\begin{proof} \ \\
\noindent
$V(dx)$ in (\ref{eq:l2}) is a L\'evy measure. The following relation is satisfied from (\ref{eq:l4}):
\begin{align*}
  \int_{-\infty}^{+\infty} Min(1,|x|)V(dx) <+ \infty.
\end{align*}
More details on the proof are provided in \cite{mca29030044}.\\

\noindent
The L\'evy-Khintchine representation \cite{barndorff2002financial} for non-negative L\'evy process is applied on $Y$. $Y=\mu +X =\mu + X_{+} - X_{-}$ and we have:
 \begin{equation}
 \begin{aligned}
\Psi(\xi)&= Log\left(Ee^{i Y\xi}\right)=i\mu\xi + Log\left(Ee^{i X_{+}\xi}\right) + Log\left(Ee^{-i X_{-}\xi}\right) \label {eq:l6a}\\
&= i\mu\xi + \int_{0}^{+\infty} \left(e^{iy\xi} -1\right)\frac{\alpha_{+}e^{-\lambda_{+}y}}{y^{1+\beta_{+}}}dy + \int_{0}^{+\infty} \left(e^{-iy\xi} -1\right)\frac{\alpha_{-}e^{-\lambda_{-}y}}{y^{1+\beta_{-}}}dy,  \end{aligned}
 \end{equation}

  \begin{equation}
 \begin{aligned}
 \int_{0}^{+\infty} \left(e^{iy\xi} -1\right)\frac{\alpha_{+}e^{-\lambda_{+}y}}{y^{1+\beta_{+}}}dy&=\alpha_{+}\lambda_{+}^{\beta_{+}}\Gamma(-\beta_{+})\sum_{k=1}^{+\infty}{\frac{\Gamma(k-\beta_{+})}{\Gamma(-\beta_{+})k!}(\frac{i\xi}{\lambda_{+}})^{k}}\\
 &=\alpha_{+}\lambda_{+}^{\beta_{+}}\Gamma(-\beta_{+})\sum_{k=1}^{+\infty}{\binom{\beta_{+}}{k}(-\frac{i\xi}{\lambda_{+}})^{k}}\\
 &=\alpha_{+}\Gamma(-\beta_{+}) \left((\lambda_{+} - i\xi)^{\beta_{+}} - \lambda_{+}^{\beta_{+}}\right). \label {eq:l6b}
 \end{aligned}
 \end{equation}
 Similarly, we have :
 \begin{align}
 \int_{0}^{+\infty} \left(e^{-iy\xi} -1\right)\frac{\alpha_{-}e^{-\lambda_{-}y}}{y^{1+\beta_{-}}}dy=\alpha_{-}\Gamma(-\beta_{-}) \left((\lambda_{-} + i\xi)^{\beta_{-}} - \lambda_{-}^{\beta_{-}}\right). \label{eq:l6c}
  \end{align}
 \noindent
 The expression in (\ref{eq:l6a}) becomes:
 \begin{align*}
\Psi(\xi)=i\mu\xi + \alpha_{+}\Gamma(-\beta_{+}) \left((\lambda_{+} - i\xi)^{\beta_{+}} - \lambda_{+}^{\beta_{+}}\right) + \alpha_{-}\Gamma(-\beta_{-}) \left((\lambda_{-} + i\xi)^{\beta_{-}} - \lambda_{-}^{\beta_{-}}\right).
\end{align*}
\end{proof}

\begin{theorem}\label{lem6} \ \\
Consider a variable $Y \sim GTS(\mu, \beta_{+}, \beta_{-}, \alpha_{+},\alpha_{-}, \lambda_{+}, \lambda_{-})$. \\
 If $(\beta_{-},\beta_{+}) \to (0,0)$, GTS becomes a Bilateral Gamma distribution with the following characteristic exponent:
   \begin{align}
\Psi(\xi)=\mu\xi i - \alpha_{+}\log\left(1-\frac{1}{\lambda_{+}}i\xi\right)- \alpha_{-}\log\left(1+\frac{1}{\lambda_{-}}i\xi\right). \label {eq:l7}
\end{align}
In addition to $(\beta_{-},\beta_{+}) \to (0,0)$, if $\alpha_{-}=\alpha_{+}=\alpha$, GTS becomes Variance-Gamma (VG) distribution  with parameter $(\mu, \delta, \sigma, \alpha, \theta)$
   \begin{align*}
\delta=\lambda_{-}-\lambda_{+} \ \  \  \sigma=1 \  \    \  \alpha=\alpha_{-}=\alpha_{+} \ \ \  \theta=\frac{1}{\lambda_{-}\lambda_{+}}
  \end{align*}
and the following characteristic exponent:
   \begin{align}
\Psi(\xi)=\mu\xi i - \alpha \log \left(1 - \frac{\lambda_{-}-\lambda_{+}}{\lambda_{+}\lambda_{-}} i\xi +\frac{1}{\lambda_{+}\lambda_{-}}\xi^2 \right). \label {eq:l8}
  \end{align}
\end{theorem}

\begin{proof}
 \begin{equation}
 \begin{aligned}
\Gamma(-\beta_{+})&=-\frac{\Gamma(1-\beta_{+})}{\beta_{+}} \\
 \lim_{\beta_{+} \to 0} \Gamma(-\beta_{+})\left((\lambda_{+} - i\xi)^{\beta_{+}} - \lambda_{+}^{\beta_{+}}\right)&=- \log\left(1-\frac{1}{\lambda_{+}}i\xi\right).\label {eq:l7a}
 \end{aligned}
  \end{equation}
  	
\noindent
Similarly, (\ref{eq:l7a}) works for  $\beta_{-} \to 0$, and  we have the characteristic exponent (\ref{eq:l7}).\\

\noindent
In addition, if $\alpha_{-}=\alpha_{+}=\alpha$, from (\ref{eq:l7}), the characteristic exponent becomes:
\begin{align*}
\Psi(\xi)=\mu\xi i - \alpha \log \left(1 - \frac{\lambda_{-}-\lambda_{+}}{\lambda_{+}\lambda_{-}} i\xi +\frac{1}{\lambda_{+}\lambda_{-}}\xi^2 \right),
  \end{align*}
\noindent
which is a Variance-Gamma (VG) distribution with parameter $(\mu,\lambda_{-}-\lambda_{+},1,\alpha,\frac{1}{\lambda_{-}\lambda_{+}})$. For more details on the VG model, refer to \cite{nzokem2022, madan1998variance}.
\end{proof}
 \begin{theorem}\label{lem7} (Cumulants $\kappa_{k}$)\\
 Consider a variable $Y \sim GTS(\mu, \beta_{+}, \beta_{-}, \alpha_{+},\alpha_{-}, \lambda_{+}, \lambda_{-})$. The cumulants $\kappa_{k}$ of the GTS distribution are defined as follows:
\begin{equation}
 \begin{aligned}
 \kappa_{0}&=0\\
\kappa_{1}&= \mu + \alpha_{+}{\frac{\Gamma(1-\beta_{+})}{\lambda_{+}^{1-\beta_{+}}}} - \alpha_{-}{\frac{\Gamma(1-\beta_{-})}{\lambda_{-}^{1-\beta_{-}}}} \\
 \kappa_{k}&=\alpha_{+}{\frac{\Gamma(k-\beta_{+})}{\lambda_{+}^{k-\beta_{+}}}} + (-1)^{k} \alpha_{-}{\frac{\Gamma(k-\beta_{-})}{\lambda_{-}^{k-\beta_{-}}}}  \quad  \quad  \forall k \in \mathbb{N} \setminus \{0,1\} . \label {eq:l9}
 \end{aligned}
\end{equation}
 \end{theorem}
\begin{proof} \ \\
\noindent
We reconsider the characteristic exponent $\Psi(\xi)$  in (\ref{eq:l6a}):
\begin{equation}
 \begin{aligned}
\Psi(\xi)&= i\mu\xi  + \int_{0}^{+\infty} \left(e^{iy\xi} -1\right)\frac{\alpha_{+}e^{-\lambda_{+}y}}{y^{1+\beta_{+}}}dy + \int_{0}^{+\infty} \left(e^{-iy\xi} -1\right)\frac{\alpha_{-}e^{\lambda_{-}y}}{y^{1+\beta_{-}}}dy \\
&= i\mu\xi + \alpha_{+}\sum_{k=1}^{+\infty}{\frac{\Gamma(k-\beta_{+})}{\lambda_{+}^{k-\beta_{+}}}\frac{(i\xi)^{k}}{k!}} + \alpha_{-}\sum_{k=1}^{+\infty}{\frac{\Gamma(k-\beta_{-})}{\lambda_{-}^{k-\beta_{-}}}\frac{(-i\xi)^{k}}{k!}}\\
&= i\mu\xi + \sum_{k=1}^{+\infty}{\frac{1}{k!}\left(\alpha_{+}\frac{\Gamma(k-\beta_{+})}{\lambda_{+}^{k-\beta_{+}}} + \alpha_{-}\frac{\Gamma(k-\beta_{-})}{\lambda_{-}^{k-\beta_{-}}}(-1)^{k}\right)(i\xi)^{k}}\\
&=\sum_{k=0}^{+\infty}\frac{\kappa_{k}}{k!}(i\xi)^{k}. \label {eq:l9a}
 \end{aligned}
\end{equation}
\noindent
Hence, the $k$-th order cumulant $\kappa_{k}$  is given by comparing the coefficients of both polynomial functions in  $i\xi$. For more details on the relationship between the characteristic exponent and cumulant functions, refer to \cite{kendall1946advanced, feller1971}.
\end{proof}
\subsection{GTS Distribution and L\'evy Process}
\begin{corollary} \ \\
Let $Y=\left(Y_{t}\right)$ be a L\'evy process on $\mathbb{R^{+}}$ generated by $GTS(\mu, \beta_{+}, \beta_{-}, \alpha_{+},\alpha_{-}, \lambda_{+}, \lambda_{-})$, then
\begin{equation}
 \begin{aligned}
Y_{t}  \sim GTS( t\mu, \beta_{+}, \beta_{-}, t\alpha_{+}, t\alpha_{-}, \lambda_{+}, \lambda_{-}) \quad \quad  \forall t \in \mathbb{R}^{+}. \label {eq:l10}
 \end{aligned}
\end{equation}
\end{corollary}

\begin{proof}[Proof:]  \ \\
\noindent
Let $\Psi(\xi,t)$ be the characteristic exponent  of the L\'evy process  $Y=\left(Y_{t}\right)$. By applying the infinitely divisible property, we have:
 \begin{equation*}
 \begin{aligned}
\Psi(\xi,t)&= Log\left(Ee^{i Y_{t}\xi}\right)= t Log\left(Ee^{i X\xi}\right)\\
&=t\mu\xi i + t\alpha_{+}\Gamma(-\beta_{+})\left( (\lambda_{+} - i\xi)^{\beta_{+}} - \lambda_{+}^{\beta_{+}}\right) + t\alpha_{-}\Gamma(-\beta_{-})\left( (\lambda_{-} + i\xi)^{\beta_{-}} - \lambda_{-}^{\beta_{-}}\right)
  \end{aligned}
   \end{equation*}
and we deduce that  $Y_{t} \sim GTS(t\mu, \beta_{+}, \beta_{-}, t\alpha_{+},t\alpha_{-}, \lambda_{+}, \lambda_{-})$.
\end{proof}
 \begin{theorem}\label{lem7} (Asymptotic distribution of Generalized Tempered Stable distribution  process)\\
Let $Y={Y_{t}}$ be a L\'evy process on $\mathbb{R}$ generated by  $GTS(\mu, \beta_{+}, \beta_{-}, \alpha_{+},\alpha_{-}, \lambda_{+}, \lambda_{-})$.
Then  $Y_{t}$  converges in distribution to a L\'evy process driving by a normal distribution with mean $\kappa_{1}$ and variance $\kappa_{2}$

 \begin{equation}
 \begin{aligned}
Y_{t} \stackrel{d}{\rightarrow}N(t\kappa_{1},t\kappa_{2}) \quad &\text{as} \quad t \to +\infty \label {eq:l11}
  \end{aligned}
   \end{equation}
where
\begin{equation*}
 \begin{aligned}
\kappa_{1}&= \mu + \alpha_{+}{\frac{\Gamma(1-\beta_{+})}{\lambda_{+}^{1-\beta_{+}}}} - \alpha_{-}{\frac{\Gamma(1-\beta_{-})}{\lambda_{-}^{1-\beta_{-}}}} \\
 \kappa_{2}&=\alpha_{+}{\frac{\Gamma(2-\beta_{+})}{\lambda_{+}^{2-\beta_{+}}}} +  \alpha_{-}{\frac{\Gamma(2-\beta_{-})}{\lambda_{-}^{2-\beta_{-}}}}. \label {eq:l12}
  \end{aligned}
\end{equation*}
\end{theorem}

\begin{proof}[Proof:]  \ \\
\noindent
The proof relies on the cumulant-generating function.  As in (\ref{eq:l9a}), the characteristic exponent ($\Psi(\xi)$)  can be written as follows:
  \begin{equation}
 \begin{aligned}
\Psi(\xi)=Log\left(Ee^{i Y\xi}\right)= \sum_{j=0}^{+\infty}\kappa_{j}\frac{(i\xi)^{j}}{j!}. \label{eq:l12a}
  \end{aligned}
   \end{equation}

\noindent
Let $\phi(\xi,t)$ be the characteristic function of the stochastic  process $\frac{{Y_{t}}-t\kappa_{1}}{\sqrt{t\kappa_{2}}}$ and we have:
  \begin{equation}
 \begin{aligned}
\phi(\xi,t)&= E\left(e^{i \frac{{Y_{t}}-t\kappa_{1}}{\sqrt{t\kappa_{2}}}\xi}\right) =e^{-i \frac{t\kappa_{1}}{\sqrt{t\kappa_{2}}}\xi}E\left(e^{i \frac{\xi}{\sqrt{t\kappa_{2}}}Y_{t}}\right)\\
&=e^{-i \frac{t\kappa_{1}}{\sqrt{t\kappa_{2}}}\xi}e^{t\Psi(\frac{\xi}{\sqrt{t\kappa_{2}}})}=e^{-i \frac{t\kappa_{1}}{\sqrt{t\kappa_{2}}}\xi}e^{\sum_{j=0}^{+\infty}\frac{t\kappa_{j}}{j!}(i\frac{\xi}{\sqrt{t\kappa_{2}}})^{j}}\\
&=e^{-\frac{\xi^{2}}{2} + \sum_{j=3}^{+\infty}\frac{t\kappa_{j}}{j!}(i\frac{\xi}{\sqrt{t\kappa_{2}}})^{j}}, \label{eq:l12b}
 \end{aligned}
   \end{equation}

 \begin{equation}
 \begin{aligned}
\lim_{t \to +\infty} \sum_{j=3}^{+\infty}\frac{t\kappa_{j}}{j!}(i\frac{\xi}{\sqrt{t\kappa_{2}}})^{j}=0 \ \
\lim_{t \to +\infty}\phi(\xi,t)&=\lim_{t \to +\infty}e^{-\frac{\xi^{2}}{2} + \sum_{j=3}^{+\infty}\frac{t\kappa_{j}}{j!}(i\frac{\xi}{\sqrt{t\kappa_{2}}})^{j}} \\
&=e^{-\frac{1}{2}\xi^2}. \label{eq:l12c}
 \end{aligned}
   \end{equation}
\end{proof}
\section{Multivariate Maximum Likelihood Method}
\subsection{ Maximum Likelihood Method: Numerical Approach}
\noindent
From a probability density function $f(x, V)$ with parameter $V=(\mu, \beta_{+}, \beta_{-}, \alpha_{+},\alpha_{-}, \lambda_{+}, \lambda_{-})$ and a sample data $x=\left(x_{j}\right)_{1 \leq j \leq m} $,  we define  the likelihood function, and its first and second derivatives as follows: 
  \begin{equation}
 \begin{aligned}
 L_{m}(x,V) &= \prod_{j=1}^{m} f(x_{j},V), \quad  l_{m}(x,V) = \sum_{j=1}^{m} log(f(x_{j},V))  \\
\frac{dl_{m}(x,V)}{dV_j} &= \sum_{i=1}^{m} \frac{\frac{df(x_{i},V)}{dV_j}}{f(x_{i},V)} \\
  \frac{d^{2}l_{m}(x,V)}{dV_{k}dV_{j}} &= \sum_{i=1}^{m} \left(\frac{\frac{d^{2}f(x_{i},V)}{dV_{k}dV_{j}}}{f(x_{i},V)}- \frac{\frac{df(x_{i},V)}{dV_{k}}}{f(x_{i},V)}\frac{\frac{df(x_{i},V)}{dV_j}}{f(x_{i},V)}\right). \label {eq:l31}
 \end{aligned}
   \end{equation}
\noindent
To perform the maximum of the likelihood function  ( $L_{m}(x, V)$), we need the gradient of the likelihood function ($\frac{dl_{m}(x, V)}{dV}$) also known as the score function, and the Hessian matrix ($\frac{d^{2}l_{m}(x, V)}{dV{dV}'}$), which is the variance-covariance matrix generated by the likelihood function. \\

\noindent
Given the parameters $V=(\mu, \beta_{+}, \beta_{-}, \alpha_{+},\alpha_{-}, \lambda_{+}, \lambda_{-})$ and the sample data set $X$, we have the following quantities (\ref{eq:l32}) from the previous development:
  \begin{align}
 I_{m}'(X,V) =\left(\frac{dl_{m}(x,V)}{dV_j}\right)_{1 \leq j \leq p},   \quad  \quad  I_{m}''(X,V) = \left(\frac{d^{2}l_{m}(x,V)}{dV_{k}dV_{j}}\right)_{\substack{{1 \leq k \leq p} \\ {1 \leq j \leq p}}}. \label {eq:l32}
 \end{align}
\noindent
We use a computational algorithm built as a composite of a standard FRFT to compute the likelihood function and its derivatives (\ref{eq:l31}) in the optimization process. More details on applying the composite of FRFTs for parameter estimations are provided in \cite{nzokem2022fitting, Nzokem_2022, Nzokem_Montshiwa_2023, Nzokem_2021}; for other computations (such probability density and cumulative functions), see \cite{nzokem2023european, mca29030044, eberlein2010analysis, Eberlein2014, cherubini2010fourier}.\\

\noindent
The computational algorithm  yields a local solution, $V$, and a negative semi-definite matrix, $I_{m}''(x,V)$, when the following  two conditions are satisfied:
  \begin{align}
 I_{m}'(x,V)=0, \quad \quad U^{T}\mathbf{I_{m}''(X,V)}U \leq 0\hspace{5mm},  \hbox{ $\forall U\in\mathbb{R}^{p}$}.\label{eq:l33}
 \end{align}
The solutions, $V$, in (\ref{eq:l33}) are provided by the Newton-Raphson iteration algorithm formula (\ref{eq:l34}):
  \begin{align}
V^{n+1}=V^{n}-{\left(I_{m}''(x,V^{n})\right)^{-1}}I_{m}'(x,V^{n}).\label{eq:l34}
 \end{align}
\noindent
More detail on maximum likelihood and Newton-Raphson iteration procedure are provided in \cite{giudici2013wiley}.

\subsection{Asymptotic Distribution of the Maximum Likelihood Estimator (MLE)}
\begin{theorem}\label{lem3a} (Cramer-Rao) \ \\
Let $T = T(X_{1},..., X_{m})$ be a statistic and write $E[T]= k(\theta)$. Then, under suitable (smoothness) assumptions,
\begin{align}
Var[T] \geq \frac{(\frac{dE[T]}{d\theta})^2}{m I(\theta)}. \label{eq:l35}
 \end{align}
\end{theorem}
For the proof of Theorem \ref{lem3a} refer to \cite{van2007parameter,casella2024statistical}.\\

\begin{theorem}\label{lem3b} (Consistency Estimator) \ \\
Let $X_{1}, ..., X_{m}$ be independent and identically distributed (i.i.d) random variables with density $f(x|\theta)$ satisfying some regularity conditions \cite{lehmann1999elements}. Let $\theta$ be the true parameter, then there exists a sequence $\hat{\theta}_{m}=\theta_{m}(X_{1}, ..., X_{m})$ of local maxima of the likelihood function $L_{m}(\theta)$ which is consistent, that is, which satisfies
  \begin{equation}
 \begin{aligned}
\hat{\theta}_{m}  \stackrel{a.s.}{\rightarrow} \theta \quad &\text{as} \quad m\to +\infty. \label{eq:l35a}
 \end{aligned}
   \end{equation}
\end{theorem}

More details on the proof of Theorem \ref{lem3b} are provided in \cite{lehmann1999elements, casella2024statistical}.\\

\begin{theorem}\label{lem3c} (Asymptotic Efficiency and Normality) \ \\
Let $X_{1}, ..., X_{m}$ be independent and identically distributed (i.i.d) random variables with density $f(x |\theta)$ satisfying some regularity conditions in \cite{lehmann1999elements}.
There exists a solution $\hat{\theta}_{m}=\theta_{m}(X_{1}, ..., X_{m})$ of the likelihood equations which
is consistent, and any such solution satisfies:
 \begin{equation}
 \begin{aligned}
\hat{\theta}_{m} -\theta  \stackrel{d}{\rightarrow} N\left(0, I^{-1}_{m}(\theta)\right) \quad &\text{as} \quad m \to +\infty, \label{eq:l35a}
 \end{aligned}
 \end{equation}
where $\theta=(\theta_{1},...,\theta_{k})$ is the actual parameter and $I_{m}(\theta)$ is the Fisher information matrix.\\

\end{theorem}

More details on the proof of Theorem \ref{lem3c} are provided in \cite{lehmann1999elements, olive2014statistical, hall2023course}. \\

\begin{theorem}\label{lem3d} (Likelihood Ratio Test) \ \\
Suppose the assumptions of Theorem \ref{lem3c} hold and that $(\hat{\theta}_{1n},\dotsc, \hat{\theta}_{kn})$ are consistent roots of the likelihood equations for $\theta= (\theta_{1},\dotsc, \theta_{k})$ . In addition, suppose that the corresponding assumptions hold for the parameter vector $(\theta_{r+1},\dotsc, \theta_{k})$ when $r<k$, and that $\hat{\hat{\theta}}_{r+1,n},\dotsc, \hat{\hat{\theta}}_{kn}$ are consistent roots of the likelihood equations for $(\theta_{r+1},\dotsc, \theta_{k})$  under null hypothesis. We consider the likelihood ratio statistic
 \begin{equation}
 \begin{aligned}
\frac{l_{m}(x,\hat{\theta})}{l_{m}(x,\hat{\hat{\theta}})}\label{eq:l36}
 \end{aligned}
 \end{equation}
where $\hat{\hat{\theta}}=(\theta_{1}, \dotsc, \theta_{r},\hat{\hat{\theta}}_{r+1,n},\dotsc, \hat{\hat{\theta}}_{kn})$.
Then under the null hypothesis $H_{0}$, if
 \begin{equation}
 \begin{aligned}
\Delta_{n} = l_{m}(x,\hat{\theta}) - l_{m}(x,\hat{\hat{\theta}}), \label{eq:l37}
 \end{aligned}
 \end{equation}
the statistic $2\Delta_{n}$ has a limiting $\chi^2_{r}$ distribution.\\
\end{theorem}

More details on the proof of Theorem \ref{lem3d} are provided in \cite{lehmann1999elements, vuong1989likelihood}.\\

\subsection{Asymptotic Test and Confidence Interval}
\noindent
The above results allow us to construct an asymptotically efficient estimator $\hat{\theta}_{m}=(\hat{\theta}_{1m},..., \hat{\theta}_{km})$ of  $\theta=(\theta_{1},\dotsc,\theta_{k})$ such that
 \begin{equation}
 \begin{aligned}
 (\hat{\theta}_{1m}-\theta_{1},..., \hat{\theta}_{km}-\theta_{k})
 \end{aligned}
 \end{equation}
 has a joint multivariate limit distribution with mean $(0,...,0)$  and covariance matrix $I^{-1}_{m}(\theta)=(J_{ij})$. In particular, we have:
 \begin{equation}
 \begin{aligned}
 \hat{\theta}_{jm}-\theta_{j} \stackrel{d}{\rightarrow} N(0,J_{jj}) \quad \text{as} \quad m \to +\infty. \label{eq:l37a}
 \end{aligned}
 \end{equation}

\noindent
One approach to constructing an asymptotically valid confidence interval for the parameters is via the asymptotic distribution of the ML estimator (\ref{eq:l36}). An approximate $(1 - \frac{\alpha}{2})$ confidence interval for $\hat{\theta}_{jm}$ can be written as follows:
 \begin{equation}
 \begin{aligned}
\hat{\theta}_{jm} \pm z(\frac{\alpha}{2})*\sqrt{J_{jj}} \quad  \text{as} \quad m \to +\infty, \label{eq:l37b}
 \end{aligned}
 \end{equation}
 where $z(\frac{\alpha}{2})$ is the $\frac{\alpha}{2}$ quantile of the standard normal distribution.

\subsection{Applications of the Log-likelihood Estimator to the Normal Distribution}
\noindent
We suppose the sample data $x=\left(x_{j}\right)_{1 \leq j \leq m}$ are independent observations and have a normal distribution \cite{MENSAH2023} with parameter $V(\mu,\sigma^2)$, that is,  $y \sim \mathcal{N}(\mu,\sigma^2) $, then the density is
\begin{align}
f(y |V)=(2\pi\sigma^2)^{-\frac{1}{2}} exp \left(-\frac{(y-\mu)^2}{2\sigma^2} \right). \label{eq:l37c}\end{align}
\noindent
The log-likelihood function in (\ref{eq:l31}) becomes
\begin{align}
 l_{m}(x|V) = \sum_{j=1}^{m} log(f(x_{j} | V)) = -\frac{m}{2} log(2\pi\sigma^2) -\frac{1}{2\sigma^2} \sum_{j=1}^{m}(x_{j}-\mu)^2. \label{eq:l38}\end{align}

\noindent
The first-order derivatives of the log-likelihood function with respect to $\mu$ and $\sigma^2$  in (\ref{eq:l31}) becomes
  \begin{equation}
 \begin{aligned}
 I_{m}'(X,V) =\begin{pmatrix}
 \frac{dl_{m}(x,V)}{d\mu}\\
 \frac{dl_{m}(x,V)}{d\sigma^2}
\end{pmatrix}=\begin{pmatrix}
\frac{1}{\sigma^2} \sum_{j=1}^{m}(x_{j} - \mu)\\
 \frac{1}{2\sigma^4} \sum_{j=1}^{m}(x_{j}-\mu)^2 - \frac{m }{2\sigma^2}. \label{eq:l38a}
\end{pmatrix}
 \end{aligned}
 \end{equation}

 \noindent
By setting $ I_{m}'(X,V) =0$, we have
\begin{align}
 \hat{\mu}=\frac{1}{m}\sum_{j=1}^{m}x_{j} \quad \quad   \hat{\sigma}^2=\frac{1}{m}\sum_{j=1}^{m}(x_{j}-\hat{\mu})^2. \label{eq:l39b}
 \end{align}

\noindent
The second-order derivative of the log-likelihood function with respect to $\mu$ and $\sigma^2$ in (\ref{eq:l31}) becomes
 \begin{equation}
 \begin{aligned}
 I_{m}''(X,V) &= \begin{pmatrix}
\frac{d^2l_{m}(x,V)}{d\mu^2}  &  \frac{dl_{m}(x,V)}{d\mu d\sigma^2} \\
\frac{dl_{m}(x,V)}{d\sigma^2 d\mu} & \frac{d^2l_{m}(x,V)}{(d\sigma^2)^2}
\end{pmatrix}\\
&=\begin{pmatrix}
 -\frac{m}{\sigma^2}  & -\frac{\sum_{j=1}^{m}(x_{j} - \mu)}{\sigma^4} \\
-\frac{\sum_{j=1}^{m}(x_{j} - \mu)}{2\sigma^4} &  - \frac{1}{\sigma^6} \sum_{j=1}^{m}(x_{j}-\mu)^2 +\frac{m}{2\sigma^4}.
\end{pmatrix}
 \label {eq:l39c}
  \end{aligned}
\end{equation}

Refer to \cite{casella2024statistical} for more details.\\

\noindent
We have the Fisher information matrix and the inverse:
 \begin{equation}
 \begin{aligned}
 I_{m}(V) = -E\left(I_{m}''(X,V)\right)=\begin{pmatrix}
 \frac{m}{\sigma^2}  & 0 \\
0 & \frac{m}{2 \sigma^4},
\end{pmatrix} \quad \quad  I_{m}^{-1}(V) =\begin{pmatrix}
 \frac{\sigma^2}{m}  & 0 \\
0 & \frac{2 \sigma^4}{m}
\end{pmatrix}.\label{eq:l39d}
 \end{aligned}
 \end{equation}

\begin{corollary}\label{lem34}  \ \\
The limiting distribution of the MLE is given by:
 \begin{equation}
 \begin{aligned}
\begin{pmatrix} \hat{\mu} \\ \hat{\sigma}^2
\end{pmatrix}   \stackrel{d}{\rightarrow} N\left(\begin{pmatrix} \mu \\ \sigma^2 \end{pmatrix},\begin{pmatrix} \frac{\sigma^2}{m}  & 0 \\ 0 & \frac{2 \sigma^4}{m}\end{pmatrix}\right),   \quad  \text{as} \quad m \to +\infty. \label{eq:l39e} 
 \end{aligned}
 \end{equation}
\end{corollary}
The proof of Corollary \ref{lem34} comes from Theorem \ref{lem3c}, Equation (\ref{eq:l35a}).

\clearpage
\newpage
\section{Fitting Tempered Stable Distribution to Cryptocurrencies: Bitcoin BTC and Ethereum \label{estim}}
\subsection{Data Summaries}
 \noindent
 Bitcoin was the first cryptocurrency created in 2009 by Satoshi Nakamoto. The idea behind Bitcoin was to create
a peer-to-peer electronic payment system that allows online payments to be sent directly from one party to another without going through a financial institution\cite{nakamoto2008bitcoin}. Since its inception, Bitcoin has grown in popularity and adoption and is now viewed as a viable legal tender in some countries. Bitcoin is currently used more as an investment tool, a risk-diversified tool, and less as a medium of exchange, a store of value, or a unit of account \cite{mca29030044}. \\

  \noindent
Bitcoin (BTC) and Ethereum (ETH) prices were extracted from CoinMarketCap. The period spans from April 28, 2013, to July 04, 2024, for Bitcoin, and from August 07, 2015, to July 04, 2024, for Ethereum.\\

The daily price dynamics are provided in Fig \ref{fig1}. The prices have an increasing trend, even after having major significant increases and decreases over the studied period. Fig \ref{fig10} and Fig \ref{fig11}  show that Bitcoin outperforms Ethereum, which is the second-largest cryptocurrency by market capitalization after Bitcoin.
 \begin{figure}[ht]
\vspace{-0.4cm}
    \centering
\hspace{-1cm}
  \begin{subfigure}[b]{0.45\linewidth}
    \includegraphics[width=\linewidth]{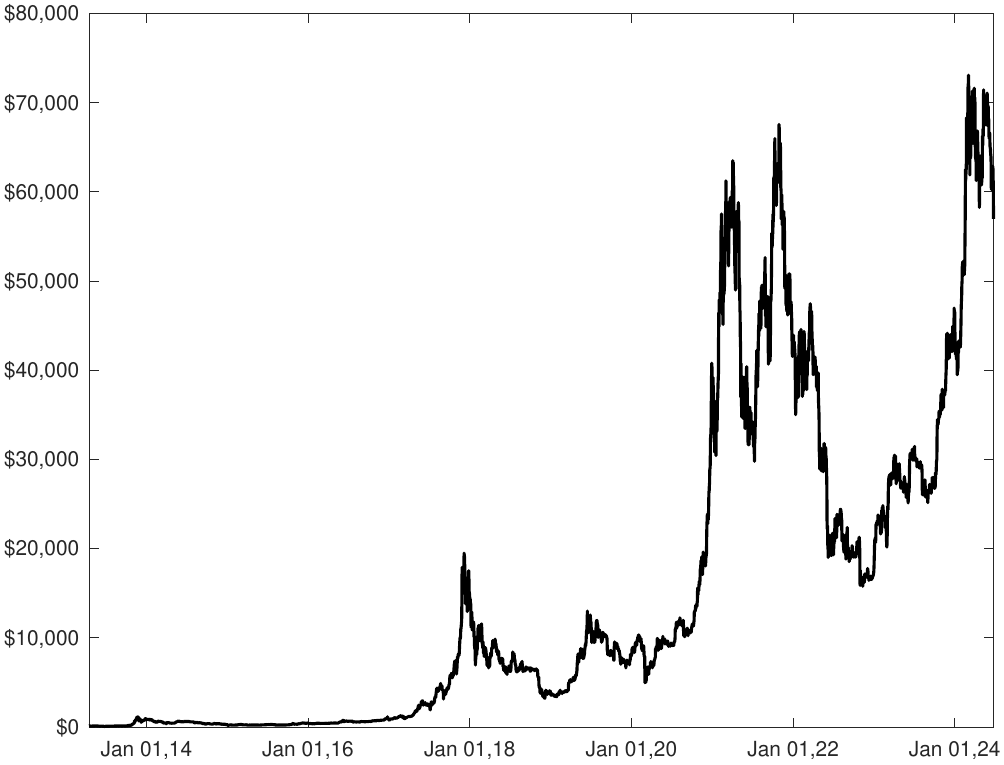}
\vspace{-0.6cm}
     \caption{Bitcoin (BTC) Daily Price}
         \label{fig10}
  \end{subfigure}
  \begin{subfigure}[b]{0.45\linewidth}
    \includegraphics[width=\linewidth]{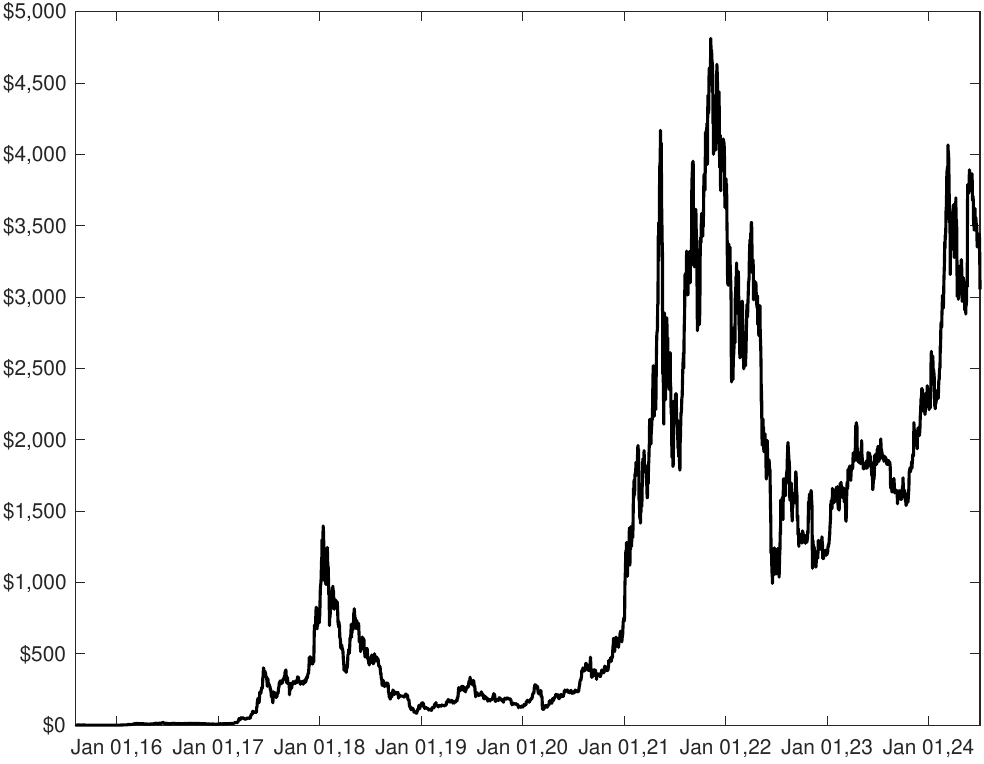}
\vspace{-0.5cm}
     \caption{Ethereum (ETH) Daily Price}
         \label{fig11}
          \end{subfigure}
\vspace{-0.6cm}
  \caption{Daily Price}
  \label{fig1}
\vspace{-0.6cm}
\end{figure}

\noindent
Let $m$, the number of observations, and $S_{j}$, the daily observed price on the day $t_{j}$ with $j=1,\dots,m$. The daily return $(y_{j})$ is computed as follows:
\begin{align}
y_{j}=\log(S_{j}/S_{j-1}) \hspace{10 mm}  \hbox{ $j=2,\dots,m$}.\label{eq:l40}
 \end{align}
\noindent  
As shown in Fig \ref{fig21} and Fig \ref{fig22}, the daily return reaches the lowest level ($-46\%$ for Bitcoin and $-55\%$ for Ethereum) in the first quarter of $2020$ amid the coronavirus pandemic and massive disruptions in the global economy. Nine values were identified as outliers and removed from the data set to avoid a negative impact on the GTS model estimation and the empirical statistics. 
\begin{figure}[ht]
\vspace{-0.2cm}
    \centering
  \begin{subfigure}[b]{0.4\linewidth}
    \includegraphics[width=\linewidth]{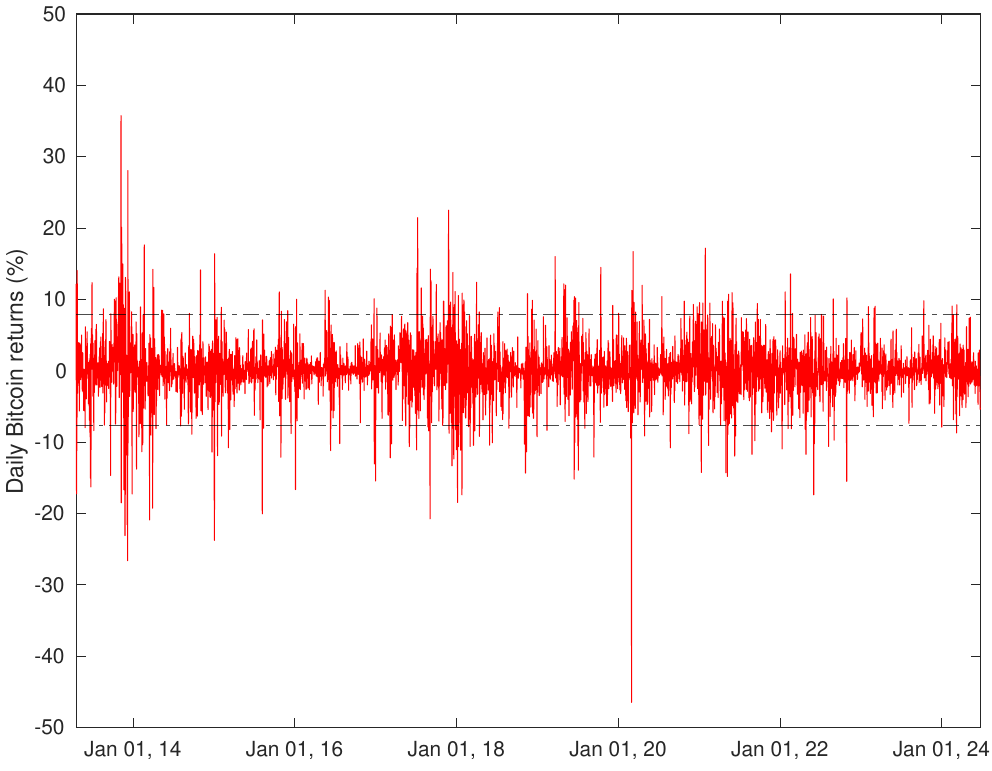}
\vspace{-0.5cm}
     \caption{Daily Bitcoin Return}
         \label{fig21}
  \end{subfigure}
  \begin{subfigure}[b]{0.4\linewidth}
    \includegraphics[width=\linewidth]{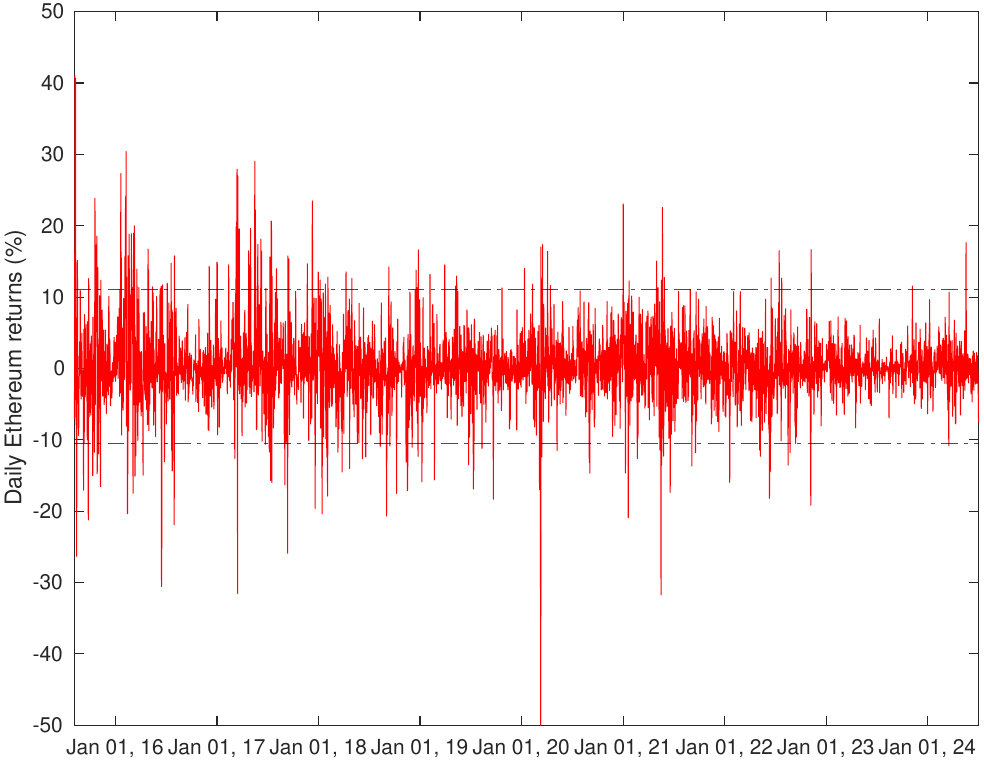}
\vspace{-0.5cm}
     \caption{ Daily Ethereum Return }
         \label{fig22}
          \end{subfigure}
\vspace{-0.6cm}
  \caption{Daily Return}
  \label{fig2}
\vspace{-0.6cm}
\end{figure}
\subsection{Multidimensional Estimation Results for Cryptocurrencies}
\noindent
 The results of the GTS parameter estimation are summarised in Table \ref{tab1} for Bitcoin and Table \ref{tab2} for Ethereum data. The brackets are the asymptotic standard errors computed using the inverse of the Hessian matrix built in (\ref{eq:l31}). The ML estimate of $\mu$ is negative for both Bitcoin and Ethereum, while others are positive, as expected in the literature. The asymptotic standard error for $\mu$ is quite large and suggests that $\mu$ is not statistically significant at 5\%.\\
 \noindent
 The log-likelihood, Akaike's information Criteria (AIC), and Bayesian information criteria (BIK) statistics show that the GTS distribution with seven parameters performs better than the two-parameter Normal distribution (GBM). A comprehensive and detailed examination of the statistical significance of the results will be carried out in Section \ref{test}.\\
 \noindent
 Table \ref{tab1} provides a summary of the estimation results for Bitcoin returns. The skewness parameters ($\lambda_{+}$, $\lambda_{-}$) are statistically significant at 5\%.
  \begin{table}[ht]
 \vspace{-0.6cm}
\centering
\caption{Maximum Likelihood  GTS Parameter Estimation for Bitcoin}
\label{tab1}
\begin{tabular}{@{} c|c|c|ccccc @{}}
\toprule
\multicolumn{1}{c|}{\textbf{Model}} & \multicolumn{1}{c|}{\textbf{Parameter}} & \multicolumn{1}{c|}{\textbf{Estimate}} & \multicolumn{1}{c|}{\textbf{Std Err}} & \multicolumn{1}{c|}{\textbf{z}} & \multicolumn{1}{c|}{\textbf{$Pr(Z >|z|)$}} & \multicolumn{2}{c}{\textbf{[95\% Conf.Interval]}}  \\ \toprule

\multirow{10}{*}{\textbf{GTS}} & \multirow{1}{*}{\textbf{$\mu$}} &\multirow{1}{*}{-0.121571} & \multirow{1}{*}{(0.375)} & \multirow{1}{*}{-0.32} & \multirow{1}{*}{7.5E-01} & \multirow{1}{*}{-0.856} & \multirow{1}{*}{0.613}  \\

 & \multirow{1}{*}{\textbf{$\beta_{+}$}} & \multirow{1}{*}{0.315548} & \multirow{1}{*}{(0.136)} & \multirow{1}{*}{2.33} & \multirow{1}{*}{2.0E-02} & \multirow{1}{*}{0.050} & \multirow{1}{*}{0.581}  \\

 &\multirow{1}{*}{ \textbf{$\beta_{-}$}} & \multirow{1}{*}{0.406563} & \multirow{1}{*}{(0.117)} & \multirow{1}{*}{3.48} & \multirow{1}{*}{4.9E-04} & \multirow{1}{*}{0.178} & \multirow{1}{*}{0.635}  \\

 & \multirow{1}{*}{\textbf{$\alpha_{+}$}} & \multirow{1}{*}{0.747714} & \multirow{1}{*}{(0.047)} & \multirow{1}{*}{15.76} & \multirow{1}{*}{6.2E-56} & \multirow{1}{*}{0.655} & \multirow{1}{*}{0.841}  \\

 & \multirow{1}{*}{\textbf{$\alpha_{-}$}} & \multirow{1}{*}{0.544565} & \multirow{1}{*}{(0.037)} & \multirow{1}{*}{14.56} & \multirow{1}{*}{4.8E-48} & \multirow{1}{*}{0.471} & \multirow{1}{*}{0.618}  \\

 &\multirow{1}{*}{\textbf{$\lambda_{+}$}} & \multirow{1}{*}{0.246530} & \multirow{1}{*}{(0.036)} & \multirow{1}{*}{6.91} & \multirow{1}{*}{4.9E-12} & \multirow{1}{*}{0.177} & \multirow{1}{*}{0.316}  \\

 & \multirow{1}{*}{\textbf{$\lambda_{-}$}} & \multirow{1}{*}{0.174772} & \multirow{1}{*}{(0.026)} & \multirow{1}{*}{6.69} & \multirow{1}{*}{2.2E-11} & \multirow{1}{*}{0.124} & \multirow{1}{*}{0.226}  \\ \cmidrule(l){2-8}

  & \multirow{1}{*}{\textbf{Log(ML)}} & \multirow{1}{*}{-10606} & \multirow{1}{*}{} & \multirow{1}{*}{} & \multirow{1}{*}{} & \multirow{1}{*}{} & \multirow{1}{*}{}  \\
 & \multirow{1}{*}{\textbf{AIC}} & \multirow{1}{*}{21227} & \multirow{1}{*}{} & \multirow{1}{*}{} & \multirow{1}{*}{} & \multirow{1}{*}{} &  \multirow{1}{*}{}  \\
 & \multirow{1}{*}{\textbf{BIK}} & \multirow{1}{*}{21271} & \multirow{1}{*}{} & \multirow{1}{*}{} & \multirow{1}{*}{} & \multirow{1}{*}{} &  \multirow{1}{*}{} \\ \toprule
 \multirow{5}{*}{\textbf{GBM}} & \multirow{1}{*}{\textbf{$\mu$}} & \multirow{1}{*}{0.151997} & \multirow{1}{*}{(0.060)} & \multirow{1}{*}{2.51} & \multirow{1}{*}{1.2E-02} & \multirow{1}{*}{0.033} & \multirow{1}{*}{0.271} \\
 & \multirow{1}{*}{\textbf{$\sigma$}} & \multirow{1}{*}{3.865132} & \multirow{1}{*}{(0.330)} & \multirow{1}{*}{11.69} & \multirow{1}{*}{7.2E-32} & \multirow{1}{*}{3.217} & \multirow{1}{*}{4.513}  \\ \cmidrule(l){2-8}
 & \multirow{1}{*}{\textbf{Log(ML)}} & \multirow{1}{*}{-11313} & \multirow{1}{*}{} & \multirow{1}{*}{} & \multirow{1}{*}{} & \multirow{1}{*}{} &  \multirow{1}{*}{}   \\
 & \multirow{1}{*}{\textbf{AIC}} & \multirow{1}{*}{22630} & \multirow{1}{*}{} & \multirow{1}{*}{} & \multirow{1}{*}{} & \multirow{1}{*}{} &  \multirow{1}{*}{}   \\
  & \multirow{1}{*}{\textbf{BIK}} & \multirow{1}{*}{22638} & \multirow{1}{*}{} & \multirow{1}{*}{} & \multirow{1}{*}{} & \multirow{1}{*}{} &  \multirow{1}{*}{}  \\ \bottomrule
\end{tabular}
\vspace{-0.5cm}
\end{table}

The difference is positive and statistically significant, which proves that the Bitcoin return is asymmetric and skewed to the left. The process intensity parameters ($\alpha_{+}$, $\alpha_{-}$) are statistically significant at 5\%. Similarly, the difference is positive and statistically significant, which shows Bitcoin is more likely to produce positive returns than negative returns. The index of stability parameters ($\beta_{+}$, $\beta_{-}$) are both statistically significant at 5\%.\\
However, the difference is positive but not statistically significant. The GTS distribution with $\beta=\beta_{+}=\beta_{-}$, called Kobol distribution, was fitted as well, and the estimation results are presented in Appendix \ref{eq:an3}. As shown in Table \ref{tabc1}, all the parameters are statistically significant at 5\%, and have the expected positive sign. However, the likelihood ratio test in Table \ref{tab6} shows that the GTS distribution is not significantly different from the Kobol distribution as the p-value ($69.6\%$) is large. Refer to \cite{boyarchenko2002non} for more details on Kobol distribution\\

\noindent
The parameters for Ethereum returns data are statistically significant at 5\%, except $\mu$ and $\beta_{-}$. The difference ($\lambda_{+} - \lambda_{-}$) in skewness parameters is negative and not statistically significant, showing that the Ethereum return is asymmetric and skewed to the right. Similarly, the difference ($\alpha_{+} - \alpha_{-}$) in the intensity parameters is positive and not statistically significant, as shown the confidence interval in Table \ref{tab2}. Contrary to the Bitcoin return, the Ethereum return has a larger process intensity, which provides evidence that Ethereum has a lower level of peakedness and a higher level of thickness. 
\begin{table}[ht]
\vspace{-0.5cm}
\centering
\caption{Maximum Likelihood  GTS Parameter Estimation for Ethereum}
\label{tab2}
\begin{tabular}{@{} c|c|c|ccccc@{}}
\toprule
\multicolumn{1}{c|}{\textbf{Model}} & \multicolumn{1}{c|}{\textbf{Param}} & \multicolumn{1}{c|}{\textbf{Estimate}} & \multicolumn{1}{c|}{\textbf{Std Err}} & \multicolumn{1}{c|}{\textbf{z}} & \multicolumn{1}{c|}{\textbf{$Pr(Z >|z|)$}} & \multicolumn{2}{c}{\textbf{[95\% Conf.Interval]}} \\ \toprule

\multirow{10}{*}{\textbf{GTS}} & \multirow{1}{*}{\textbf{$\mu$}} &  \multirow{1}{*}{-0.4854} & \multirow{1}{*}{(1.008)} & \multirow{1}{*}{-0.48} & \multirow{1}{*}{6.3E-01} & \multirow{1}{*}{-2.461} & \multirow{1}{*}{1.491} \\

 & \multirow{1}{*}{\textbf{$\beta_{+}$}} & \multirow{1}{*}{0.3904} & \multirow{1}{*}{(0.164)} & \multirow{1}{*}{2.38} & \multirow{1}{*}{1.7E-02} & \multirow{1}{*}{0.069} & \multirow{1}{*}{0.712} \\

 &\multirow{1}{*}{ \textbf{$\beta_{-}$}} & \multirow{1}{*}{0.4045} & \multirow{1}{*}{(0.210)} & \multirow{1}{*}{1.93} & \multirow{1}{*}{5.4E-02} & \multirow{1}{*}{-0.007} & \multirow{1}{*}{0.816} \\

 & \multirow{1}{*}{\textbf{$\alpha_{+}$}} & \multirow{1}{*}{0.9582} & \multirow{1}{*}{(0.106)} & \multirow{1}{*}{9.01} & \multirow{1}{*}{1.1E-19} & \multirow{1}{*}{0.750} & \multirow{1}{*}{1.167} \\

 & \multirow{1}{*}{\textbf{$\alpha_{-}$}} & \multirow{1}{*}{0.8005} & \multirow{1}{*}{(0.110)} & \multirow{1}{*}{7.25} & \multirow{1}{*}{4.2E-13} & \multirow{1}{*}{0.584} & \multirow{1}{*}{1.017} \\

 &\multirow{1}{*}{\textbf{$\lambda_{+}$}} & \multirow{1}{*}{0.1667} & \multirow{1}{*}{(0.029)} & \multirow{1}{*}{5.72} & \multirow{1}{*}{1.1E-08} & \multirow{1}{*}{0.110} & \multirow{1}{*}{0.224} \\

 & \multirow{1}{*}{\textbf{$\lambda_{-}$}} & \multirow{1}{*}{0.1708} & \multirow{1}{*}{(0.036)} & \multirow{1}{*}{4.71} & \multirow{1}{*}{2.5E-06} & \multirow{1}{*}{0.110} & \multirow{1}{*}{0.242} \\ \cmidrule(l){2-8}

  & \multirow{1}{*}{\textbf{Log(ML)}} &  \multirow{1}{*}{-9552} & \multirow{1}{*}{} & \multirow{1}{*}{} & \multirow{1}{*}{} & \multirow{1}{*}{}  \\
 & \multirow{1}{*}{\textbf{AIC}} &  \multirow{1}{*}{19119} & \multirow{1}{*}{} & \multirow{1}{*}{} & \multirow{1}{*}{} & \multirow{1}{*}{} & \multirow{1}{*}{} \\
 & \multirow{1}{*}{\textbf{BIK}} &  \multirow{1}{*}{19162} & \multirow{1}{*}{} & \multirow{1}{*}{} & \multirow{1}{*}{} & \multirow{1}{*}{} &\multirow{1}{*}{}  \\ \toprule

 \multirow{5}{*}{\textbf{GBM}} & \multirow{1}{*}{\textbf{$\mu$}}&  \multirow{1}{*}{0.267284} & \multirow{1}{*}{(0.091)} & \multirow{1}{*}{2.93} & \multirow{1}{*}{3.4E-03} & \multirow{1}{*}{0.088} & \multirow{1}{*}{0.446} \\
 & \multirow{1}{*}{\textbf{$\sigma$}} & \multirow{1}{*}{5.205539} & \multirow{1}{*}{(0.672)} & \multirow{1}{*}{7.74} & \multirow{1}{*}{1.0E-14} & \multirow{1}{*}{3.887} & \multirow{1}{*}{6.524} \\ \cmidrule(l){2-8}
 & \multirow{1}{*}{\textbf{Log(ML)}}  & \multirow{1}{*}{-9960} & \multirow{1}{*}{} & \multirow{1}{*}{} & \multirow{1}{*}{} & \multirow{1}{*}{} & \multirow{1}{*}{}  \\
 & \multirow{1}{*}{\textbf{AIC}} & \multirow{1}{*}{19925} & \multirow{1}{*}{} & \multirow{1}{*}{} & \multirow{1}{*}{} & \multirow{1}{*}{} & \multirow{1}{*}{}  \\
  & \multirow{1}{*}{\textbf{BIK}} & \multirow{1}{*}{19933} & \multirow{1}{*}{} & \multirow{1}{*}{} & \multirow{1}{*}{} & \multirow{1}{*}{} & \multirow{1}{*}{}  \\ \bottomrule
\end{tabular}%
\vspace{-0.3cm}
\end{table}

\noindent
We consider the following constraints $\alpha=\alpha_{+} =\alpha_{-}$ and $\beta=\beta_{+} = \beta_{-}$, which is the Carr–Geman–Madan–Yor (CGMY) distribution, also called Classical Tempered Stable Distribution. The CGMY distribution was fitted as well, and the estimation results are presented in Appendix \ref{eq:an4}. As shown in Table \ref{tabd1}, all the parameters are statistically significant at 5\%, and have the expected positive sign. However, the likelihood ratio test in Table \ref{tab6} shows, with a high p-value ($35.3\%$), That the GTS distribution is not significantly different from the  CGMY distribution, and the null hypothesis can not be rejected. Refer to \cite{carr2003stochastic, rachev2011financial} for more details on CGMY distribution\\

\noindent
Table \ref{tab1} and Table \ref{tab2} summarized the last row of Table \ref{taba1} and Table \ref{taba2} respectively in appendix \ref{eq:an1}, which describe the convergence process of the GTS parameter for Bitcoin and Ethereum data. The convergence process was obtained using the Newton-Raphson iteration algorithm (\ref{eq:l34}). Each row has eleven columns made of the iteration number, the seven parameters \textbf{$\mu$}, \textbf{$\beta_{+}$}, \textbf{$\beta_{-}$}, \textbf{$\alpha_{+}$}, \textbf{$\alpha_{-}$}, \textbf{$\lambda_{+}$}, \textbf{$\lambda_{-}$}, and three statistical indicators: the log-likelihood (\textbf{$Log(ML)$}), the norm of the partial derivatives (\textbf{$||\frac{dLog(ML)}{dV}||$}), and the maximum value of the eigenvalues (\textbf{$Max Eigen Value$}). The statistical indicators aim at checking if the two necessary and sufficient conditions described in (\ref{eq:l33}) are all met. \textbf{$Log(ML)$} displays the value of the Naperian logarithm of the likelihood function $L(x, V)$, as described in (\ref{eq:l31}); \textbf{$||\frac{dLog(ML)}{dV}||$} displays the value of the norm of the first derivatives (\textbf{$\frac{dl(x, V)}{dV_j}$}) described in (\ref{eq:l32}); and \textbf{$Max Eigen Value$} displays the maximum value of the seven eigenvalues generated by the Hessian matrix (\textbf{$\frac{d^{2}l(x, V)}{dV_{k}dV_{j}}$}), as described in (\ref{eq:l32}).\\

\noindent
Similarly, Table \ref{tabc2} and Table \ref{tabd2} describe the Convergence process of the Kobol distribution parameter for Bitcoin returns and the CGMY distribution parameter for Ethereum returns.\\

\noindent
GTS parameter estimation in Table \ref{tab1} and Table \ref{tab2} are used to evaluate the impact of each parameter on the GTS probability density function. As shown in Fig \ref{fig111} and  Fig \ref{fig22bis}, the effect of the GTS parameters on the probability density function has the same patterns on Bitcoin and Ethereum returns. However, the magnitudes are different. As shown in Fig \ref{fig11a} and Fig \ref{fig11b}, \textbf{$\beta_{-}$} (\textbf{$\alpha_{-}$} ) has a higher effect on the probability density function (pdf) than \textbf{$\beta_{+}$} (\textbf{$\alpha_{+}$}). However, \textbf{$\lambda_{-}$} and \textbf{$\lambda_{+}$}  in both graphs seem symmetric and have the same impact.

\begin{figure}[ht]
\vspace{-0.5cm}
    \centering
  \begin{subfigure}[b]{0.3\linewidth}
    \includegraphics[width=\linewidth]{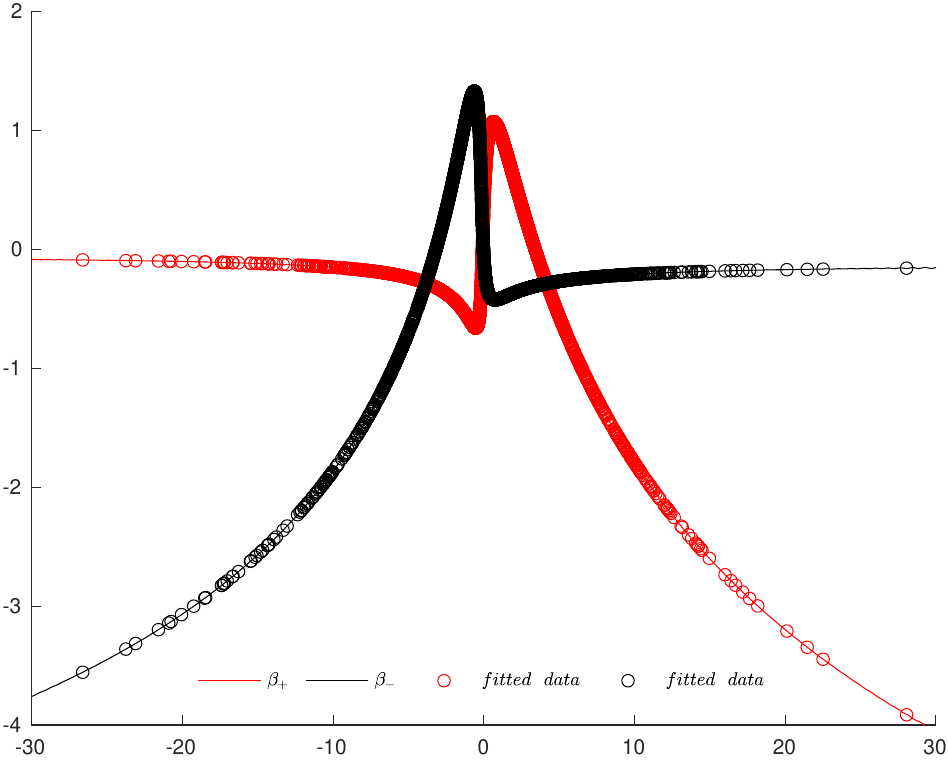}
\vspace{-0.5cm}
     \caption{$V_{j}=\beta_{+}, V_{j}=\beta_{-}$}
         \label{fig11a}
  \end{subfigure}
  \begin{subfigure}[b]{0.3\linewidth}
    \includegraphics[width=\linewidth]{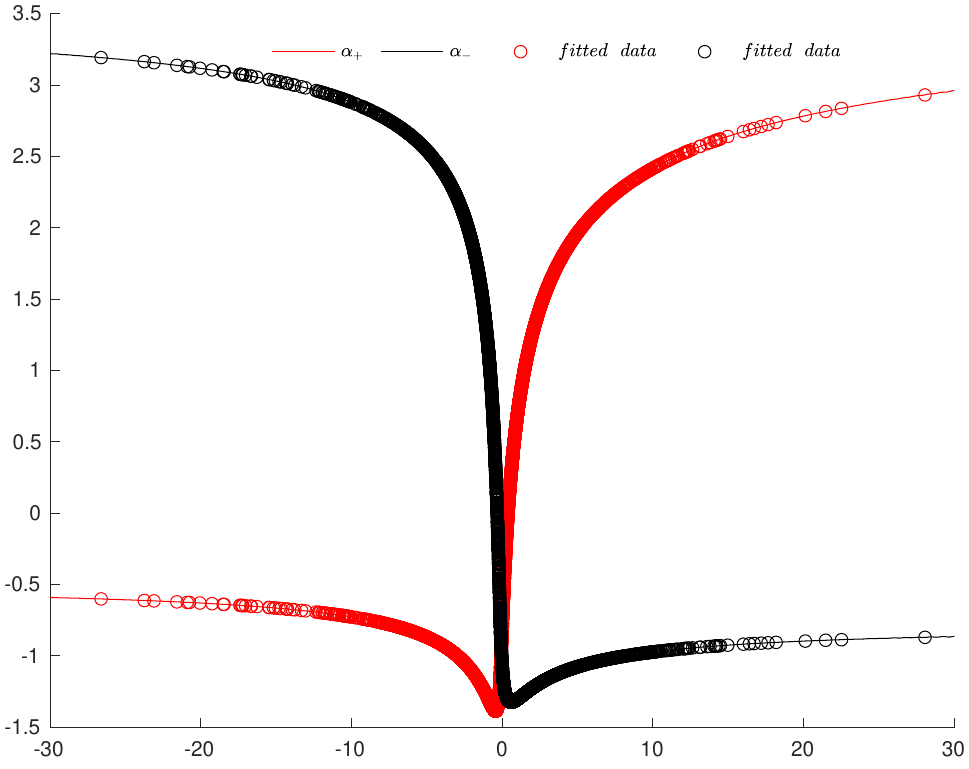}
\vspace{-0.5cm}
     \caption{$V_{j}=\alpha_{+}, V_{j}=\alpha_{-}$}
         \label{fig11b}
          \end{subfigure}
  \begin{subfigure}[b]{0.3\linewidth}
    \includegraphics[width=\linewidth]{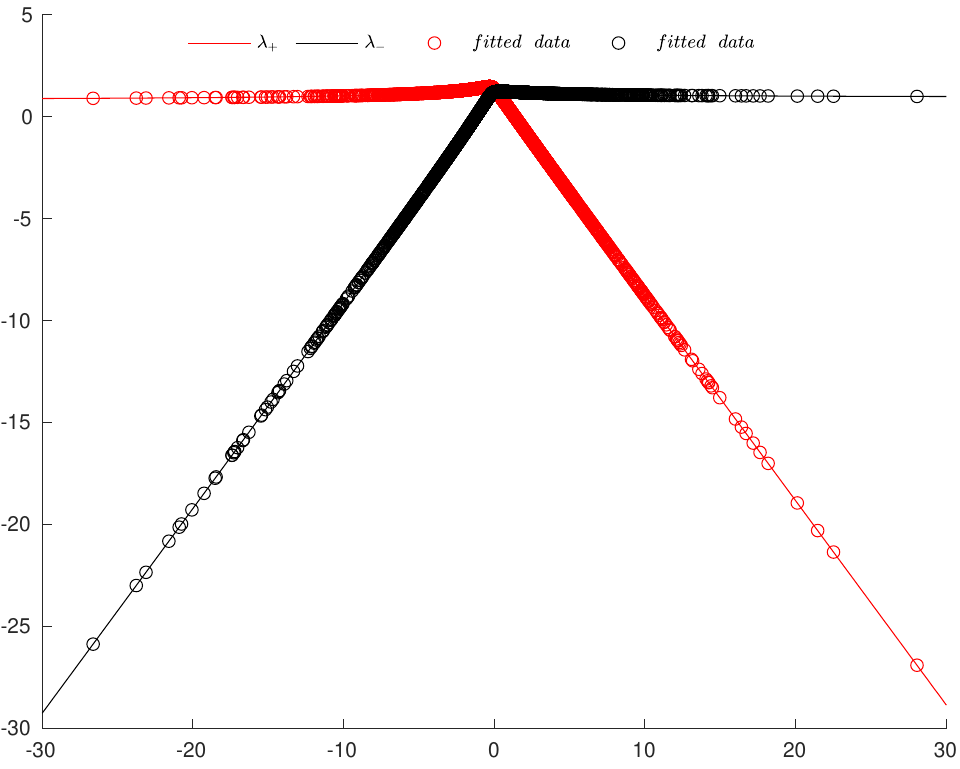}
\vspace{-0.5cm}
     \caption{$V_{j}=\lambda_{+}, V_{j}=\lambda_{-}$}
         \label{fig11c}
          \end{subfigure}
\vspace{-0.2cm}
  \caption{$\frac{\frac{df(x, V)}{dV_j}}{f(x, V)}$: Effect of parameters on the GTS probability density (Bitcoin Returns)}
  \label{fig111}
\vspace{-0.3cm}
\end{figure}

\begin{figure}[ht]
\vspace{-0.5cm}
    \centering
  \begin{subfigure}[b]{0.3\linewidth}
    \includegraphics[width=\linewidth]{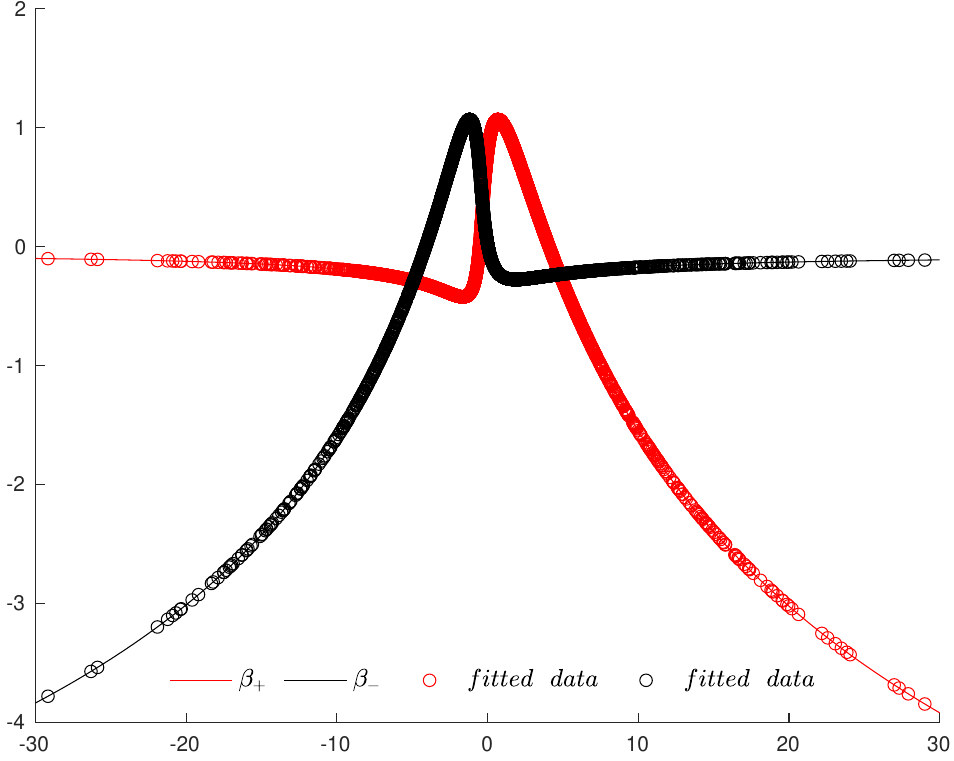}
\vspace{-0.5cm}
     \caption{$V_{j}=\beta_{+}, V_{j}=\beta_{-}$}
         \label{fig22a}
  \end{subfigure}
  \begin{subfigure}[b]{0.3\linewidth}
    \includegraphics[width=\linewidth]{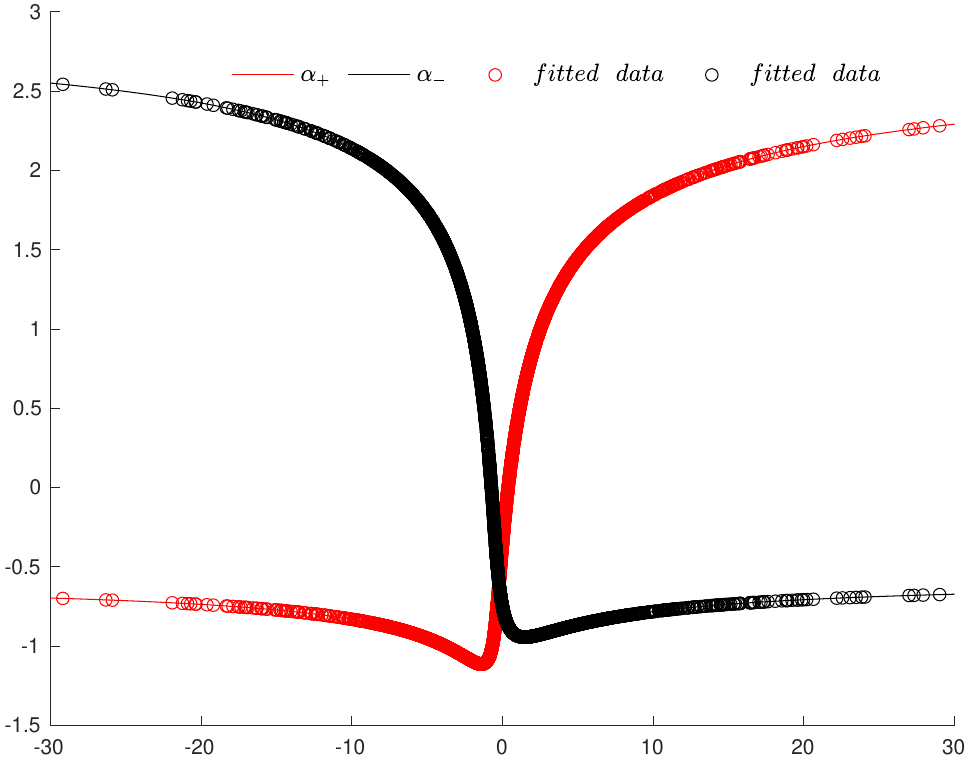}
\vspace{-0.5cm}
     \caption{$V_{j}=\alpha_{+}, V_{j}=\alpha_{-}$}
         \label{fig22b}
          \end{subfigure}
  \begin{subfigure}[b]{0.3\linewidth}
    \includegraphics[width=\linewidth]{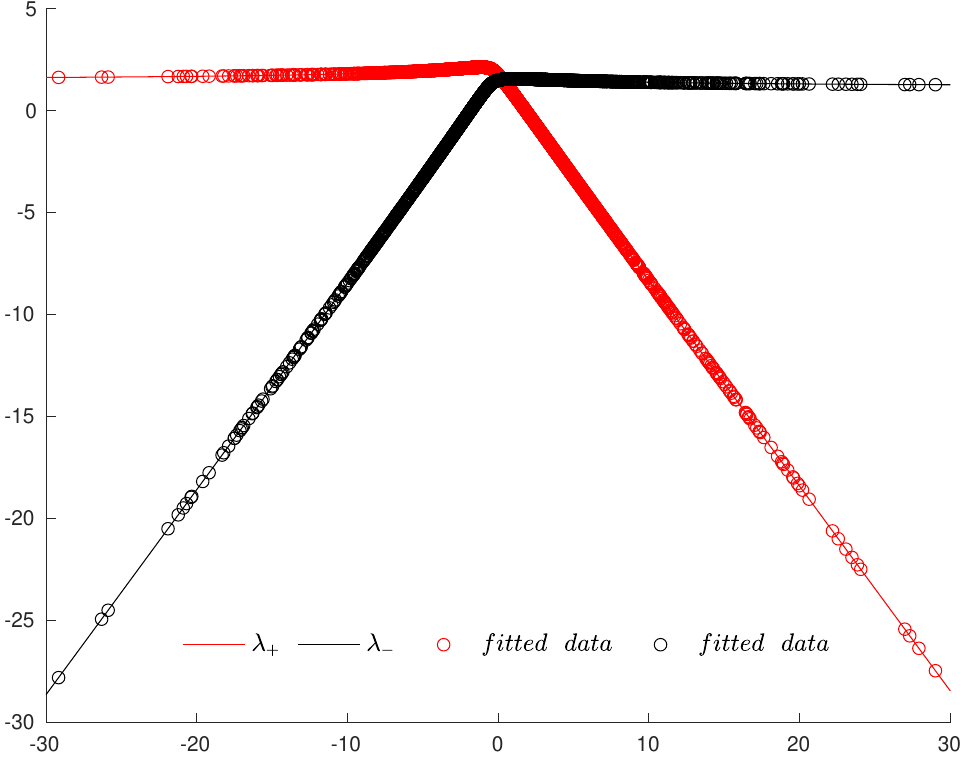}
\vspace{-0.5cm}
     \caption{$V_{j}=\lambda_{+}, V_{j}=\lambda_{-}$}
         \label{fig22c}
          \end{subfigure}
\vspace{-0.2cm}
  \caption{$\frac{\frac{df(x, V)}{dV_j}}{f(x, V)}$: Effect of parameters on the GTS probability density (Ethereum Returns)}
  \label{fig22bis}
\vspace{-0.6cm}
\end{figure}


\subsection{Evaluation of the Method of Moments}
\noindent
The Method of Moments estimates the parameters of the GTS distribution by equating empirical moments and the theoretical moments of the GTS distribution. We empirically estimate the $k^{th}$ moments ($m_{k} = E(x^{k})$), based on sample data $x=\left(x_{j}\right)_{1 \leq j \leq m} $ as follows:
 \begin{equation}
 \begin{aligned}
\hat{m}_{k} = \frac{1}{M}\sum_{j=1}^{m} x^{k}_{j} \quad  \text{for}  \quad k=1,\dotsc,7. \label{eq:l37ba}
 \end{aligned}
 \end{equation}
 On the other side, the Cumulants ($\kappa_{k}$) in theorem \ref{lem7}  can be related to the moment of the GTS distribution by the following relationship \cite{smith1995, Poloskov_2021, rota2000283}:
  \begin{equation}
 \begin{aligned}
m_{k} = E(x^{k}) = \sum_{j=1}^{k-1} \binom{k-1}{j-1} \kappa_{j} m_{k-j} + \kappa_{k} \quad  \text{for}  \quad k=1,\dotsc,7. \label{eq:l37bb}
 \end{aligned}
 \end{equation}
 The method of moments estimator for $V=(\mu, \beta_{+}, \beta_{-}, \alpha_{+},\alpha_{-}, \lambda_{+}, \lambda_{-})$ is defined as the solution to the following system of equations:
  \begin{equation}
 \begin{aligned}
\hat{m}_{k} = m_{k} \quad  \text{for}  \quad k=1,\dotsc,7. \label{eq:l37bc}
 \end{aligned}
 \end{equation}
 \noindent
The system of equations (\ref{eq:l37bc}) is often not analytically solvable. For the conditions of existence and uniqueness of the solution, refer to \cite{kuchler2013tempered}.\\

\noindent
Maximum likelihood GTS parameter estimation in Table \ref{tab1} and Table \ref{tab2} are used to evaluate the system of equations in (\ref{eq:l37bc}).  As shown in Table \ref{tab3}, the solution of the maximum Likelihood method satisfies at a certain extent the equations for the first four moments: $\hat{m}_{1}$, $\hat{m}_{2}$, $\hat{m}_{3}$, $\hat{m}_{4}$ in the system (\ref{eq:l37bc}). The $7^{th}$ moment equation has the highest relative error: 89.9\% for the Bitcoin BTC and 68.3\% for the Ethereum. Therefore, the maximum likelihood GTS parameter estimation is not the same as the GTS parameter estimation from the method of moments.\\

\noindent
In addition to the method of moments estimations, the lower relative errors in Table \ref{tab3} show that empirical and theoretical standard deviation ($\sigma$), skewness, and kurtosis seem to be consistent for Bitcoin and Ethereum. The empirical \& theoretical statistics show that the average Ethereum daily return is greater and more volatile than the Bitcoin daily returns. Both assets are thicker than the Normal distribution. However, the daily return of Bitcoin is skewed to the left, whereas the daily return of Ethereum is skewed to the right.
\begin{table}[ht]
\vspace{-0.3cm}
\centering
 \caption{Evaluation of the Method of Moments}
\label{tab3}
\setlength{\tabcolsep}{0.3mm}
\begin{tabular}{@{}lcclccl@{}}
\toprule
 & \multicolumn{3}{c}{\textbf{Bitcoin BTC}} & \multicolumn{3}{c}{\textbf{Ethereum}} \\ \midrule
\multirow{1}{*}{} & \multirow{1}{*}{\textbf{Empirical(1)}} & \multirow{1}{*}{\textbf{Theoretical(2)}} & \multirow{1}{*}{\textbf{$\frac{(1)-(2)}{2}$}} & \multirow{1}{*}{\textbf{Empirical(1)}} & \multirow{1}{*}{\textbf{Theoretical(2)}} & \multirow{1}{*}{\textbf{$\frac{(1)-(2)}{2}$}} \\ \midrule
\multirow{1}{*}{\textbf{Sample size}} & \multirow{1}{*}{4083} & \multirow{1}{*}{} & \multirow{1}{*}{} & \multirow{1}{*}{3246} & \multirow{1}{*}{} & \multirow{1}{*}{} \\
\multirow{1}{*}{\textbf{$\hat{m}_{1}$}} & \multirow{1}{*}{0.152} & \multirow{1}{*}{0.152} & \multirow{1}{*}{0.0\%} & \multirow{1}{*}{0.267} & \multirow{1}{*}{0.267} & \multirow{1}{*}{0.0\%} \\
\multirow{1}{*}{\textbf{$\hat{m}_{2}$}} & \multirow{1}{*}{14.960} & \multirow{1}{*}{15.020} & \multirow{1}{*}{0.4\%} & \multirow{1}{*}{27.161} & \multirow{1}{*}{27.388} & \multirow{1}{*}{0.8\%} \\
\multirow{1}{*}{\textbf{$\hat{m}_{3}$}} & \multirow{1}{*}{-11.320} & \multirow{1}{*}{-15.640} & \multirow{1}{*}{27.6\%} & \multirow{1}{*}{55.363} & \multirow{1}{*}{57.867} & \multirow{1}{*}{4.3\%} \\
\multirow{1}{*}{\textbf{$\hat{m}_{4}$}} & \multirow{1}{*}{2033} & \multirow{1}{*}{2256} & \multirow{1}{*}{9.8\%} & \multirow{1}{*}{5267} & \multirow{1}{*}{6307} & \multirow{1}{*}{16.5\%} \\
\multirow{1}{*}{\textbf{$\hat{m}_{5}$}} & \multirow{1}{*}{-5823} & \multirow{1}{*}{-15480} & \multirow{1}{*}{62.3\%} & \multirow{1}{*}{22368} & \multirow{1}{*}{32518} & \multirow{1}{*}{31.2\%} \\
\multirow{1}{*}{\textbf{$\hat{m}_{6}$}} & \multirow{1}{*}{670695} & \multirow{1}{*}{1123215} & \multirow{1}{*}{40.2\%} & \multirow{1}{*}{2114788} & \multirow{1}{*}{4361562} & \multirow{1}{*}{51.5\%} \\
\multirow{1}{*}{\textbf{$\hat{m}_{7}$}} & \multirow{1}{*}{-1997196} & \multirow{1}{*}{-19777988} & \multirow{1}{*}{89.9\%} & \multirow{1}{*}{12411809} & \multirow{1}{*}{39253001} & \multirow{1}{*}{68.3\%} \\
\multirow{1}{*}{\textbf{Standard deviation \footnotemark[1]}} & \multirow{1}{*}{3.865} & \multirow{1}{*}{3.873} & \multirow{1}{*}{0.2\%} & \multirow{1}{*}{5.206} & \multirow{1}{*}{5.226} & \multirow{1}{*}{0.4\%} \\
\multirow{1}{*}{\textbf{Skewness \footnotemark[2]}} & \multirow{1}{*}{-0.314} & \multirow{1}{*}{-0.387} & \multirow{1}{*}{18.8\%} & \multirow{1}{*}{0.238} & \multirow{1}{*}{0.252} & \multirow{1}{*}{5.2\%} \\
\multirow{1}{*}{\textbf{Kurtosis \footnotemark[3]}} & \multirow{1}{*}{9.154} & \multirow{1}{*}{10.082} & \multirow{1}{*}{9.2\%} & \multirow{1}{*}{7.112} & \multirow{1}{*}{8.385} & \multirow{1}{*}{15.2\%} \\
\multirow{1}{*}{\textbf{Max value}} & \multirow{1}{*}{28.052} & \multirow{1}{*}{} & \multirow{1}{*}{} & \multirow{1}{*}{29.013} & \multirow{1}{*}{} & \multirow{1}{*}{} \\
\multirow{1}{*}{\textbf{Min Value}} & \multirow{1}{*}{-26.620}& \multirow{1}{*}{} & \multirow{1}{*}{} & \multirow{1}{*}{-29.174} & \multirow{1}{*}{} & \multirow{1}{*}{} \\ \bottomrule
\end{tabular}%
\vspace{-0.3cm}
\end{table}
\footnotetext[1]{$\sigma=\sqrt{\kappa_{2}}$}
\footnotetext[2]{Skewness is estimated as $\frac{\kappa_{3}}{\kappa_{2}^{3/2}}$}
\footnotetext[3]{Kurtosis is estimated as $3+\frac{\kappa_{4}}{\kappa_{2}^{2}}$;  $\kappa_{1}$, $\kappa_{2}$ and $\kappa_{2}$ are defined in (\ref{eq:l12})}

\clearpage
\newpage
\section{Fitting Tempered Stable Distribution to Traditional Indices: S\&P 500 and SPY EFT\label{estim1}}
\subsection{Data Summaries}
\noindent
The Standard $\&$ Poor’s 500 Composite Stock Price Index, also known as the S$\&$P 500, is a stock index that tracks the share prices of 500 of the largest public companies with stocks listed on the New York Stock Exchange (NYSE) and the Nasdaq in the United States. It was introduced in 1957 and often treated as a proxy for describing the overall health of the stock market or the United States (US) economy. The SPDR S\&P 500   ETF (SPY), also known as the SPY ETF,  is an Exchange-Traded Fund (ETF)that tracks the performance of the S\&P 500. SPY ETF provides a mutual fund's diversification, the stock's flexibility, and lower trading fees. The data were extracted from Yahoo Finance. The historical prices span from 04 January 2010 to 22 July 2024 and were adjusted for splits and dividends.\\

\noindent
The daily price dynamics are provided in Fig \ref{fig3}. Prices have an increasing trend, even after being temporally disrupted in the first quarter of 2020 by the coronavirus pandemic. The S\&P 500 is priced in thousands of US dollars,  whereas the SPY ETF is in hundreds of US dollars. The SPY ETF is cheaper and provides all the attributes of the S\&P 500 index, as shown in  Fig \ref{fig31} and  Fig \ref{fig32}.
\begin{figure}[ht]
    \centering
\hspace{-1cm}
  \begin{subfigure}[b]{0.45\linewidth}
    \includegraphics[width=\linewidth]{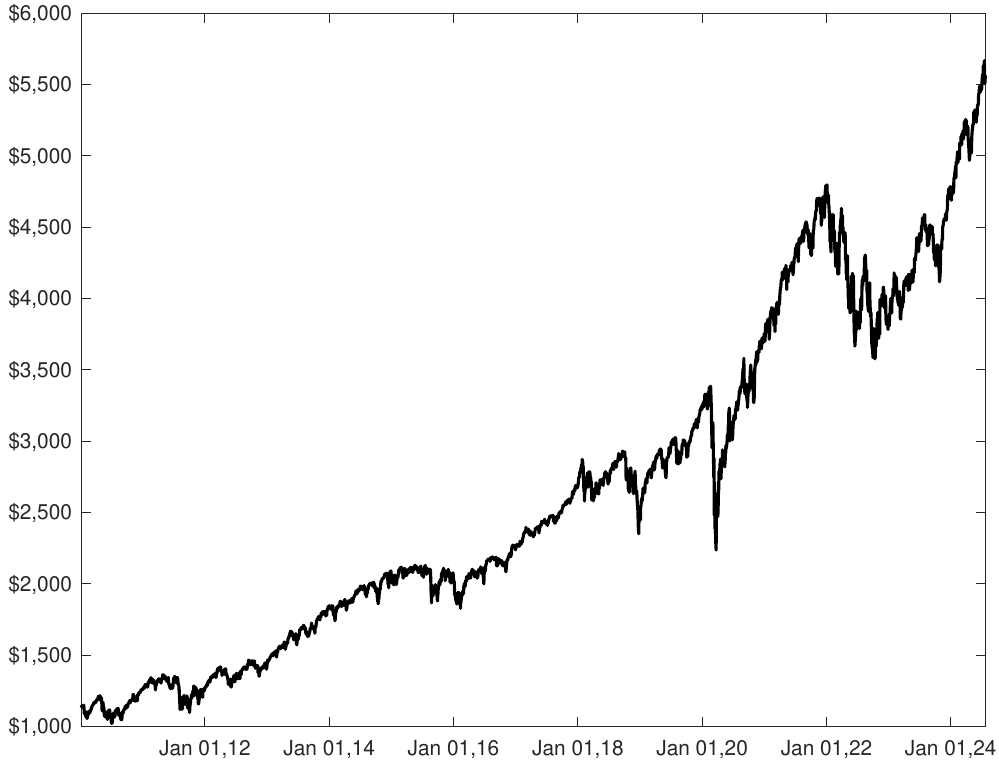}
\vspace{-0.6cm}
     \caption{S\&P 500 Daily Price}
         \label{fig31}
  \end{subfigure}
  \begin{subfigure}[b]{0.45\linewidth}
    \includegraphics[width=\linewidth]{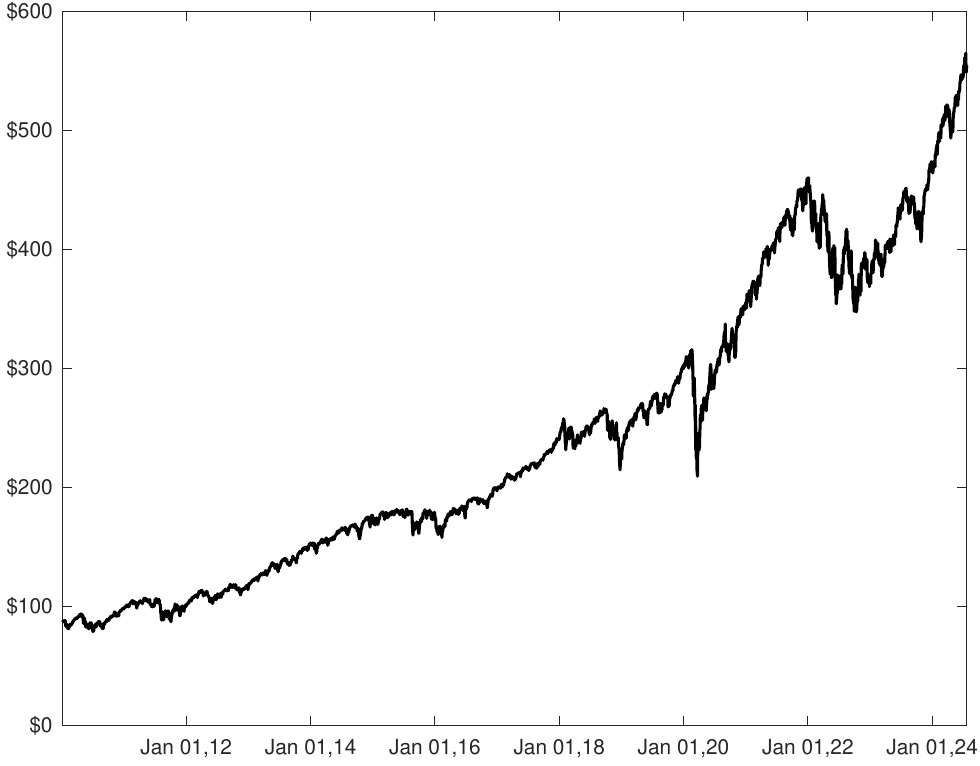}
\vspace{-0.6cm}
     \caption{SPY EFT Daily Price}
         \label{fig32}
          \end{subfigure}
\vspace{-0.3cm}
  \caption{Daily Price}
  \label{fig3}
\vspace{-0.6cm}
\end{figure}

\noindent
Let the number of observations $m$, and the daily observed price $S_{j}$ on day $t_{j}$ with $j=1,\dots,m$; $t_{1}$ is the first observation date (January 04, 2010) and $t_{m}$ is the last observation date (July 22, 2024). The daily return, $y_{j}$, is computed as in (\ref{eq:l42}):\\
\begin{align}
y_{j}=\log(S_{j}/S_{j-1}) \hspace{10 mm}  \hbox{ $j=2,\dots,m$}.\label{eq:l42}
 \end{align}
\noindent  
SPY ETF aims to mirror the performance of the S\&P 500. Fig \ref{fig41} and Fig \ref{fig42} look similar, which is consistent with the goal of SPY ETF. As shown in  Fig \ref{fig41} and Fig \ref{fig42}, the daily return reaches the lowest level ($-12.7\%$ for S\&P 500 and $-11.5\%$ for SPY ETF) in the first quarter of $2020$ amid the coronavirus pandemic and massive disruptions in the global economy. Nine values were identified as outliers and removed from the data set to avoid a negative impact on the GTS model estimation and the empirical statistics. 
\begin{figure}[ht]
\vspace{-0.4cm}
    \centering
  \begin{subfigure}[b]{0.42\linewidth}
    \includegraphics[width=\linewidth]{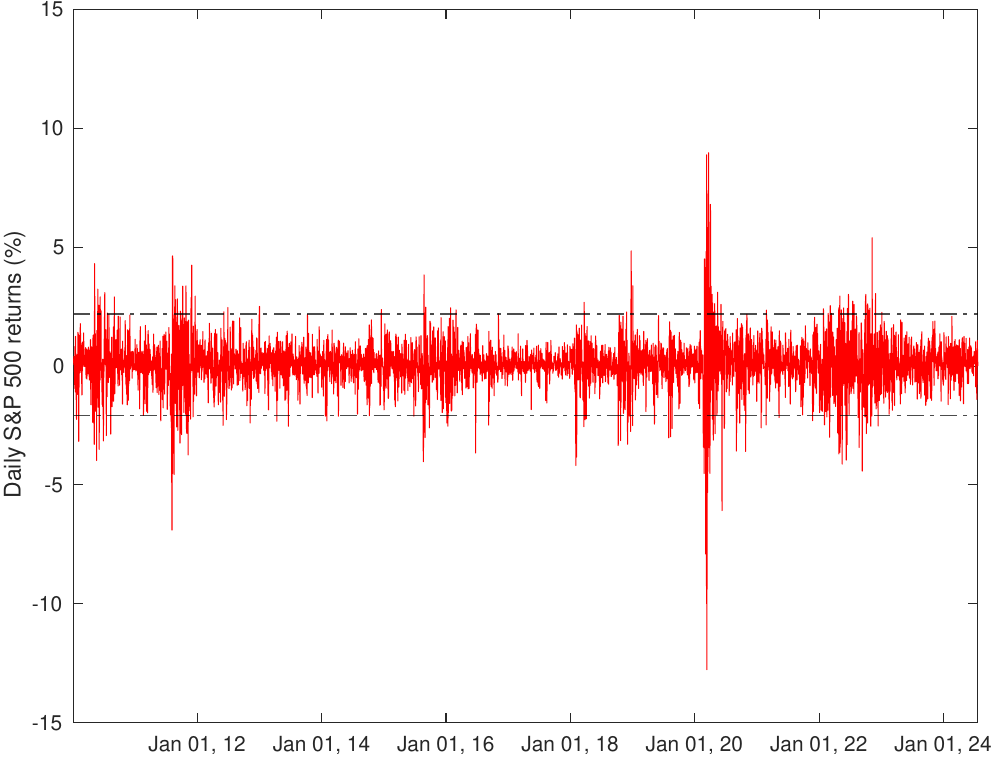}
\vspace{-0.5cm}
     \caption{Daily S\&P500 return}
         \label{fig41}
  \end{subfigure}
  \begin{subfigure}[b]{0.42\linewidth}
    \includegraphics[width=\linewidth]{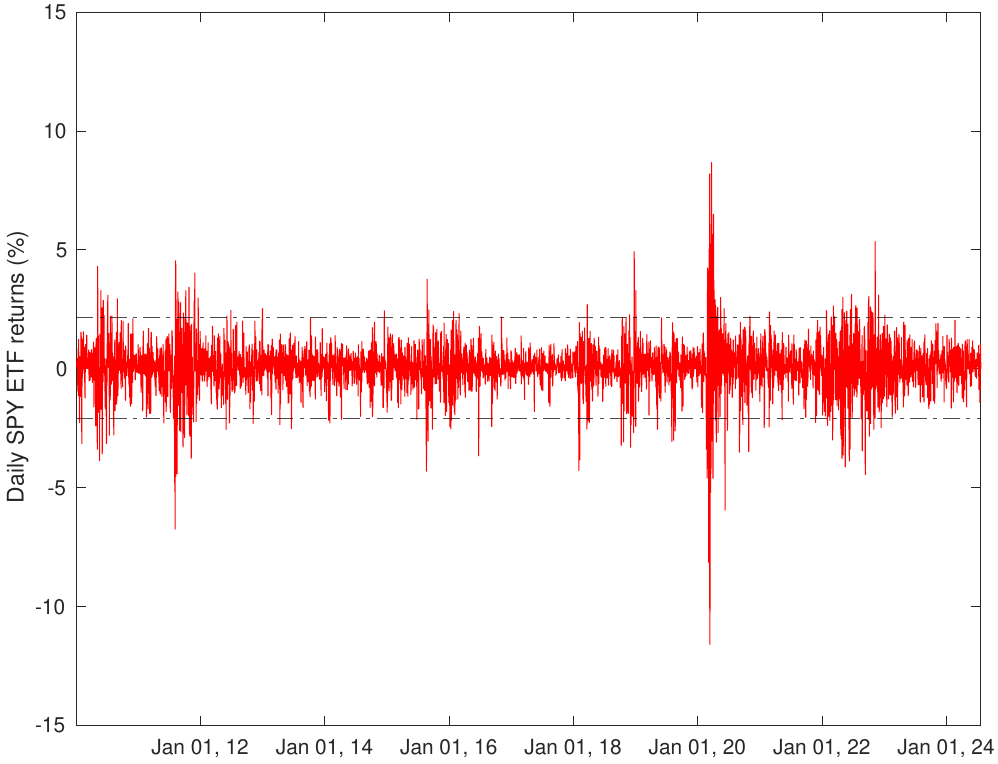}
\vspace{-0.5cm}
     \caption{ Daily SPY ETF return }
         \label{fig42}
          \end{subfigure}
\vspace{-0.4cm}
  \caption{Daily Return}
  \label{fig4}
\vspace{-0.8cm}
\end{figure}

\subsection{Multidimensional Estimation Results for Traditional Indices}
 The estimation results are provided in Table \ref{tab4} for  S\&P 500  return data and Table \ref{tab5} for SPY EFT return data.  As previously, the log-likelihood, AIC, and BIK statistics suggest that the GTS distribution with seven parameters performs better than the two-parameter Normal distribution (GBM). \\
  \noindent  
As shown in both Tables \ref{tab4} \& \ref{tab5}, the ML estimate of $\mu$ is negative, while others are positive, as expected in the literature. The asymptotic standard error for $\mu$, $\beta_{+}$  and $\beta_{-}$ are pretty large and it results that  $\mu$, $\beta_{+}$  and $\beta_{-}$ are not significantly different from zero.
 \begin{table}[ht]
 \vspace{-0.6cm}
\centering
\caption{Maximum Likelihood  GTS Parameter Estimation for S\&P 500 Index}
\label{tab4}
\begin{tabular}{@{} c|c|c|ccccc @{}}
\toprule
\multicolumn{1}{c|}{\textbf{Model}} & \multicolumn{1}{c|}{\textbf{Param}} & \multicolumn{1}{c|}{\textbf{Estimate}} & \multicolumn{1}{c|}{\textbf{Std Err}} & \multicolumn{1}{c|}{\textbf{z}} & \multicolumn{1}{c|}{\textbf{$Pr(Z >|z|)$}} & \multicolumn{2}{c}{\textbf{[95\% Conf.Interval]}}  \\ \toprule

\multirow{10}{*}{\textbf{GTS}} & \multirow{1}{*}{\textbf{$\mu$}} & \multirow{1}{*}{-0.249408} & \multirow{1}{*}{(0.208)} & \multirow{1}{*}{-1.20} & \multirow{1}{*}{2.3E-01} & \multirow{1}{*}{-0.658} & \multirow{1}{*}{0.159} \\

 & \multirow{1}{*}{\textbf{$\beta_{+}$}} & \multirow{1}{*}{0.328624} & \multirow{1}{*}{(0.308)} & \multirow{1}{*}{1.07} & \multirow{1}{*}{2.9E-01} & \multirow{1}{*}{-0.275} & \multirow{1}{*}{0.932}  \\

 &\multirow{1}{*}{ \textbf{$\beta_{-}$}} & \multirow{1}{*}{0.088640} & \multirow{1}{*}{(0.176)} & \multirow{1}{*}{0.50} & \multirow{1}{*}{6.1E-01} & \multirow{1}{*}{-0.256} & \multirow{1}{*}{0.433}  \\

 & \multirow{1}{*}{\textbf{$\alpha_{+}$}} & \multirow{1}{*}{0.792426} & \multirow{1}{*}{(0.350)} & \multirow{1}{*}{2.26} & \multirow{1}{*}{2.4E-02} & \multirow{1}{*}{0.106} & \multirow{1}{*}{1.479} \\

 & \multirow{1}{*}{\textbf{$\alpha_{-}$}} & \multirow{1}{*}{0.542250} & \multirow{1}{*}{(0.107)} & \multirow{1}{*}{5.09} & \multirow{1}{*}{3.6E-07} & \multirow{1}{*}{0.333} & \multirow{1}{*}{0.751}  \\

 &\multirow{1}{*}{\textbf{$\lambda_{+}$}} & \multirow{1}{*}{1.279743} & \multirow{1}{*}{(0.348)} & \multirow{1}{*}{3.68} & \multirow{1}{*}{2.4E-04} & \multirow{1}{*}{0.597} & \multirow{1}{*}{1.962}  \\

 & \multirow{1}{*}{\textbf{$\lambda_{-}$}} & \multirow{1}{*}{0.937133} & \multirow{1}{*}{(0.144)} & \multirow{1}{*}{6.50} & \multirow{1}{*}{8.0E-11} & \multirow{1}{*}{0.655} & \multirow{1}{*}{1.220} \\ \cmidrule(l){2-8}

  & \multirow{1}{*}{\textbf{Log(ML)}} & \multirow{1}{*}{-4920} & \multirow{1}{*}{} & \multirow{1}{*}{} & \multirow{1}{*}{} & \multirow{1}{*}{} & \multirow{1}{*}{}   \\
 & \multirow{1}{*}{\textbf{AIC}} & \multirow{1}{*}{9851} & \multirow{1}{*}{} & \multirow{1}{*}{} & \multirow{1}{*}{} & \multirow{1}{*}{} & \multirow{1}{*}{}  \\
 & \multirow{1}{*}{\textbf{BIK}} & \multirow{1}{*}{9898} & \multirow{1}{*}{} & \multirow{1}{*}{} & \multirow{1}{*}{} & \multirow{1}{*}{} & \multirow{1}{*}{} \\ \toprule

 \multirow{5}{*}{\textbf{GBM}} & \multirow{1}{*}{\textbf{$\mu$}} & \multirow{1}{*}{0.044875} & \multirow{1}{*}{(0.018)} & \multirow{1}{*}{2.51} & \multirow{1}{*}{1.2E-02} & \multirow{1}{*}{0.010} & \multirow{1}{*}{0.080}  \\
 & \multirow{1}{*}{\textbf{$\sigma$}} & \multirow{1}{*}{1.081676} & \multirow{1}{*}{(0.027)} & \multirow{1}{*}{39.53} & \multirow{1}{*}{0.000} & \multirow{1}{*}{1.028} & \multirow{1}{*}{1.135}  \\ \cmidrule(l){2-8}
 & \multirow{1}{*}{\textbf{Log(ML)}} & \multirow{1}{*}{-5330} & \multirow{1}{*}{} & \multirow{1}{*}{} & \multirow{1}{*}{} & \multirow{1}{*}{} &\multirow{1}{*}{}  \\
 & \multirow{1}{*}{\textbf{AIC}} & \multirow{1}{*}{10665} & \multirow{1}{*}{} & \multirow{1}{*}{} & \multirow{1}{*}{} & \multirow{1}{*}{} & \multirow{1}{*}{}   \\
  & \multirow{1}{*}{\textbf{BIK}} & \multirow{1}{*}{10677} & \multirow{1}{*}{} & \multirow{1}{*}{} & \multirow{1}{*}{} & \multirow{1}{*}{} & \multirow{1}{*}{}  \\ \bottomrule
\end{tabular}%
\vspace{-0.5cm}
\end{table}
\newpage
 \noindent  
However, other parameters have larger t-statistics ($|z|>2$) and are statistically significant at 5\%. Except for the index of stability parameters ($\beta_{+}$, $\beta_{-}$),  the estimation results for S\&P 500 and SPY ETF indexes show that the difference in skewness parameters ($\lambda_{+}$, $\lambda_{-}$) and intensity parameters ($\alpha_{+}$, $\alpha_{-}$) are positive but are not statistically significant.
\begin{table}[ht]
\centering
\caption{Maximum Likelihood  GTS Parameter Estimation for SPY EFT Data}
\label{tab5}
\begin{tabular}{@{} c|c|c|ccccc @{}}
\toprule
\multicolumn{1}{c|}{\textbf{Model}} & \multicolumn{1}{c|}{\textbf{Param}} & \multicolumn{1}{c|}{\textbf{Estimate}} & \multicolumn{1}{c|}{\textbf{Std Err}} & \multicolumn{1}{c|}{\textbf{z}} & \multicolumn{1}{c|}{\textbf{$Pr(Z >|z|)$}} & \multicolumn{2}{c}{\textbf{[95\% Conf.Interval]}} \\ \toprule

\multirow{10}{*}{\textbf{GTS}} & \multirow{1}{*}{\textbf{$\mu$}} &  \multirow{1}{*}{-0.260643} & \multirow{1}{*}{(0.135)} & \multirow{1}{*}{-1.94} & \multirow{1}{*}{5.3E-02} & \multirow{1}{*}{-0.524} & \multirow{1}{*}{0.003} \\

 & \multirow{1}{*}{\textbf{$\beta_{+}$}} & \multirow{1}{*}{0.340880} & \multirow{1}{*}{(0.189)} & \multirow{1}{*}{1.80} & \multirow{1}{*}{7.1E-02} & \multirow{1}{*}{-0.030} & \multirow{1}{*}{0.711} \\

 &\multirow{1}{*}{ \textbf{$\beta_{-}$}}  & \multirow{1}{*}{0.022212} & \multirow{1}{*}{(0.212)} & \multirow{1}{*}{0.10} & \multirow{1}{*}{9.2E-01} & \multirow{1}{*}{-0.393} & \multirow{1}{*}{0.437} \\

 & \multirow{1}{*}{\textbf{$\alpha_{+}$}}  & \multirow{1}{*}{0.787757} & \multirow{1}{*}{(0.225)} & \multirow{1}{*}{3.50} & \multirow{1}{*}{4.6E-04} & \multirow{1}{*}{0.347} & \multirow{1}{*}{1.229} \\

 & \multirow{1}{*}{\textbf{$\alpha_{-}$}} & \multirow{1}{*}{0.597110} & \multirow{1}{*}{(0.141)} & \multirow{1}{*}{4.22} & \multirow{1}{*}{2.4E-05} & \multirow{1}{*}{0.320} & \multirow{1}{*}{0.874} \\

 &\multirow{1}{*}{\textbf{$\lambda_{+}$}} & \multirow{1}{*}{1.288555} & \multirow{1}{*}{(0.226)} & \multirow{1}{*}{5.70} & \multirow{1}{*}{1.2E-08} & \multirow{1}{*}{0.846} & \multirow{1}{*}{1.731} \\

 & \multirow{1}{*}{\textbf{$\lambda_{-}$}}  & \multirow{1}{*}{1.014353} & \multirow{1}{*}{(0.177)} & \multirow{1}{*}{5.74} & \multirow{1}{*}{9.4E-09} & \multirow{1}{*}{0.668} & \multirow{1}{*}{1.361} \\ \cmidrule(l){2-8}

  & \multirow{1}{*}{\textbf{Log(ML)}} &  \multirow{1}{*}{-4893} & \multirow{1}{*}{} & \multirow{1}{*}{} & \multirow{1}{*}{} & \multirow{1}{*}{} & \multirow{1}{*}{}  \\
 & \multirow{1}{*}{\textbf{AIC}} & \multirow{1}{*}{9800} & \multirow{1}{*}{} & \multirow{1}{*}{} & \multirow{1}{*}{} & \multirow{1}{*}{} & \multirow{1}{*}{} \\
 & \multirow{1}{*}{\textbf{BIK}} & \multirow{1}{*}{9843} & \multirow{1}{*}{} & \multirow{1}{*}{} & \multirow{1}{*}{} & \multirow{1}{*}{} & \multirow{1}{*}{} \\ \toprule

 \multirow{5}{*}{\textbf{GBM}} & \multirow{1}{*}{\textbf{$\mu$}} & \multirow{1}{*}{0.054344} & \multirow{1}{*}{(0.017)} & \multirow{1}{*}{3.13} & \multirow{1}{*}{1.8E-03} & \multirow{1}{*}{0.020} &  \multirow{1}{*}{0.088} \\
 & \multirow{1}{*}{\textbf{$\sigma$}} & \multirow{1}{*}{1.050217} & \multirow{1}{*}{(0.026)} & \multirow{1}{*}{40.71} & \multirow{1}{*}{0.000} & \multirow{1}{*}{1.000} &  \multirow{1}{*}{1.101} \\ \cmidrule(l){2-8}
 & \multirow{1}{*}{\textbf{Log(ML)}} & \multirow{1}{*}{-54275} & \multirow{1}{*}{} & \multirow{1}{*}{} & \multirow{1}{*}{} & \multirow{1}{*}{} & \multirow{1}{*}{} \\
 & \multirow{1}{*}{\textbf{AIC}} & \multirow{1}{*}{10554} & \multirow{1}{*}{} & \multirow{1}{*}{} & \multirow{1}{*}{} & \multirow{1}{*}{} & \multirow{1}{*}{} \\
  & \multirow{1}{*}{\textbf{BIK}} & \multirow{1}{*}{10566} & \multirow{1}{*}{} & \multirow{1}{*}{} & \multirow{1}{*}{} & \multirow{1}{*}{} & \multirow{1}{*}{}  \\ \bottomrule
\end{tabular}%
\vspace{-0.3cm}
\end{table}

\noindent  
The hypothesis with $\beta_{+}=\beta_{-}=0$  was considered by fitting the S\&P 500 and SPY ETF indexes to the Bilateral Gamma distribution. The estimation results are summarised in  Appendix \ref{eq:an5} \& \ref{eq:an6}. As shown in Table \ref{tabe1} and Table \ref{tabf1}, the skewness parameters ($\lambda_{+}$, $\lambda_{-}$) are positive and statistically significant, and the difference ($\lambda_{+} - \lambda_{-}$) is also positive and statistically significant, which prove that S\&P 500 and SPY ETF returns are skewed to the left. We have the same statistical features for the intensity parameters ($\alpha_{+}$, $\alpha_{-}$), and Both indexes are more likely to produce positive returns than negative returns. Refer to \cite{KUCHLER2008261, nzokem2021fitting} for more details on Bilateral Gamma distribution.\\

\noindent 
The likelihood ratio test in Table \ref{tab6} shows that, even with non-statistically significant parameters,  the GTS distribution fits significantly better than the Bilateral Gamma distribution for both S\&P 500 and SPY ETF indexes. Contrary to the AIC statistics, the BIK statistics do not provide the same information. A comprehensive and detailed examination of the statistical significance of the results is carried out in Section \ref{test}.\\

\noindent
Table \ref{tab4} and Table \ref{tab5} summarized the last row of Table \ref{taba3} and Table \ref{taba4}, respectively in appendix \ref{eq:an1}, which describe the convergence process of the GTS parameter for Bitcoin, Ethereum, S\&P 500 index and SPY ETF return data. The convergence process was obtained using the Newton-Raphson iteration algorithm (\ref{eq:l34}). Each row has eleven columns made of the iteration number, the seven parameters \textbf{$\mu$}, \textbf{$\beta_{+}$}, \textbf{$\beta_{-}$}, \textbf{$\alpha_{+}$}, \textbf{$\alpha_{-}$}, \textbf{$\lambda_{+}$}, \textbf{$\lambda_{-}$}, and three statistical indicators: the log-likelihood (\textbf{$Log(ML)$}), the norm of the partial derivatives (\textbf{$||\frac{dLog(ML)}{dV}||$}), and the maximum value of the eigenvalues (\textbf{$Max Eigen Value$}). The statistical indicators aim at checking if the two necessary and sufficient conditions described in (\ref{eq:l33}) are all met. \textbf{$Log(ML)$} displays the value of the Naperian logarithm of the likelihood function $L(x, V)$, as described in (\ref{eq:l31}); \textbf{$||\frac{dLog(ML)}{dV}||$} displays the value of the norm of the first derivatives (\textbf{$\frac{dl(x, V)}{dV_j}$}) described in Equation (\ref{eq:l32}); and \textbf{$Max Eigen Value$} displays the maximum value of the seven eigenvalues generated by the Hessian matrix (\textbf{$\frac{d^{2}l(x, V)}{dV_{k}dV_{j}}$}), as described in (\ref{eq:l32}).\\

\noindent
Similarly, Table \ref{tabe2} and Table \ref{tabf2} describe the convergence process of the Bilateral Gamma distribution parameter for S\&P 500 index and SPY ETF return data.
\begin{table}[ht]
\vspace{-0.5cm}
 \caption{ Likelihood Ratio Test Statistic \& P-value}
\label{tab6}
\centering
\setlength{\tabcolsep}{0.8mm}
\begin{tabular}{@{}c|c|c|c|ccc@{}}
\toprule
 \multicolumn{1}{c|}{} & \multicolumn{1}{c|}{} & \multicolumn{1}{c|}{\textbf{GTS}} & \multicolumn{1}{c|}{\textbf{GTS variants}} & \multicolumn{1}{c|}{\textbf{{$\chi^{2}$}-Value}} & \multicolumn{1}{c|}{\textbf{df}} & \multicolumn{1}{c}{\textbf{P-Value}} \\ \midrule
\textbf{}    & \textbf{Log(ML)} & -10606.73 & -10606.81 & 0.1525 & 1 & 0.6962 \\
\textbf{Bitcoin}           & \textbf{AIC}    & 21227.47 & 21225.62 &  &  &   \\
\textbf{}           & \textbf{BIK}    & 21271.67 & 21263.51 &  &  & \\ \midrule
\textbf{}   & \textbf{Log(ML)} & -9552.86 & -9553.90 & 2.0810 & 2 & 0.3533 \\
\textbf{Ethereum}           & \textbf{AIC}    & 19119.72 & 19117.81 &  &  &   \\
\textbf{}           & \textbf{BIK}    & 19162.32 & 19148.23 &  &  &  \\ \midrule
\textbf{} & \textbf{Log(ML)} & -4920.52 & -4924.62 & 8.1828 & 2 & 0.0167\\
\textbf{S\&P   500}           & \textbf{AIC}    & 9851.06 & 9859.24 &  &  &        \\
\textbf{}           & \textbf{BIK}    & 9898.49 & 9890.26 &  &  &   \\ \midrule

\textbf{}  & \textbf{Log(ML)} & -4893.21 & -4898.67 & 10.9234 & 2 & 0.0042 \\
\textbf{SPY   ETF}           & \textbf{AIC}    & 9800.42 & 9807.34 &  &  &        \\
\textbf{}           & \textbf{BIK}    &  9843.84 & 9838.36 &  &  &       \\ \bottomrule
\end{tabular}%
\end{table}

\noindent
GTS parameter estimation in Table \ref{tab4} and Table \ref{tab5} were used to evaluate the impact of the parameters on the GTS probability density function. As shown in Fig \ref{fig55} and  Fig \ref{fig66}, the effect of the GTS parameters on the probability density function generated by S\&P 500  and SPY ETF have the same patterns. As shown in Fig \ref{fig5a} and Fig \ref{fig5b}, based on the  S\&P 500  return data, \textbf{$\beta_{+}$} (\textbf{$\alpha_{+}$} ) has a higher effect on the probability density function than \textbf{$\beta_{-}$} (\textbf{$\alpha_{-}$}). However, \textbf{$\lambda_{-}$} and \textbf{$\lambda_{+}$}  in Fig \ref{fig5c} are symmetric and have the same impact.

\begin{figure}[ht]
\vspace{-0.3cm}
    \centering
  \begin{subfigure}[b]{0.3\linewidth}
    \includegraphics[width=\linewidth]{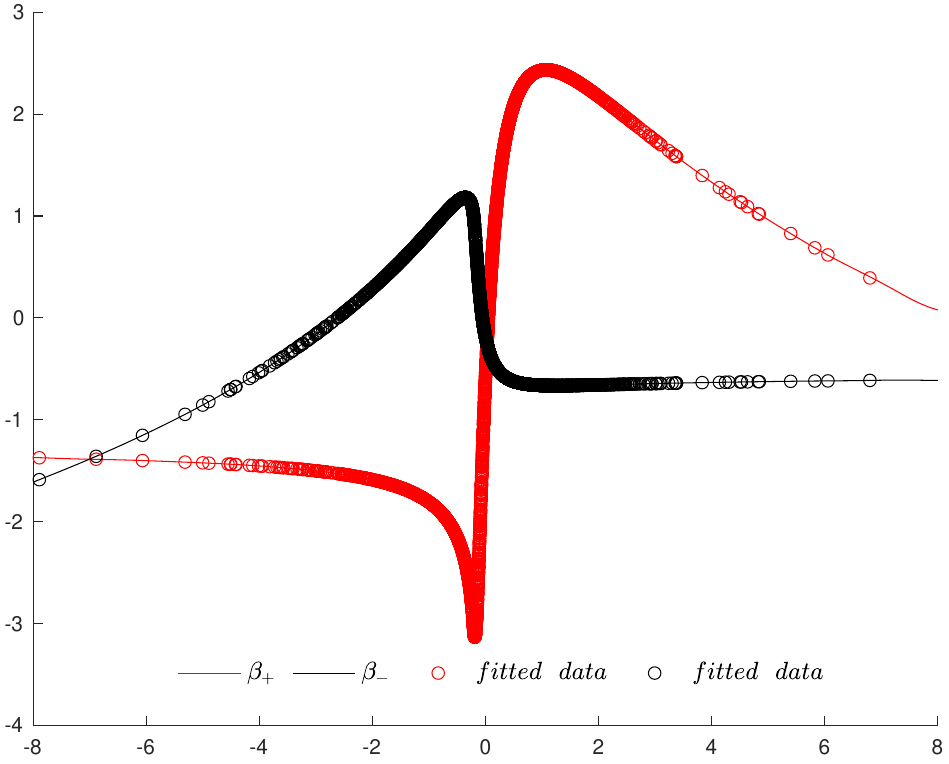}
\vspace{-0.5cm}
     \caption{$V_{j}=\beta_{+}, V_{j}=\beta_{-}$}
         \label{fig5a}
  \end{subfigure}
  \begin{subfigure}[b]{0.3\linewidth}
    \includegraphics[width=\linewidth]{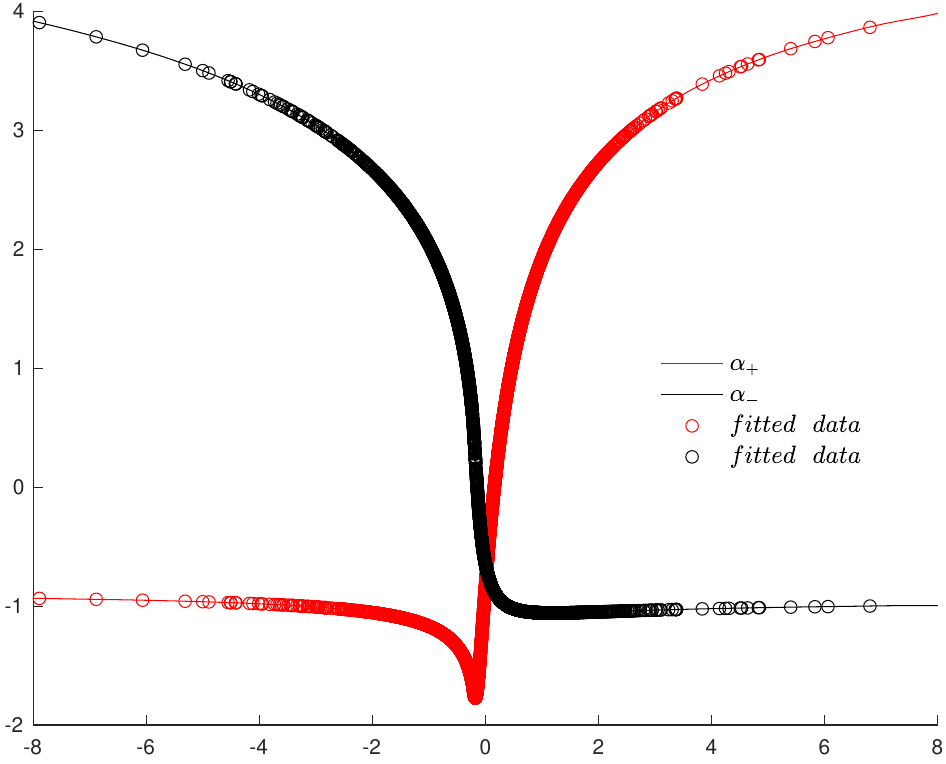}
\vspace{-0.5cm}
     \caption{$V_{j}=\alpha_{+}, V_{j}=\alpha_{-}$}
         \label{fig5b}
          \end{subfigure}
  \begin{subfigure}[b]{0.3\linewidth}
    \includegraphics[width=\linewidth]{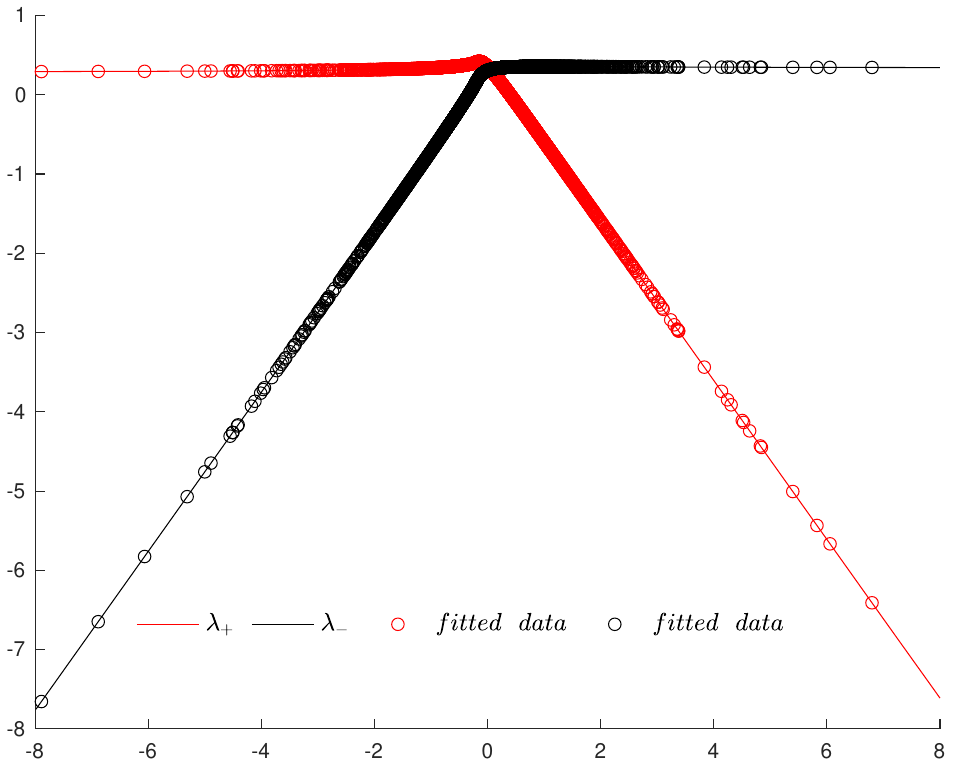}
\vspace{-0.5cm}
     \caption{$V_{j}=\lambda_{+}, V_{j}=\lambda_{-}$}
         \label{fig5c}
          \end{subfigure}
\vspace{-0.3cm}
  \caption{$\frac{\frac{df(x,V)}{dV_j}}{f(x,V)}$:  Effect of Parameters on the GTS Probability Density (S\&P 500 Index)}
  \label{fig55}
\vspace{-0.5cm}
\end{figure}

\begin{figure}[ht]
\vspace{-0.5cm}
    \centering
  \begin{subfigure}[b]{0.3\linewidth}
    \includegraphics[width=\linewidth]{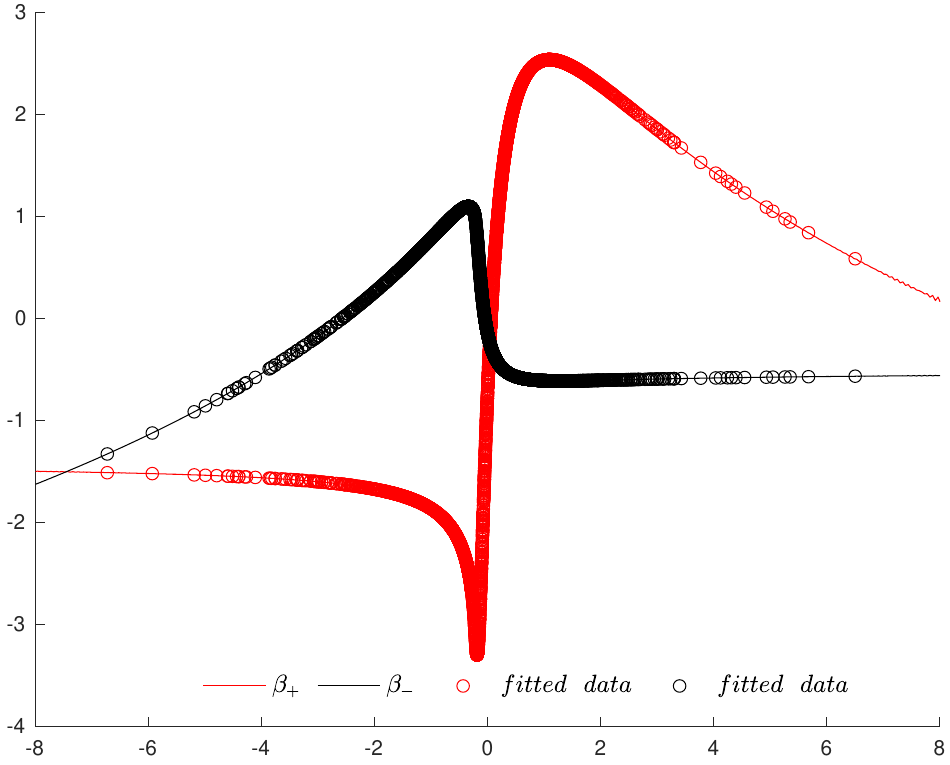}
\vspace{-0.5cm}
     \caption{$V_{j}=\beta_{+}, V_{j}=\beta_{-}$}
         \label{fig6a}
  \end{subfigure}
  \begin{subfigure}[b]{0.3\linewidth}
    \includegraphics[width=\linewidth]{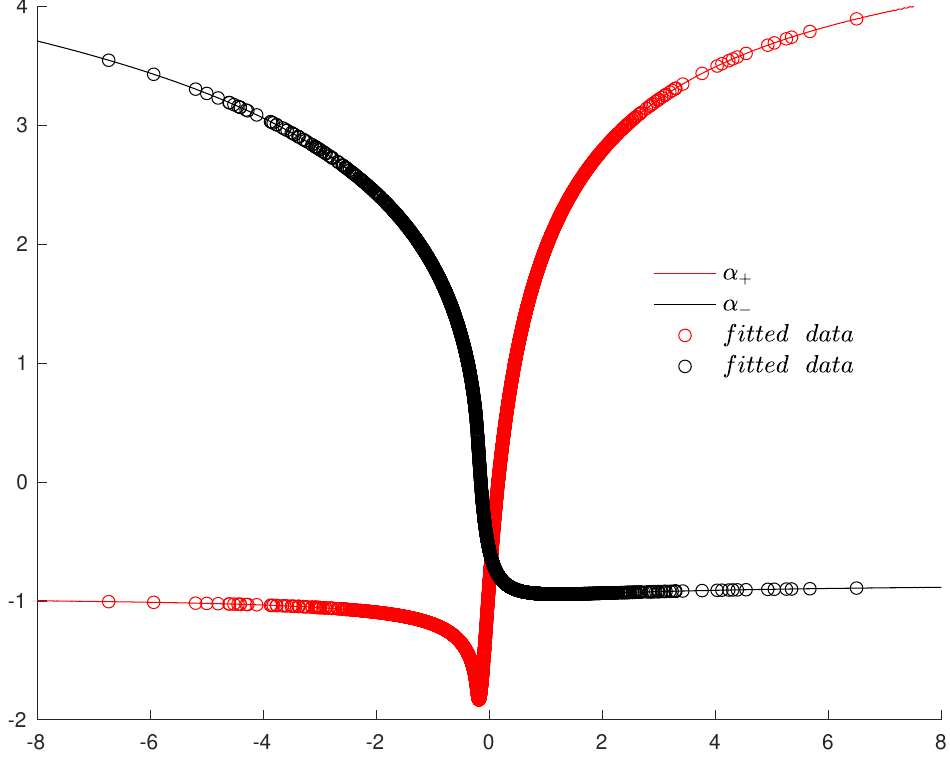}
\vspace{-0.5cm}
     \caption{$V_{j}=\alpha_{+}, V_{j}=\alpha_{-}$}
         \label{fig6b}
          \end{subfigure}
  \begin{subfigure}[b]{0.3\linewidth}
    \includegraphics[width=\linewidth]{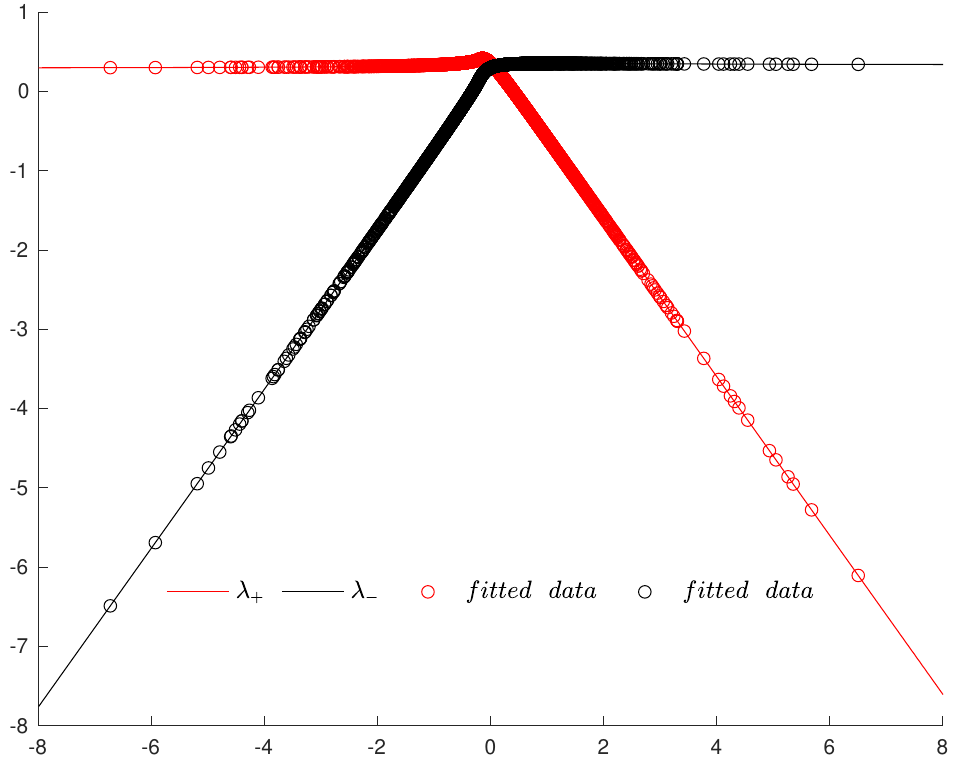}
\vspace{-0.5cm}
     \caption{$V_{j}=\lambda_{+}, V_{j}=\lambda_{-}$}
         \label{fig6c}
          \end{subfigure}
\vspace{-0.3cm}
  \caption{$\frac{\frac{df(x, V)}{dV_j}}{f(x, V)}$: Effect of Parameters on the GTS Probability Density (SPY EFT )}
  \label{fig66}
\vspace{-0.5cm}
\end{figure}
\subsection{Evaluation of the Methods of Moments}
\noindent
Maximum likelihood GTS parameter estimation in Table \ref{tab4} and Table \ref{tab5} are used to evaluate the system of equations in (\ref{eq:l37bc}).  As shown in Table \ref{tab7}, the solution of the maximum likelihood method satisfies at a certain extent the equations for the following four moments: $\hat{m}_{1}$, $\hat{m}_{2}$, $\hat{m}_{4}$, $\hat{m}_{5}$ in the system (\ref{eq:l37bc}). As for Bitcoin and Ethereum, the $7^{th}$ moment equation has the highest relative error: 53.3\% for S\&P 500 index and -85.9\% for SPY ETF. Therefore, the maximum likelihood GTS parameter estimation is not the GTS parameter estimation from the method of moments.
\begin{table}[ht]
\vspace{-0.6cm}
\centering
 \caption{Evaluation of the Methods of Moment}
\label{tab7}
\setlength{\tabcolsep}{0.3mm}
\begin{tabular}{@{}lcccccc@{}}
\toprule
 & \multicolumn{3}{c}{\textbf{S\&P 500 Index}} & \multicolumn{3}{c}{\textbf{SPY ETF}} \\ \midrule
\multirow{1}{*}{} & \multirow{1}{*}{\textbf{Empirical(1)}} & \multirow{1}{*}{\textbf{Theoretical(2)}} & \multirow{1}{*}{\textbf{$\frac{(1)-(2)}{2}$}} & \multirow{1}{*}{\textbf{Empirical(1)}} & \multirow{1}{*}{\textbf{Theoretical(2)}} & \multirow{1}{*}{\textbf{$\frac{(1)-(2)}{2}$}} \\ \midrule
\multirow{1}{*}{\textbf{Sample size}} & \multirow{1}{*}{3656} & \multirow{1}{*}{} & \multirow{1}{*}{} & \multirow{1}{*}{3655} & \multirow{1}{*}{} & \multirow{1}{*}{} \\
\multirow{1}{*}{\textbf{$\hat{m}_{1}$}} & \multirow{1}{*}{0.045} & \multirow{1}{*}{0.045} & \multirow{1}{*}{-0.5\%} & \multirow{1}{*}{0.054} & \multirow{1}{*}{0.054} & \multirow{1}{*}{0.0\%} \\
\multirow{1}{*}{\textbf{$\hat{m}_{2}$}} & \multirow{1}{*}{1.069} & \multirow{1}{*}{1.083} & \multirow{1}{*}{-1.3\%} & \multirow{1}{*}{1.053} & \multirow{1}{*}{1.044} & \multirow{1}{*}{0.8\%} \\
\multirow{1}{*}{\textbf{$\hat{m}_{3}$}} & \multirow{1}{*}{-0.447} & \multirow{1}{*}{-0.341} & \multirow{1}{*}{31.2\%} & \multirow{1}{*}{-0.214} & \multirow{1}{*}{-0.351} & \multirow{1}{*}{-39.0\%} \\
\multirow{1}{*}{\textbf{$\hat{m}_{4}$}} &\multirow{1}{*}{8.371} & \multirow{1}{*}{9.764} & \multirow{1}{*}{-14.3\%} & \multirow{1}{*}{8.197} & \multirow{1}{*}{7.691} & \multirow{1}{*}{6.6\%} \\
\multirow{1}{*}{\textbf{$\hat{m}_{5}$}} &\multirow{1}{*}{-16.386} & \multirow{1}{*}{-11.128} & \multirow{1}{*}{47.3\%} & \multirow{1}{*}{-3.969} & \multirow{1}{*}{-12.717} & \multirow{1}{*}{-68.8\%} \\
\multirow{1}{*}{\textbf{$\hat{m}_{6}$}} & \multirow{1}{*}{193.563} & \multirow{1}{*}{247.811} & \multirow{1}{*}{-21.9\%} & \multirow{1}{*}{157.645} & \multirow{1}{*}{162.048} & \multirow{1}{*}{-2.7\%} \\
\multirow{1}{*}{\textbf{$\hat{m}_{7}$}} & \multirow{1}{*}{-840.097} & \multirow{1}{*}{-547.882} & \multirow{1}{*}{53.3\%} & \multirow{1}{*}{-85.003} & \multirow{1}{*}{-602.447} & \multirow{1}{*}{-85.9\%} \\
\multirow{1}{*}{\textbf{Standard deviation \footnotemark[1]}} & \multirow{1}{*}{1.082} & \multirow{1}{*}{1.033} & \multirow{1}{*}{4.7\%} & \multirow{1}{*}{1.050} & \multirow{1}{*}{1.021} & \multirow{1}{*}{2.9\%} \\
\multirow{1}{*}{\textbf{Skewness \footnotemark[2]}} &  \multirow{1}{*}{-0.432} & \multirow{1}{*}{-0.535} & \multirow{1}{*}{-19.2\%} & \multirow{1}{*}{-0.358} & \multirow{1}{*}{-0.490} & \multirow{1}{*}{-26.9\%} \\
\multirow{1}{*}{\textbf{Kurtosis \footnotemark[3]}} & \multirow{1}{*}{8.413} & \multirow{1}{*}{7.435} & \multirow{1}{*}{13.1\%} & \multirow{1}{*}{7.495} & \multirow{1}{*}{7.177} & \multirow{1}{*}{4.4\%} \\
\multirow{1}{*}{\textbf{Max value}} & \multirow{1}{*}{6.797} & \multirow{1}{*}{} & \multirow{1}{*}{} & \multirow{1}{*}{6.501} & \multirow{1}{*}{} & \multirow{1}{*}{} \\
\multirow{1}{*}{\textbf{Min Value}} & \multirow{1}{*}{-7.901} & \multirow{1}{*}{} & \multirow{1}{*}{} & \multirow{1}{*}{-6.734} & \multirow{1}{*}{} & \multirow{1}{*}{} \\ \bottomrule
\end{tabular}%
\vspace{-0.3cm}
\end{table}
\footnotetext[1]{$\sigma=\sqrt{\kappa_{2}}$}
\footnotetext[2]{Skewness is estimated as $\frac{\kappa_{3}}{\kappa_{2}^{3/2}}$}
\footnotetext[3]{Kurtosis is estimated as $3+\frac{\kappa_{4}}{\kappa_{2}^{2}}$;  $\kappa_{1}$, $\kappa_{2}$ and $\kappa_{2}$ are defined in (\ref{eq:l12})}

\noindent
In addition to the moment estimations in Table \ref{tab7}, the lower relative errors show that the empirical and theoretical standard deviation ($\sigma$), skewness, and kurtosis are consistent for S\&P 500  and SPY ETF. The empirical \& theoretical statistics show that both assets are skewed to the left and also thicker than the Normal distribution.

\clearpage
\newpage

\section{Goodness-of-fit Analysis\label{test}}

\subsection{Kolmogorov-Smirnov (KS) Analysis}
\noindent
Given the sample of daily return $\{y_{1}, y_{2}\dots y_{m}\}$ of size $m$ and the empirical cumulative distribution function, $F_{m}(x)$, for each index, the Kolmogorov-Smirnov (KS) test is performed under the null hypothesis, $H_{0}$, that the sample $\{y_{1}, y_{2}\dots y_{m}\}$ comes from the GTS distribution, $F(x)$. The cumulative distribution function of the theoretical distribution, $F(x)$, needs to be computed. The density function, $f(x)$, does not have a closed form, the same for the cumulative function, $F(x)$, in (\ref{eq:l52}). However, we know the closed form of the Fourier of the density function, $\scrF[f]$, and the relationship in (\ref{eq:l53}) provides the Fourier of the cumulative distribution function, $\scrF[F]$. The GTS distribution function, $F(x)$, was computed from the inverse of the Fourier of the cumulative distribution, $\scrF[F]$, in (\ref{eq:l54}):
  \begin{align}
 Y &\sim GTS(\textbf{$\mu$}, \textbf{$\beta_{+}$}, \textbf{$\beta_{-}$}, \textbf{$\alpha_{+}$},\textbf{$\alpha_{-}$}, \textbf{$\lambda_{+}$}, f\textbf{$\lambda_{-}$}) \label{eq:l51}\\
   F(x)&= \int_{-\infty}^{x} f(t) \mathrm{d}t \hspace{5mm}  \hbox{$f$ is the density function of $Y$} \label{eq:l52}\\
 \scrF[F](x)&=\frac{\scrF[f](x)}{ix} + \pi\scrF[f](0)\delta (x) \label{eq:l53}\\
F(x) &= \frac{1} {2\pi}\int_{-\infty}^{+\infty}\! \frac{\scrF[f](y)}{iy}e^{ixy}\, \mathrm{d}y + \frac{1}{2}\label{eq:l54}
 \end{align}

 \noindent
See Appendix A in \cite{nzokem2021fitting} for (\ref{eq:l53}) proof.\\

\noindent
The two-sided KS goodness-of-fit statistic ($D_{m}$) is defined as follows:
 \begin{align}
D_{m} = \sup_{x}{|F(x)-F_{m}(x)|}, \label{eq:l55}
 \end{align}
\noindent
where $m$ is the sample size, $F_{m}(x)$ denotes the empirical cumulative distribution of $\{y_{1}, y_{2}\dots y_{m}\}$.\\

\noindent
The distribution of Kolmogorov’s goodness-of-fit measure $D_{m}$ has been studied extensively in the literature. It was shown \cite{massey1951kolmogorov} that the $D_{m}$ distribution is independent of the theoretical distribution, $F(x)$, under the null hypothesis, $H_{0}$. The discrete, mixed, and discontinuous distributions case has also been studied \cite{dimitrova2020computing}. Under the null hypothesis, $H_{0}$, that the sample $\{y_{1}, y_{2}\dots y_{m}\}$ of size $m$ comes from the hypothesized continuous distribution, it was shown \cite{an1933sulla} that the asymptotic statistic $\sqrt {n}D_{n}$ converges to the Kolmogorov distribution.\\
  \noindent
The limiting form for the distribution function of Kolmogorov’s goodness-of-fit measure $D_{m}$  is
 \begin{align}
\lim_{m \to +\infty}Pr(\sqrt{m}D_{m} \leq x) =1 - 2 \sum_{k=1}^{+\infty}(-1)^{k-1}e^{-2k^{2}x^{2}}= \frac{\sqrt{2\pi}}{x} \sum_{k=1}^{+\infty}e^{-\frac{(2k-1)^{2}\pi^{2}}{8x^{2}}}.\label{eq:l56}
 \end{align}
\noindent
The first representation was given in \cite{an1933sulla},  and the second came from a standard relation for theta functions \cite{marsaglia2003evaluating}. \\

\noindent
 As shown in Fig \ref{fig62}, the asymptotic statistic, $\sqrt {n}D_{n}$, is a positively skewed distribution with a mean and  a standard deviation \cite{marsaglia2003evaluating}
 \begin{align}
\mu= \sqrt{\frac{\pi}{2}} log(2)\sim 0.8687, \quad \quad \sigma=\sqrt{\frac{\pi^2}{12} - \mu^2} \sim 0.2603. \label{eq:l57}
 \end{align}

 \noindent
At $5\%$ risk level, the risk threshold is $d=1.3581$ and represents the area in the shaded area under the probability density function.

 \begin{figure}[ht]
 \vspace{-0.2cm}
     \centering
         \includegraphics[scale=0.45] {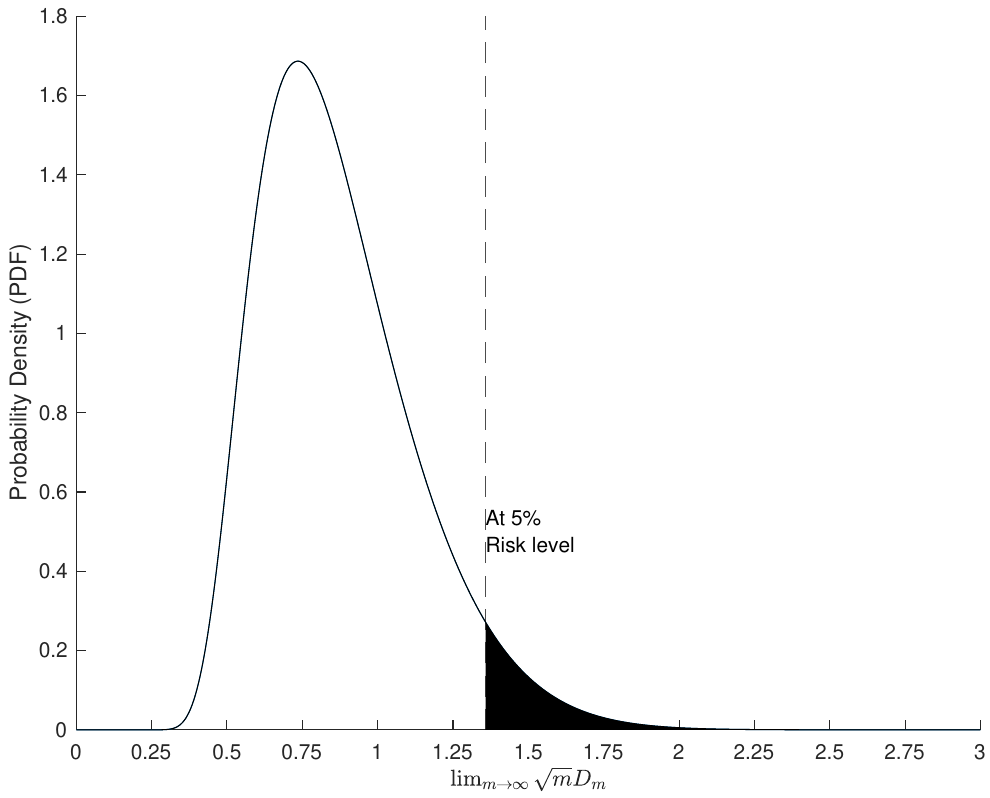}
 \vspace{-0.1cm}
        \caption{Asymptotic Statistic ($\sqrt{m}D_{m}$) Probability Density Function (PDF)}
        \label{fig62}
 \vspace{-0.3cm}
\end{figure}

 \noindent
 The p-value of the test statistic, $D_{m}$, is computed based on (\ref{eq:l56}) as follows:
 \begin{align}
P\_{value} = Pr(D_{m}> \hat{D}_{m} | H_{0})=1 - Pr(\sqrt{m}D_{m}\leq \sqrt{m}\hat{D}_{m}). \label{eq:l58}
 \end{align}

\noindent
A p-value is defined as the probability that values are even more extreme or more in the tail than our test statistic.  A small p-value leads to a rejection of the null hypothesis, $H_{0}$, because the test statistic, $D_{m}$, is already extreme. We reject the hypothesis if the p-value is less than our level of significance, which we take to be equal to 0.05.\\

  \noindent
$\hat{D}_{m}$ is a realization value of the KS estimator $D_{m}$ computed from the sample $\{y_{1}, y_{2}\dots y_{m}\}$. $\hat{D}_{m}$ is estimated \cite{krysicki1999rachunek} as follows:
\begin{align}
\hat{D}_{m}= Max(\sup_{0\leq j\leq P}{|F(x_{j})-F_{m}(x_{j})|}, \sup_{1\leq j\leq P}{|F(x_{j})-F_{m}(x_{j-1})|}). \label{eq:l59}
 \end{align}
The following computations were performed for Bitcoin BTC data, and the quantity $\hat{D}_{m}$ was obtained:
 \begin{equation}
 \begin{aligned}
&\sup_{0\leq j\leq P}{|F(x_{j})-F_{m}(x_{j})|}= 0.01300 \\
&\sup_{1\leq j\leq P}{|F(x_{j})-F_{m}(x_{j-1})|}=0.00538	\\  \label{eq:l59a}
&\hat{D}_{m}=0.01300 \\
&P\_{value} = prob(\sqrt{m}D_{m}>0.6903 | H_{0})=49.48\%.
 \end{aligned}
  \end{equation}
\noindent
For each index and model, KS-statistics ($\hat{D}_{m}$) and p-values associated were computed and summarized in Table \ref{tab8} along with the index sample size, $m$.
\begin{table}[ht]
 \caption{ Kolmogorov-Smirnov Statistic \& P-value }
\vspace{-0.3cm}
\label{tab8}
\centering
\setlength{\tabcolsep}{0.75mm}
\begin{tabular}{@{}l | c c c | c c c | c c c |c@{}}
\toprule
 & \multicolumn{3}{c|}{GTS} & \multicolumn{3}{c|}{GBN} & \multicolumn{3}{c|}{GTS variants} &\multicolumn{1}{c}{Sample size}\\ \toprule
 \multirow{1}{*}{Index} & \multirow{1}{*}{$\hat{D}_{m}$} & \multirow{1}{*}{$\sqrt{m}\hat{D}_{m}$} &  \multirow{1}{*}{P\_value} &  \multirow{1}{*}{$\hat{D}_{m}$} & \multirow{1}{*}{$\sqrt{m}\hat{D}_{m}$} &  \multirow{1}{*}{P\_value} &\multirow{1}{*}{$\hat{D}_{m}$} &  \multirow{1}{*}{$\sqrt{m}\hat{D}_{m}$} &  \multirow{1}{*}{P\_value} & \multirow{1}{*}{m}  \\ \toprule
Bicoin BTC & \multirow{1}{*}{0.013} & \multirow{1}{*}{0.830}&\multirow{1}{*}{0.494} & \multirow{1}{*}{0.106} & \multirow{1}{*}{6.803}& \multirow{1}{*}{0.000}& \multirow{1}{*}{0.014 \footnotemark[4]} &\multirow{1}{*}{0.863}& \multirow{1}{*}{0.445}& \multirow{1}{*}{4083}  \\
Ethereum & \multirow{1}{*}{0.012} & \multirow{1}{*}{0.721}&\multirow{1}{*}{0.674} & \multirow{1}{*}{0.092} & \multirow{1}{*}{5.249}& \multirow{1}{*}{0.000}& \multirow{1}{*}{0.013 \footnotemark[5]} &\multirow{1}{*}{0.749}& \multirow{1}{*}{0.627}& \multirow{1}{*}{3246} \\
S\&P 500  & \multirow{1}{*}{0.012} & \multirow{1}{*}{0.750}& \multirow{1}{*}{0.627} & \multirow{1}{*}{0.091} & \multirow{1}{*}{5.550}& \multirow{1}{*}{0.000} & \multirow{1}{*}{0.014 \footnotemark[6]} &\multirow{1}{*}{0.897}& \multirow{1}{*}{0.395} & \multirow{1}{*}{3656}  \\
SPY ETF & \multirow{1}{*}{0.014} & \multirow{1}{*}{0.869}& \multirow{1}{*}{0.436} & \multirow{1}{*}{0.089} & \multirow{1}{*}{5.438}& \multirow{1}{*}{0.000}& \multirow{1}{*}{0.016 \footnotemark[6]} &\multirow{1}{*}{1.010}& \multirow{1}{*}{0.258}& \multirow{1}{*}{3655} \\ \bottomrule
\end{tabular}%
\end{table}
\footnotetext[4]{Kobol distribution ($\beta=\beta_{-}=\beta_{+}$)}
\footnotetext[5]{Carr-Geman-Madan-Yor (CGMY) distributions ($\beta=\beta_{-}=\beta_{+}$; $\alpha=\alpha_{-}=\alpha_{+}$)}
\footnotetext[6]{bilateral Gamma distribution ($\beta_{-}=\beta_{+}=0$)}

\noindent
The asymptotic statistics, $\sqrt {n}D_{n}$, produced from the two-parameter geometric Brownian motion (GBM) hypothesis, have high values and show that the GBM hypothesis is always rejected. On the other hand, the high p-values generated by the asymptotic statistics suggest insufficient evidence to reject the assumption that the data were randomly sampled from a GTS.  The same observations work for the GTS variants: Kobol, CGMY, and Bilateral Gamma distributions. In addition, as shown the p-value indicator in Table \ref{tab8}, the GTS distribution outperforms the Bilateral Gamma distribution for the S\&P 500 and SPY ETF indexes. However, the Kobol and CGMY distributions respectively, for Bitcoin and Ethereum have almost the same performance as the GTS distribution.

\subsection{ Anderson-Darling Test Analysis}
\noindent
The Anderson-Darling test \cite{anderson} is a goodness-of-fit test that allows the control of the hypothesis that the distribution of a random variable
observed in a sample follows a certain theoretical distribution. The Anderson-Darling statistic belongs to the class of quadratic EDF statistics\cite{stephens1974edf} based on the empirical distribution function. The quadratic EDF statistics measure the distance between the hypothesized distribution ($F(x)$) and empirical distribution.  It is defined as
 \begin{align}
m\int_{-\infty}^{+\infty}\left(F_{m}(x) - F(x)\right)^2w(x) \,dF{x}, \label{eq:l71}
 \end{align}
where $m$ is the number of elements in the sample,$w(x)$ is a weighting function, and $F_{n}(x)$ is the empirical distribution function defined on the sample of size $n$.\\

\noindent
When the weighting function is $w(x)=1$, the statistic (\ref{eq:l71}) is the Cramér–von Mises statistic, while the Anderson–Darling statistic is obtained by choosing the weighting function
$w(x)=F(x)\left(1-F(x)\right)$.  Compared with the Cramér–von Mises statistic, the Anderson–Darling statistic places more weight on the tails of the distribution.\\

The Anderson–Darling statistic is
 \begin{align}
A_{m}^2=m\int_{-\infty}^{+\infty }\frac{\left(F_{m}(x) - F(x)\right)}{F(x)\left(1-F(x)\right)} \,dF(x). \label{eq:l73}
 \end{align}
It can be shown that the asymptotic distribution of the Anderson–Darling statistic, $A_{m}^2$, is independent of the theoretical distribution under the null hypothesis.  The asymptotic distribution \cite{lewis1961distribution,marsaglia2004evaluating} is defined as follows:
 \begin{equation}
 \begin{aligned}
G(x)&=\lim_{m\to\infty}Pr \left[A_{m}^2<x\right] =\sum_{j=0}^{+\infty } a_{j}(xb_{j})^{-\frac{1}{2}}exp(-\frac{b_{j}}{x}) \int_{0}^{+\infty }f_{j}(y)exp(-y^2) dy\\ \label{eq:l74}
f_{j}(y)&=exp\left(\frac{1}{8}\frac{x b_{j}}{y^2 x + b_{j}} \right), \quad \quad a_{j}= \frac{(-1)^{j}(2)^{\frac{1}{2}}(4j+1)\Gamma(j+\frac{1}{2})}{j!} \\ b_{j}&=\frac{1}{2}(4j+1)^2 \pi^2.
\end{aligned}
\end{equation}

  \noindent
 As shown in Fig \ref{fig71}, the asymptotic distribution of the Anderson–Darling statistic ($A_{m}^2$) is a positively skewed distribution with a mean and a standard deviation \cite{Anderson2011}
 \begin{align}
\mu= 1, \quad \quad \sigma=\sqrt{\frac{2}{3}(\pi^2 - 9)} \sim 0.761. \label{eq:l57}
 \end{align}
 
\newpage
 \begin{figure}[ht]
 \vspace{-0.5cm}
     \centering
         \includegraphics[scale=0.40] {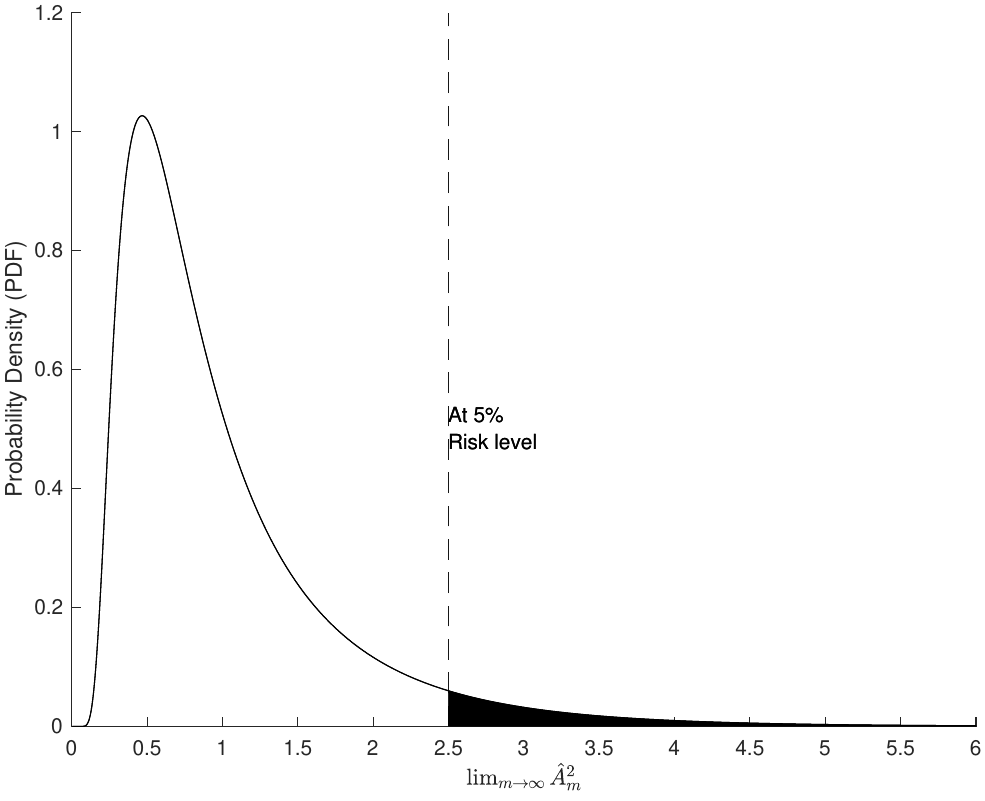}
        \caption{Asymptotic Anderson–Darling Statistic ($A_{m}^2$) Probability Density Function (PDF)}
        \label{fig71}
 \vspace{-0.8cm}
\end{figure}
 \noindent
At $5\%$ risk level, the risk threshold is $d=2.4941$ and represents the area in the shaded area under the probability density function.\\

\noindent
The p-value of the test statistic, $A_{m}^2$, is defined as follows:
 \begin{align}
P\_{value} = prob(A^2_{m} > \hat{A}^2_{m} | H_{0}) =1 - G(\hat{A}^2_{m}). \label{eq:l75}
 \end{align}

\noindent
In order to compute the Anderson–Darling statistic, $A_{m}^2$, in (\ref{eq:l73}), the sample of daily return $\{y_{1}, y_{2}\dots y_{m}\}$ of size $m$ is arranged in ascending order: $y_{(1)}< y_{(2)}<\dots <y_{(m)}$. The Anderson–Darling statistic \cite{lewis1961distribution} then becomes
 \begin{align}
A_{m}^2=-m - \frac{1}{m}\sum_{j=1}^{m}\left[(2j-1)log(F(y_{(j)}))+(2(n-j)+1)log(F(y_{(j)}))\right].\label{eq:l76}
 \end{align}
\noindent
For each index, the Anderson–Darling statistic (\ref{eq:l76}) is computed along with the p-value statistic. Table \ref{tab9} shows the KS-statistics ($A^{2}_{m}$) and $P\_values$ for the GTS, GBM, and the GTS variant distributions. While the two-parameter GBM hypothesis is always rejected, the GTS hypothesis is accepted and yields a very high p-value.
\begin{table}[ht]
\vspace{-0.3cm}
 \caption{Anderson–Darling Statistic \& P-value}
\label{tab9}
\centering
\begin{tabular}{@{}l | c c | c c | c c |c@{}}
\toprule
 & \multicolumn{2}{c|}{GTS} & \multicolumn{2}{c|}{GBN} & \multicolumn{2}{c|}{GTS variants} &\multicolumn{1}{c}{Sample size}\\ \toprule
 \multirow{1}{*}{Index} & \multirow{1}{*}{$\hat{A}_{m}^2$} &  \multirow{1}{*}{P\_value} &  \multirow{1}{*}{$\hat{A}_{m}^2$} &  \multirow{1}{*}{P\_value} &  \multirow{1}{*}{$\hat{A}_{m}^2$} &  \multirow{1}{*}{P\_value} &\multirow{1}{*}{m} \\ \toprule
Bicoin BTC & \multirow{1}{*}{0.1098} &\multirow{1}{*}{0.9999} & \multirow{1}{*}{99.706} & \multirow{1}{*}{0.0000}&  \multirow{1}{*}{0.1105 \footnotemark[4]} & \multirow{1}{*}{0.9999}&\multirow{1}{*}{4083}  \\
Ethereum & \multirow{1}{*}{0.1018} &\multirow{1}{*}{0.9999} & \multirow{1}{*}{59.157} & \multirow{1}{*}{0.0001}& \multirow{1}{*}{0.2123 \footnotemark[5]} & \multirow{1}{*}{0.9866}& \multirow{1}{*}{3246} \\
S\&P 500  & \multirow{1}{*}{0.3007} & \multirow{1}{*}{0.9376} & \multirow{1}{*}{54.304} & \multirow{1}{*}{0.0001}& \multirow{1}{*}{0.5010 \footnotemark[6]} & \multirow{1}{*}{0.7458}&\multirow{1}{*}{3656}  \\
SPY ETF & \multirow{1}{*}{0.3017} & \multirow{1}{*}{0.9368} & \multirow{1}{*}{51.516} & \multirow{1}{*}{0.0001}& \multirow{1}{*}{0.6684 \footnotemark[6]} & \multirow{1}{*}{0.5857}& \multirow{1}{*}{3655} \\ \bottomrule
\end{tabular}%
\end{table}
\footnotetext[4]{Kobol distribution ($\beta=\beta_{-}=\beta_{+}$)}
\footnotetext[5]{Carr-Geman-Madan-Yor (CGMY) distributions ($\beta=\beta_{-}=\beta_{+}$; $\alpha=\alpha_{-}=\alpha_{+}$)}
\footnotetext[6]{Bilateral Gamma distribution ($\beta_{-}=\beta_{+}=0$)}

\noindent
In addition, as shown by the p-value indicator in Table \ref{tab9}, the GTS distribution outperforms the Bilateral Gamma distribution for the S\&P 500 and SPY ETF indexes. However, the Kobol and CGMY distributions for Bitcoin and Ethereum, respectively, have almost the same performance as the GTS distribution.

\subsection{ Pearson's Chi-squared Test  Analysis}
\noindent
Pearson's chi-squared test \cite{schoutens2003levy} counts the number of sample points falling into certain intervals and compares them with the expected number under the null hypothesis. Under the null hypothesis, $H_{0}$, a random sample $\{y_{1}, y_{2}\dots y_{m}\}$ comes from the GTS distribution, which depends on the seven parameters estimated in Section \ref{estim}. Suppose that $m$ observations in the sample from a population are classified into $K$ mutually exclusive classes with respective observed numbers of observations $N_{j}$ (for $j = 1,2,…,K$), and a null hypothesis gives the probability $\Pi_{j}=F(x_{j}) - F(x_{j-1})$ (\ref{eq:l54}) that an observation falls into the $j^{th}$ class. \\

\noindent
The following Pearson statistic calculates the value of the chi-squared goodness-of-fit test:

 \begin{equation}
 \begin{aligned}
\chi^2(K-1-p) = \sum_{j=1}^{K}\frac{\left( N_{j} -  m\Pi_{j}\right)^2}{m\Pi_{j}}. \label{eq:l77}
\end{aligned}
\end{equation}

Under the null hypothesis assumption, as $m $ goes to $+\infty$,  the limiting distribution\cite{schoutens2003levy} of the  Pearson statistic (\ref{eq:l77}) follows the $\chi^2(K-1-p)$ distribution with $K-1-p$ degrees of freedom, p is the number of estimated parameters.\\
\noindent
Table \ref{tab10} shows the Pearson chi-squared statistics ($\hat{\chi}^2(K-1-p)$), $P-values$ and Class Number for the GTS, GBM and the GTS variant distributions. While the two-parameter GBM hypothesis is always rejected, the GTS hypothesis is accepted and yields a high p-value.
\begin{table}[ht]
\vspace{-0.2cm}
 \caption{ Pearson Statistics \& P-values }
\label{tab10}
\centering
\setlength{\tabcolsep}{0.7mm}
\begin{tabular}{@{}l | cc | cc | ccc | c@{}}
\toprule
 & \multicolumn{2}{c|}{GTS} & \multicolumn{2}{c|}{GBN} &\multicolumn{3}{c|}{GTS variants}& \multicolumn{1}{c}{Class Number} \\ \toprule
 \multirow{1}{*}{Index} & \multirow{1}{*}{$\hat{\chi}^2(K-8)$} &  \multirow{1}{*}{P\_value} &  \multirow{1}{*}{$\hat{\chi}^2(K-3)$} &  \multirow{1}{*}{P\_value} &  \multirow{1}{*}{$\hat{\chi}^2(K-p-1)$}&  \multirow{1}{*}{$p$} &  \multirow{1}{*}{P\_value}& \multirow{1}{*}{K}  \\ \toprule
Bicoin BTC & \multirow{1}{*}{12.234} &\multirow{1}{*}{0.508} & \multirow{1}{*}{1375} & \multirow{1}{*}{0.000} & \multirow{1}{*}{12.549 \footnotemark[4]} &  \multirow{1}{*}{6} & \multirow{1}{*}{0.562} & \multirow{1}{*}{21} \\
Ethereum & \multirow{1}{*}{6.910} &\multirow{1}{*}{0.863} & \multirow{1}{*}{805} & \multirow{1}{*}{0.000} & \multirow{1}{*}{8.618 \footnotemark[5]} &  \multirow{1}{*}{5} & \multirow{1}{*}{0.854} & \multirow{1}{*}{20} \\
S\&P 500  & \multirow{1}{*}{9.886} & \multirow{1}{*}{0.703} & \multirow{1}{*}{574} & \multirow{1}{*}{0.000} & \multirow{1}{*}{12.844 \footnotemark[6]} &  \multirow{1}{*}{5} & \multirow{1}{*}{0.614} & \multirow{1}{*}{21}  \\
SPY ETF & \multirow{1}{*}{13.955} & \multirow{1}{*}{0.377} & \multirow{1}{*}{605} & \multirow{1}{*}{0.000} &  \multirow{1}{*}{18.228 \footnotemark[6]} &  \multirow{1}{*}{5} & \multirow{1}{*}{0.251} &\multirow{1}{*}{21}  \\ \bottomrule
\end{tabular}
\end{table}
\footnotetext[4]{Kobol distribution ($\beta=\beta_{-}=\beta_{+}$)}
\footnotetext[5]{Carr-Geman-Madan-Yor (CGMY) distributions ($\beta=\beta_{-}=\beta_{+}$; $\alpha=\alpha_{-}=\alpha_{+}$)}
\footnotetext[6]{Bilateral Gamma distribution ($\beta_{-}=\beta_{+}=0$)}

\noindent
In addition, as shown by the p-value indicator in Table \ref{tab10}, the GTS distribution outperforms the Bilateral Gamma distribution for the S\&P 500 and SPY ETF indexes. However, the Kobol and CGMY distributions for Bitcoin and Ethereum, respectively, have almost the same performance as the GTS distribution. For more details on the estimation of the Pearson-statistic inputs under the GTS distribution, refer to Table \ref{tabb1} in Appendix \ref{eq:an2}.

\section {Conclusion}
The study provides a methodology for fitting the rich class of the seven-parameter GTS distribution to financial data. Four historical prices were considered in the methodology application: two heavily tailed data (Bitcoin and Ethereum returns) and two peaked data (S\&P 500 and SPY ETF returns). In the study, each historical data was used to fit the seven-parameter GTS distribution to the underlying data return distribution. The advanced fast FRFT scheme, which is based on the classic fast FRFT algorithm and the 11-point composite Newton–Cotes rule, was used to perform the maximum likelihood estimation of seven parameters of the GTS distribution. The maximum likelihood estimate results show that, for each index, the location parameter, $\mu$, is negative while others are positive, as expected in the literature. The statistical significance of the parameters was analyzed. The non-statistical significance of the index of stability parameters ($\beta_{+}$, $\beta_{-}$) has led to the fitting of the Kobol, CGMY, and the Bilateral Gamma distributions. The goodness-of-fit was assessed through Kolmogorov-Smirnov, Anderson-Darling, and Pearson's chi-squared statistics. While the two-parameter GBM hypothesis is always rejected, the goodness-of-fit analysis shows that the GTS distribution fits significantly the four historical data with a very high p-value.\\
   As a main limitation of the study, The applied methodology is compute-intensive, and the researchers need good skills in computer programming. In future work, the estimated parameter of the GTS distribution will be used in the Ornstein-Uhlenbeck type process to simulate the daily cumulative returns of the financial asset.

\bibliography{fmodelArxiv}

\clearpage
\appendix
\renewcommand{\thesection}{\Alph{section}}
\setcounter{section}{0}
\setcounter{equation}{0} 
\setcounter{table}{0} 

\begin{appendices}
\section{Iterative Maximum Likelihood Estimation (MLE) Procedure}\label{eq:an1}
\begin{table*}[ht]
\vspace{-0.3cm}
\centering
\caption{Convergence of the GTS parameter for Bitcoin return data}
\label{taba1}
\resizebox{13cm}{!}{
\begin{tabular}{ccccccccccc}
\hline
\textbf{$Iterations$} & \textbf{$\mu$} & \textbf{$\beta_{+}$} & \textbf{$\beta_{-}$} & \textbf{$\alpha_{+}$} & \textbf{$\alpha_{-}$} & \textbf{$\lambda_{+}$} & \textbf{$\lambda_{-}$} & \textbf{$Log(ML)$} & \textbf{$||\frac{dLog(ML)}{dV}||$} & \textbf{$Max Eigen Value$} \\ \hline
1 & -0.7369246 & 0.4613783 & 0.2671787 & 0.8100173 & 0.5173470 & 0.2156289 & 0.1919378 & -10609.058 & 282.6765666 & 3.6240151 \\
2 & -0.7977019 & 0.4654390 & 0.2169392 & 0.7846817 & 0.4905332 & 0.2164395 & 0.2049523 & -10607.253 & 26.7522215 & -1.6194299 \\
3 & -0.4455841 & 0.3884721 & 0.3213867 & 0.7758150 & 0.5193395 & 0.2340187 & 0.1883953 & -10607.001 & 50.1355291 & 3.0916011 \\
4 & -0.7634445 & 0.4521878 & 0.2217702 & 0.7935129 & 0.4959371 & 0.2218253 & 0.2055181 & -10607.210 & 4.8235882 & -2.6390063 \\
5 & -0.4906746 & 0.4146531 & 0.3404176 & 0.7722729 & 0.5222110 & 0.2269202 & 0.1846457 & -10607.059 & 67.6646338 & 9.0971871 \\
6 & -0.5515834 & 0.4434827 & 0.3335905 & 0.7724484 & 0.5190619 & 0.2197566 & 0.1853686 & -10607.022 & 17.4476962 & -0.4021102 \\
7 & -0.4914586 & 0.4327714 & 0.3503012 & 0.7686883 & 0.5235361 & 0.2216450 & 0.1826269 & -10606.991 & 16.2838831 & -0.1781480 \\
8 & -0.2900908 & 0.3885350 & 0.3956186 & 0.7563357 & 0.5370260 & 0.2300772 & 0.1754994 & -10606.864 & 12.0116477 & -2.4090216 \\
9 & -0.2752698 & 0.3832660 & 0.3969704 & 0.7555456 & 0.5377367 & 0.2312224 & 0.1753571 & -10606.847 & 11.4457840 & -2.5487401 \\
10 & -0.2609339 & 0.3780400 & 0.3982456 & 0.7547812 & 0.5384209 & 0.2323632 & 0.1752258 & -10606.832 & 10.8628213 & -2.6874876 \\
11 & -0.2085409 & 0.3576927 & 0.4025762 & 0.7519966 & 0.5408864 & 0.2368544 & 0.1748113 & -10606.782 & 8.3600783 & -3.4356438 \\
12 & -0.1970109 & 0.3528575 & 0.4034002 & 0.7513923 & 0.5414138 & 0.2379362 & 0.1747455 & -10606.772 & 7.6818408 & -3.6954428 \\
13 & -0.1761733 & 0.3436416 & 0.4046818 & 0.7503191 & 0.5423414 & 0.2400174 & 0.1746675 & -10606.756 & 6.2516380 & -4.2766527 \\
14 & -0.1668421 & 0.3392794 & 0.4051522 & 0.7498492 & 0.5427438 & 0.2410120 & 0.1746529 & -10606.750 & 5.5002876 & -4.5807361 \\
15 & -0.1581860 & 0.3350854 & 0.4055256 & 0.7494209 & 0.5431090 & 0.2419740 & 0.1746517 & -10606.745 & 4.7262048 & -4.8824015 \\
16 & -0.1501600 & 0.3310612 & 0.4058166 & 0.7490311 & 0.5434404 & 0.2429024 & 0.1746615 & -10606.741 & 3.9306487 & -5.1742197 \\
17 & -0.1209376 & 0.3159301 & 0.4066945 & 0.7476393 & 0.5446224 & 0.2464122 & 0.1747326 & -10606.734 & 2.8592342 & -6.2251311 \\
18 & -0.1216487 & 0.3155707 & 0.4064438 & 0.7477179 & 0.5445608 & 0.2465247 & 0.1747753 & -10606.734 & 0.0014787 & -6.2014232 \\
19 & -0.1215714 & 0.3155483 & 0.4064635 & 0.7477142 & 0.5445652 & 0.2465296 & 0.1747719 & -10606.734 & 1.82E-06 & -6.2026532 \\
20 & -0.1215714 & 0.3155483 & 0.4064635 & 0.7477142 & 0.5445652 & 0.2465296 & 0.1747719 & -10606.734 & 9.80E-10 & -6.2026530\\ \hline
\end{tabular}
}
\vspace{-0.3cm}
\end{table*}

\begin{table*}[ht]
\vspace{-0.3cm}
\centering
\caption{Convergence of the GTS parameter for Ethereum return data}
\label{taba2}
\resizebox{13cm}{!}{
\begin{tabular}{ccccccccccc}
\hline
\textbf{$Iterations$} & \textbf{$\mu$} & \textbf{$\beta_{+}$} & \textbf{$\beta_{-}$} & \textbf{$\alpha_{+}$} & \textbf{$\alpha_{-}$} & \textbf{$\lambda_{+}$} & \textbf{$\lambda_{-}$} & \textbf{$Log(ML)$} & \textbf{$||\frac{dLog(ML)}{dV}||$} & \textbf{$Max Eigen Value$} \\ \hline
1 & -0.1215714 & 0.3155483 & 0.4064635 & 0.7477142 & 0.5445652 & 0.2465296 & 0.1747719 & -9745.171 & 2673.428257 & 206.013602 \\
2 & -0.1724835 & 0.3319505 & 0.4091022 & 0.7364129 & 0.5479934 & 0.2227870 & 0.1684568 & -9700.715 & 2388.609394 & 180.884105 \\
3 & -0.2041418 & 0.3384742 & 0.4118929 & 0.7338794 & 0.5531083 & 0.2083203 & 0.1632896 & -9669.986 & 2139.267660 & 157.699659 \\
4 & -0.4006157 & 0.3530035 & 0.4393474 & 0.7513784 & 0.6172425 & 0.1135743 & 0.1221930 & -9586.115 & 1471.570475 & 32.410140 \\
5 & -0.6485551 & 0.4493817 & 0.4404508 & 0.9247887 & 0.7210031 & 0.1412949 & 0.1482307 & -9556.026 & 380.605737 & 56.584055 \\
6 & -0.6290525 & 0.4371402 & 0.4359516 & 0.9780784 & 0.7824777 & 0.1582340 & 0.1608694 & -9553.005 & 24.905322 & -0.719221 \\
7 & -0.5545412 & 0.3994778 & 0.3918188 & 0.9627486 & 0.7936571 & 0.1652438 & 0.1724287 & -9552.866 & 5.834338 & -0.847574 \\
8 & -0.4744837 & 0.3913982 & 0.4093404 & 0.9582366 & 0.8022858 & 0.1665103 & 0.1699928 & -9552.862 & 2.963350 & -0.933466 \\
9 & -0.4825586 & 0.3902160 & 0.4051365 & 0.9580755 & 0.8007651 & 0.1667400 & 0.1706850 & -9552.862 & 0.214871 & -0.931142 \\
10 & -0.4853678 & 0.3904369 & 0.4044899 & 0.9582486 & 0.8004799 & 0.1667119 & 0.1707853 & -9552.862 & 0.004754 & -0.931872 \\
11 & -0.4853800 & 0.3904362 & 0.4044846 & 0.9582487 & 0.8004779 & 0.1667121 & 0.1707862 & -9552.862 & 2.96E-07 & -0.931836 \\
12 & -0.4853800 & 0.3904362 & 0.4044846 & 0.9582487 & 0.8004779 & 0.1667121 & 0.1707862 & -9552.862 & 1.18E-10 & -0.931836 \\
13 & -0.4853800 & 0.3904362 & 0.4044846 & 0.9582487 & 0.8004779 & 0.1667121 & 0.1707862 & -9552.862 & 1.27E-11 & -0.931836\\ \hline
\end{tabular}
}
\vspace{-0.3cm}
\end{table*}

\begin{table*}[ht]
\vspace{-0.3cm}
\centering
\caption{Convergence of the GTS parameter for S\&P 500 return data}
\label{taba3}
\resizebox{13cm}{!}{%
\begin{tabular}{ccccccccccc}
\hline
\textbf{$Iterations$} & \textbf{$\mu$} & \textbf{$\beta_{+}$} & \textbf{$\beta_{-}$} & \textbf{$\alpha_{+}$} & \textbf{$\alpha_{-}$} & \textbf{$\lambda_{+}$} & \textbf{$\lambda_{-}$} & \textbf{$Log(ML)$} & \textbf{$||\frac{dLog(ML)}{dV}||$} & \textbf{$Max Eigen Value$} \\ \hline
1 & -0.2606426 & 0.34087979 & 0.02221141 & 0.78775729 & 0.59711061 & 1.28855513 & 1.01435308 & -4921.0858 & 147.214541 & -0.476265 \\
2 & -0.2747887 & 0.37848567 & 0.02517846 & 0.72538248 & 0.594628 & 1.22107935 & 1.01081205 & -4920.9765 & 107.910271 & -12.169518 \\
3 & -0.2852743 & 0.34562742 & 0.01628972 & 0.78353361 & 0.58024658 & 1.27423544 & 0.9888729 & -4920.6236 & 23.70873 & 11.9588258 \\
4 & -0.2971254 & 0.37985815 & 0.05392593 & 0.74068472 & 0.55972179 & 1.22737986 & 0.96278568 & -4920.5493 & 4.21443356 & 0.29705471 \\
5 & -0.3415082 & 0.42600675 & 0.0432239 & 0.69783497 & 0.56106365 & 1.18286494 & 0.966753 & -4920.5722 & 37.0642417 & -1.7903876 \\
6 & -0.2995817 & 0.40315129 & 0.12236507 & 0.7168274 & 0.522172 & 1.20383351 & 0.9117451 & -4920.574 & 3.07232514 & -0.7101089 \\
7 & -0.2944623 & 0.3977257 & 0.12218751 & 0.72174351 & 0.52260032 & 1.20899714 & 0.9121201 & -4920.5701 & 2.63567879 & -1.0187469 \\
8 & -0.2767429 & 0.37561063 & 0.11561097 & 0.7427615 & 0.52696799 & 1.23067384 & 0.91742165 & -4920.5511 & 1.83311761 & -2.1103436 \\
9 & -0.274204 & 0.37177939 & 0.11355883 & 0.74659763 & 0.52814335 & 1.2345524 & 0.91893546 & -4920.5477 & 1.75839181 & -2.177405 \\
10 & -0.2559812 & 0.34147926 & 0.09643312 & 0.77815581 & 0.53784221 & 1.26594249 & 0.93144448 & -4920.5308 & 1.33811298 & -2.6954121 \\
11 & -0.2496977 & 0.32954013 & 0.08928069 & 0.79125494 & 0.54186044 & 1.27868642 & 0.93662846 & -4920.5291 & 0.79520373 & -2.8166517 \\
12 & -0.2494237 & 0.32866495 & 0.08869445 & 0.79238161 & 0.54221561 & 1.27970094 & 0.93708759 & -4920.5291 & 0.00166731 & -2.6765739 \\
13 & -0.2494072 & 0.32862462 & 0.08864569 & 0.79242579 & 0.54224632 & 1.27974278 & 0.93712865 & -4920.5291 & 0.00013552 & -2.6768326 \\
14 & -0.2494082 & 0.32862428 & 0.08864047 & 0.79242619 & 0.54224944 & 1.27974312 & 0.93713293 & -4920.5291 & 1.47E-05 & -2.6766945 \\
15 & -0.2494083 & 0.32862424 & 0.08863992 & 0.79242624 & 0.54224977 & 1.27974315 & 0.93713338 & -4920.5291 & 1.57E-06 & -2.67668 \\
16 & -0.2494083 & 0.32862424 & 0.08863985 & 0.79242624 & 0.54224981 & 1.27974316 & 0.93713344 & -4920.5291 & 1.89E-09 & -2.6766783 \\
17 & -0.2494083 & 0.32862424 & 0.08863985 & 0.79242624 & 0.54224981 & 1.27974316 & 0.93713344 & -4920.5291 & 2.09E-10 & -2.6766783 \\ \hline
\end{tabular}
}
\end{table*}

\begin{table*}[ht]
\vspace{-0.3cm}
\centering
\caption{Convergence of the GTS parameter for SPY EFT return data}
\label{taba4}
\resizebox{13cm}{!}{
\begin{tabular}{ccccccccccc}
\hline
\textbf{$Iterations$} & \textbf{$\mu$} & \textbf{$\beta_{+}$} & \textbf{$\beta_{-}$} & \textbf{$\alpha_{+}$} & \textbf{$\alpha_{-}$} & \textbf{$\lambda_{+}$} & \textbf{$\lambda_{-}$} & \textbf{$Log(ML)$} & \textbf{$||\frac{dLog(ML)}{dV}||$} & \textbf{$Max Eigen Value$} \\ \hline
1 & -0.0518661 & 0.1161846 & 0.2186548 & 1.04269292 & 0.52712574 & 1.52244991 & 0.91168779 & -4894.2279 & 14.7801725 & -6.5141947 \\
2 & -0.1102477 & 0.18491276 & 0.17478472 & 0.94756655 & 0.52844271 & 1.4399315 & 0.91415148 & -4893.8278 & 29.8166141 & -1.9290981 \\
3 & -0.2094204 & 0.29377592 & 0.0891446 & 0.84029122 & 0.56054563 & 1.34797271 & 0.96500981 & -4893.3554 & 16.9940095 & 4.33892902 \\
4 & -0.1985564 & 0.29758208 & 0.13230013 & 0.83156167 & 0.53656079 & 1.33833078 & 0.93230856 & -4893.4206 & 10.9048744 & 1.04588745 \\
5 & -0.078883 & 0.25939922 & 0.39611543 & 0.84865673 & 0.40365522 & 1.35595932 & 0.7410936 & -4895.8806 & 241.028178 & 94.6293224 \\
6 & -0.0753571 & 0.26704857 & 0.33754158 & 0.84120908 & 0.45446164 & 1.3452823 & 0.80751063 & -4894.3899 & 25.1995505 & -2.805571 \\
7 & -0.196642 & 0.31624372 & 0.20068543 & 0.80509106 & 0.50322368 & 1.30837612 & 0.88967028 & -4893.888 & 140.257551 & 34.770691 \\
8 & -0.1898283 & 0.3045047 & 0.15900291 & 0.81380451 & 0.52672075 & 1.31341912 & 0.91775259 & -4893.4694 & 6.29433991 & -4.6080872 \\
9 & -0.2275214 & 0.32940996 & 0.10770535 & 0.79340215 & 0.55020025 & 1.29360449 & 0.95260474 & -4893.3049 & 7.34361008 & -8.1891832 \\
10 & -0.2726283 & 0.34972465 & 0.01601222 & 0.78061523 & 0.59736004 & 1.28153304 & 1.01757433 & -4893.2192 & 14.0784211 & -3.6408772 \\
11 & -0.2499816 & 0.32645286 & 0.01851524 & 0.80217301 & 0.60018154 & 1.30243672 & 1.01792703 & -4893.2125 & 6.27794755 & -4.6455546 \\
12 & -0.2575953 & 0.33832596 & 0.02643321 & 0.79008383 & 0.59450637 & 1.29085215 & 1.01101001 & -4893.208 & 1.23298227 & -6.8035318 \\
13 & -0.2607071 & 0.34052644 & 0.02075252 & 0.78817075 & 0.59805376 & 1.28895161 & 1.01555438 & -4893.2076 & 0.07363298 & -6.71708 \\
14 & -0.2606368 & 0.34088815 & 0.02227012 & 0.78774693 & 0.59707082 & 1.28854532 & 1.01430383 & -4893.2076 & 0.00156771 & -6.6908109 \\
15 & -0.2606432 & 0.34087911 & 0.02220633 & 0.78775813 & 0.59711397 & 1.28855593 & 1.01435731 & -4893.2076 & 0.00010164 & -6.6915902 \\
16 & -0.2606426 & 0.34087985 & 0.02221188 & 0.78775721 & 0.5971103 & 1.28855506 & 1.01435268 & -4893.2076 & 8.45E-06 & -6.6915177 \\
17 & -0.2606426 & 0.34087979 & 0.02221141 & 0.78775729 & 0.59711061 & 1.28855513 & 1.01435308 & -4893.2076 & 7.21E-07 & -6.6915239\\ \hline
\end{tabular}
}
\end{table*}
\newpage
\setcounter{equation}{0} 
\setcounter{table}{0} 
\section{Pearson Statistic Inputs}\label{eq:an2}
\begin{table}[ht]
\vspace{-0.9cm}
\centering
\caption{ Observed versus Expected  statistics under GTS distribution}
\label{tabb1}
\setlength\tabcolsep{3pt}  
\begin{tabular}{@{}c|ccc|ccc|ccc|ccc@{}}
\toprule
& \multicolumn{3}{c|}{Bitcoin} & \multicolumn{3}{c|}{Ethereum} & \multicolumn{3}{c|}{sp500} & \multicolumn{3}{c}{SPY   EFT} \\ \toprule
 \multicolumn{1}{c|}{k} & \multicolumn{1}{c|}{x(k)} & \multicolumn{1}{c|}{n*$\Pi_{k}$} & \multicolumn{1}{c|}{n(k)} & \multicolumn{1}{c|}{x(k)} & \multicolumn{1}{c|}{n*$\Pi_{k}$} & \multicolumn{1}{c|}{n(k)} & \multicolumn{1}{c|}{x(k)} & \multicolumn{1}{c|}{n*$\Pi_{k}$} & \multicolumn{1}{c|}{n(k)} & \multicolumn{1}{c|}{x(k)} & \multicolumn{1}{c|}{n*$\Pi_{k}$} & \multicolumn{1}{c}{n(k)} \\ \midrule
 
\multirow{1}{*}{1} & \multirow{1}{*}{-18.988} & \multirow{1}{*}{7.512} & \multirow{1}{*}{8} & \multirow{1}{*}{-20.861} & \multirow{1}{*}{7.531} & \multirow{1}{*}{6} & \multirow{1}{*}{-4.341} & \multirow{1}{*}{10.264} & \multirow{1}{*}{12} & \multirow{1}{*}{-4.405} & \multirow{1}{*}{8.327} & \multirow{1}{*}{11} \\ \midrule
\multirow{1}{*}{2} & \multirow{1}{*}{-17.080} & \multirow{1}{*}{4.144} & \multirow{1}{*}{7} & \multirow{1}{*}{-18.321} & \multirow{1}{*}{5.583} & \multirow{1}{*}{6} & \multirow{1}{*}{-3.935} & \multirow{1}{*}{5.456} & \multirow{1}{*}{5} & \multirow{1}{*}{-4.007} & \multirow{1}{*}{4.562} & \multirow{1}{*}{3} \\ \midrule
\multirow{1}{*}{3} & \multirow{1}{*}{-15.172} & \multirow{1}{*}{6.603} & \multirow{1}{*}{7} & \multirow{1}{*}{-15.781} & \multirow{1}{*}{10.018} & \multirow{1}{*}{15} & \multirow{1}{*}{-3.529} & \multirow{1}{*}{8.442} & \multirow{1}{*}{7} & \multirow{1}{*}{-3.608} & \multirow{1}{*}{7.116} & \multirow{1}{*}{7} \\ \midrule
\multirow{1}{*}{4} & \multirow{1}{*}{-13.264} & \multirow{1}{*}{10.678} & \multirow{1}{*}{9} & \multirow{1}{*}{-13.241} & \multirow{1}{*}{18.331} & \multirow{1}{*}{20} & \multirow{1}{*}{-3.123} & \multirow{1}{*}{13.138} & \multirow{1}{*}{15} & \multirow{1}{*}{-3.210} & \multirow{1}{*}{11.144} & \multirow{1}{*}{12} \\ \midrule
\multirow{1}{*}{5} & \multirow{1}{*}{-11.356} & \multirow{1}{*}{17.586} & \multirow{1}{*}{13} & \multirow{1}{*}{-10.700} & \multirow{1}{*}{34.424} & \multirow{1}{*}{29} & \multirow{1}{*}{-2.717} & \multirow{1}{*}{20.588} & \multirow{1}{*}{20} & \multirow{1}{*}{-2.811} & \multirow{1}{*}{17.538} & \multirow{1}{*}{18} \\ \midrule
\multirow{1}{*}{6} & \multirow{1}{*}{-9.448} & \multirow{1}{*}{29.661} & \multirow{1}{*}{32} & \multirow{1}{*}{-8.160} & \multirow{1}{*}{66.980} & \multirow{1}{*}{68} & \multirow{1}{*}{-2.311} & \multirow{1}{*}{32.543} & \multirow{1}{*}{30} & \multirow{1}{*}{-2.413} & \multirow{1}{*}{27.775} & \multirow{1}{*}{27} \\ \midrule
\multirow{1}{*}{7} & \multirow{1}{*}{-7.540} & \multirow{1}{*}{51.657} & \multirow{1}{*}{47} & \multirow{1}{*}{-5.620} & \multirow{1}{*}{137.268} & \multirow{1}{*}{134} & \multirow{1}{*}{-1.905} & \multirow{1}{*}{52.023} & \multirow{1}{*}{48} & \multirow{1}{*}{-2.015} & \multirow{1}{*}{44.350} & \multirow{1}{*}{41} \\ \midrule
\multirow{1}{*}{8} & \multirow{1}{*}{-5.632} & \multirow{1}{*}{94.188} & \multirow{1}{*}{107} & \multirow{1}{*}{-3.080} & \multirow{1}{*}{305.591} & \multirow{1}{*}{305} & \multirow{1}{*}{-1.499} & \multirow{1}{*}{84.479} & \multirow{1}{*}{89} & \multirow{1}{*}{-1.616} & \multirow{1}{*}{71.617} & \multirow{1}{*}{73} \\ \midrule
\multirow{1}{*}{9} & \multirow{1}{*}{-3.724} & \multirow{1}{*}{184.486} & \multirow{1}{*}{168} & \multirow{1}{*}{-0.540} & \multirow{1}{*}{769.951} & \multirow{1}{*}{769} & \multirow{1}{*}{-1.093} & \multirow{1}{*}{140.406} & \multirow{1}{*}{147} & \multirow{1}{*}{-1.218} & \multirow{1}{*}{117.552} & \multirow{1}{*}{122} \\ \midrule

\multirow{1}{*}{10} & \multirow{1}{*}{-1.816} & \multirow{1}{*}{411.503} & \multirow{1}{*}{419} & \multirow{1}{*}{2.000} & \multirow{1}{*}{965.210} & \multirow{1}{*}{966} & \multirow{1}{*}{-0.687} & \multirow{1}{*}{242.564} & \multirow{1}{*}{244} & \multirow{1}{*}{-0.819} & \multirow{1}{*}{198.101} & \multirow{1}{*}{186} \\ \midrule

\multirow{1}{*}{11} & \multirow{1}{*}{0.092} & \multirow{1}{*}{1195.725} & \multirow{1}{*}{1186} & \multirow{1}{*}{4.540} & \multirow{1}{*}{458.955} & \multirow{1}{*}{466} & \multirow{1}{*}{-0.281} & \multirow{1}{*}{455.971} & \multirow{1}{*}{456} & \multirow{1}{*}{-0.421} & \multirow{1}{*}{351.660} & \multirow{1}{*}{348} \\ \midrule
\multirow{1}{*}{12} & \multirow{1}{*}{2.000} & \multirow{1}{*}{1150.470} & \multirow{1}{*}{1159} & \multirow{1}{*}{7.080} & \multirow{1}{*}{219.873} & \multirow{1}{*}{222} & \multirow{1}{*}{0.126} & \multirow{1}{*}{896.809} & \multirow{1}{*}{896} & \multirow{1}{*}{-0.023} & \multirow{1}{*}{725.476} & \multirow{1}{*}{730} \\ \midrule
\multirow{1}{*}{13} & \multirow{1}{*}{3.908} & \multirow{1}{*}{473.177} & \multirow{1}{*}{469} & \multirow{1}{*}{9.620} & \multirow{1}{*}{111.760} & \multirow{1}{*}{101} & \multirow{1}{*}{0.532} & \multirow{1}{*}{749.106} & \multirow{1}{*}{733} & \multirow{1}{*}{0.376} & \multirow{1}{*}{867.735} & \multirow{1}{*}{867} \\ \midrule
\multirow{1}{*}{14} & \multirow{1}{*}{5.816} & \multirow{1}{*}{217.783} & \multirow{1}{*}{227} & \multirow{1}{*}{12.160} & \multirow{1}{*}{59.253} & \multirow{1}{*}{60} & \multirow{1}{*}{0.938} & \multirow{1}{*}{430.841} & \multirow{1}{*}{426} & \multirow{1}{*}{0.774} & \multirow{1}{*}{541.022} & \multirow{1}{*}{522} \\ \midrule
\multirow{1}{*}{15} & \multirow{1}{*}{7.724} & \multirow{1}{*}{107.387} & \multirow{1}{*}{102} & \multirow{1}{*}{14.700} & \multirow{1}{*}{32.379} & \multirow{1}{*}{32} & \multirow{1}{*}{1.344} & \multirow{1}{*}{234.692} & \multirow{1}{*}{260} & \multirow{1}{*}{1.173} & \multirow{1}{*}{300.464} & \multirow{1}{*}{325} \\ \midrule
\multirow{1}{*}{16} & \multirow{1}{*}{9.632} & \multirow{1}{*}{55.272} & \multirow{1}{*}{51} & \multirow{1}{*}{17.241} & \multirow{1}{*}{18.099} & \multirow{1}{*}{21} & \multirow{1}{*}{1.750} & \multirow{1}{*}{126.820} & \multirow{1}{*}{137} & \multirow{1}{*}{1.571} & \multirow{1}{*}{163.491} & \multirow{1}{*}{189} \\ \midrule
\multirow{1}{*}{17} & \multirow{1}{*}{11.540} & \multirow{1}{*}{29.294} & \multirow{1}{*}{39} & \multirow{1}{*}{19.781} & \multirow{1}{*}{10.296} & \multirow{1}{*}{12} & \multirow{1}{*}{2.156} & \multirow{1}{*}{68.688} & \multirow{1}{*}{57} & \multirow{1}{*}{1.969} & \multirow{1}{*}{88.862} & \multirow{1}{*}{82} \\ \midrule
\multirow{1}{*}{18} & \multirow{1}{*}{13.448} & \multirow{1}{*}{15.861} & \multirow{1}{*}{14} & \multirow{1}{*}{22.321} & \multirow{1}{*}{5.939} & \multirow{1}{*}{5} & \multirow{1}{*}{2.562} & \multirow{1}{*}{37.387} & \multirow{1}{*}{33} & \multirow{1}{*}{2.368} & \multirow{1}{*}{48.491} & \multirow{1}{*}{36} \\ \midrule
\multirow{1}{*}{19} & \multirow{1}{*}{15.356} & \multirow{1}{*}{8.729} & \multirow{1}{*}{9} & \multirow{1}{*}{24.861} & \multirow{1}{*}{3.465} & \multirow{1}{*}{5} & \multirow{1}{*}{2.968} & \multirow{1}{*}{20.460} & \multirow{1}{*}{15} & \multirow{1}{*}{2.766} & \multirow{1}{*}{26.600} & \multirow{1}{*}{23} \\ \midrule
\multirow{1}{*}{20} & \multirow{1}{*}{17.264} & \multirow{1}{*}{4.866} & \multirow{1}{*}{4} & \multirow{1}{*}{} & \multirow{1}{*}{5.091} & \multirow{1}{*}{4} & \multirow{1}{*}{3.374} & \multirow{1}{*}{11.256} & \multirow{1}{*}{12} & \multirow{1}{*}{3.165} & \multirow{1}{*}{14.669} & \multirow{1}{*}{13} \\ \midrule
\multirow{1}{*}{21} & \multirow{1}{*}{} & \multirow{1}{*}{6.419} & \multirow{1}{*}{6} & \multirow{1}{*}{} & \multirow{1}{*}{} & \multirow{1}{*}{} & \multirow{1}{*}{} & \multirow{1}{*}{14.067} & \multirow{1}{*}{14} & \multirow{1}{*}{} & \multirow{1}{*}{18.450} & \multirow{1}{*}{20} \\ \bottomrule
\end{tabular}%
\vspace{-0.5cm}
\end{table}

\clearpage
\newpage
\setcounter{equation}{0}
\setcounter{table}{0}
\section{Bitcoin BTC : Kobol distribution ($\beta=\beta_{-}=\beta_{+}$)}\label{eq:an3}
\begin{equation}
 \begin{aligned}
V(dx) =\left(\alpha_{+}\frac{e^{-\lambda_{+}x}}{x^{1+\beta}} \boldsymbol{1}_{x>0} + \alpha_{-}\frac{e^{-\lambda_{-}|x|}}{|x|^{1+\beta}} \boldsymbol{1}_{x<0}\right) dx
 \end{aligned}
 \end{equation}

  \begin{equation}
\begin{aligned}
\resizebox{0.9\hsize}{!}{$\Psi(\xi)=\mu\xi i+\Gamma(-\beta)\left[\alpha_{+}((\lambda_{+} - i\xi)^{\beta}-{\lambda_{+}}^{\beta})+\alpha_{-}((\lambda_{-}+i\xi)^{\beta}-{\lambda_{-}}^{\beta})\right]$} \label{eq:l27}
 \end{aligned}
 \end{equation}

  \begin{table}[ht]
\centering
\caption{ Kobol maximum-likelihood estimation for Bitcoin return data}
\label{tabc1}
\begin{tabular}{@{} c|c|c|ccccc @{}}
\toprule
\multicolumn{1}{c|}{\textbf{Model}} & \multicolumn{1}{c|}{\textbf{Parameter}} & \multicolumn{1}{c|}{\textbf{Estimate}} & \multicolumn{1}{c|}{\textbf{Std Err}} & \multicolumn{1}{c|}{\textbf{z}} & \multicolumn{1}{c|}{\textbf{$Pr(Z >|z|)$}} & \multicolumn{2}{c}{\textbf{[95\% Conf.Interval]}}  \\ \toprule

\multirow{10}{*}{\textbf{GTS}} & \multirow{1}{*}{\textbf{$\mu$}} &\multirow{1}{*}{-0.292833} & \multirow{1}{*}{(0.126)} & \multirow{1}{*}{-2.32} & \multirow{1}{*}{2.1E-02} & \multirow{1}{*}{-0.541} & \multirow{1}{*}{-0.045}  \\

 & \multirow{1}{*}{\textbf{$\beta$}} & \multirow{1}{*}{0.367074} & \multirow{1}{*}{(0.086)} & \multirow{1}{*}{4.27} & \multirow{1}{*}{1.9E-05} & \multirow{1}{*}{0.199} & \multirow{1}{*}{0.535}  \\

 & \multirow{1}{*}{\textbf{$\alpha_{+}$}} & \multirow{1}{*}{0.755914} & \multirow{1}{*}{(0.047)} & \multirow{1}{*}{16.02} & \multirow{1}{*}{4.7E-58} & \multirow{1}{*}{0.663} & \multirow{1}{*}{0.848}  \\

 & \multirow{1}{*}{\textbf{$\alpha_{-}$}} & \multirow{1}{*}{0.535121} & \multirow{1}{*}{(0.034)} & \multirow{1}{*}{15.68} & \multirow{1}{*}{9.6E-56} & \multirow{1}{*}{0.468} & \multirow{1}{*}{0.602}  \\

 &\multirow{1}{*}{\textbf{$\lambda_{+}$}} & \multirow{1}{*}{0.235266} & \multirow{1}{*}{(0.027)} & \multirow{1}{*}{8.87} & \multirow{1}{*}{3.6E-19} & \multirow{1}{*}{0.183} & \multirow{1}{*}{0.287}  \\

 & \multirow{1}{*}{\textbf{$\lambda_{-}$}} & \multirow{1}{*}{0.181602} & \multirow{1}{*}{(0.023)} & \multirow{1}{*}{7.94} & \multirow{1}{*}{9.8E-16} & \multirow{1}{*}{0.137} & \multirow{1}{*}{0.226}  \\ \cmidrule(l){2-8}
 & \multirow{1}{*}{\textbf{Log(ML)}} & \multirow{1}{*}{-10607} & \multirow{1}{*}{} & \multirow{1}{*}{} & \multirow{1}{*}{} & \multirow{1}{*}{} & \multirow{1}{*}{}  \\
 & \multirow{1}{*}{\textbf{AIC}} & \multirow{1}{*}{21226} & \multirow{1}{*}{} & \multirow{1}{*}{} & \multirow{1}{*}{} & \multirow{1}{*}{} &  \multirow{1}{*}{}  \\
 & \multirow{1}{*}{\textbf{BIK}} & \multirow{1}{*}{21264} & \multirow{1}{*}{} & \multirow{1}{*}{} & \multirow{1}{*}{} & \multirow{1}{*}{} &  \multirow{1}{*}{} \\  \bottomrule
\end{tabular}
\end{table}
\begin{table*}[ht]
\vspace{-0.3cm}
\centering
\caption{Convergence of the Kobol parameter for Bitcoin return data}
\label{tabc2}
\resizebox{13cm}{!}{%
\begin{tabular}{cccccccccc}
\hline
\textbf{$Iterations$} & \textbf{$\mu$} & \textbf{$\beta$} & \textbf{$\alpha_{+}$} & \textbf{$\alpha_{-}$} & \textbf{$\lambda_{+}$} & \textbf{$\lambda_{-}$} & \textbf{$Log(ML)$} & \textbf{$||\frac{dLog(ML)}{dV}||$} & \textbf{$Max Eigen Value$} \\ \hline
1 & -0.1215714 & 0.3155483 & 0.7477142 & 0.5445652 & 0.2465296 & 0.17477186 & -10614.93879 & 450.0556974 & 25.6678081 \\
2 & -0.255172 & 0.36516958 & 0.73215119 & 0.53194253 & 0.22955186 & 0.17909072 & -10607.01058 & 51.74982347 & -46.893383 \\
3 & -0.2912276 & 0.37096854 & 0.75410108 & 0.53529439 & 0.23387716 & 0.18070591 & -10606.81236 & 1.484798964 & -53.728563 \\
4 & -0.2922408 & 0.36574333 & 0.7559582 & 0.53508403 & 0.23560819 & 0.18189913 & -10606.81041 & 0.258928464 & -53.391237 \\
5 & -0.2928641 & 0.36714311 & 0.75591147 & 0.53512239 & 0.23524801 & 0.18158588 & -10606.81025 & 0.01286122 & -53.406734 \\
6 & -0.292837 & 0.36708319 & 0.75591382 & 0.53512107 & 0.23526357 & 0.18159941 & -10606.81025 & 0.00174219 & -53.40643 \\
7 & -0.2928328 & 0.36707373 & 0.75591419 & 0.53512086 & 0.23526603 & 0.18160154 & -10606.81025 & 1.18E-05 & -53.406384 \\
8 & -0.2928328 & 0.36707379 & 0.75591419 & 0.53512086 & 0.23526602 & 0.18160153 & -10606.81025 & 1.60E-06 & -53.406384 \\
9 & -0.2928328 & 0.36707379 & 0.75591419 & 0.53512086 & 0.23526601 & 0.18160153 & -10606.81025 & 2.18E-07 & -53.406384 \\
10 & -0.2928328 & 0.36707379 & 0.75591419 & 0.53512086 & 0.23526601 & 0.18160153 & -10606.81025 & 1.09E-08 & -53.406384\\ \hline
\end{tabular}
}
\end{table*}

\clearpage
   \newpage
 \setcounter{equation}{0}
\setcounter{table}{0}
\section{Ethereum: Carr-Geman-Madan-Yor (CGMY) distributions}\label{eq:an4}
\begin{equation}
 \begin{aligned}
V(dx) =\left(\alpha\frac{e^{-\lambda_{+} x}}{x^{1+\beta}} \boldsymbol{1}_{x>0} + \alpha\frac{e^{-\lambda_{-}|x|}}{|x|^{1+\beta}} \boldsymbol{1}_{x<0}\right) dx
 \end{aligned}
 \end{equation}

  \begin{equation}
\begin{aligned}
\resizebox{0.9\hsize}{!}{$\Psi(\xi)=\mu\xi i+\alpha\Gamma(-\beta)\left[((\lambda_{+} - i\xi)^{\beta}-{\lambda_{+}}^{\beta})+((\lambda_{-}+i\xi)^{\beta}-{\lambda_{-}}^{\beta})\right]$} \label{eq:l27}
 \end{aligned}
 \end{equation}

  \begin{table}[ht]
\centering
\caption{ CGMY maximum-likelihood estimation for Ethereum return data}
\label{tabd1}
\begin{tabular}{@{} c|c|c|ccccc @{}}
\toprule
\multicolumn{1}{c|}{\textbf{Model}} & \multicolumn{1}{c|}{\textbf{Parameter}} & \multicolumn{1}{c|}{\textbf{Estimate}} & \multicolumn{1}{c|}{\textbf{Std Err}} & \multicolumn{1}{c|}{\textbf{z}} & \multicolumn{1}{c|}{\textbf{$Pr(Z >|z|)$}} & \multicolumn{2}{c}{\textbf{[95\% Conf.Interval]}}  \\ \toprule

\multirow{10}{*}{\textbf{GTS}} & \multirow{1}{*}{\textbf{$\mu$}} &\multirow{1}{*}{-0.147089} & \multirow{1}{*}{(0.079)} & \multirow{1}{*}{-1.86} & \multirow{1}{*}{6.3E-02} & \multirow{1}{*}{-0.302} & \multirow{1}{*}{-0.008}  \\

 & \multirow{1}{*}{\textbf{$\beta$}} & \multirow{1}{*}{0.398418} & \multirow{1}{*}{(0.127)} & \multirow{1}{*}{3.12} & \multirow{1}{*}{1.8E-03} & \multirow{1}{*}{0.148} & \multirow{1}{*}{0.649}  \\

 &\multirow{1}{*}{ \textbf{$\alpha$}} & \multirow{1}{*}{0.887161} & \multirow{1}{*}{(0.058)} & \multirow{1}{*}{15.22} & \multirow{1}{*}{1.2E-52} & \multirow{1}{*}{0.773} & \multirow{1}{*}{1.001}  \\

 & \multirow{1}{*}{\textbf{$\lambda_{+}$}} & \multirow{1}{*}{0.155369} & \multirow{1}{*}{(0.023)} & \multirow{1}{*}{6.56} & \multirow{1}{*}{5.2E-11} & \multirow{1}{*}{0.109} & \multirow{1}{*}{0.202}  \\

 & \multirow{1}{*}{\textbf{$\lambda_{-}$}} & \multirow{1}{*}{0.185991} & \multirow{1}{*}{(0.025)} & \multirow{1}{*}{7.29} & \multirow{1}{*}{2.9E-13} & \multirow{1}{*}{0.136} & \multirow{1}{*}{0.236}  \\ \cmidrule(l){2-8}

  & \multirow{1}{*}{\textbf{Log(ML)}} & \multirow{1}{*}{-9554} & \multirow{1}{*}{} & \multirow{1}{*}{} & \multirow{1}{*}{} & \multirow{1}{*}{} & \multirow{1}{*}{}  \\
 & \multirow{1}{*}{\textbf{AIC}} & \multirow{1}{*}{19118} & \multirow{1}{*}{} & \multirow{1}{*}{} & \multirow{1}{*}{} & \multirow{1}{*}{} &  \multirow{1}{*}{}  \\
 & \multirow{1}{*}{\textbf{BIK}} & \multirow{1}{*}{19149} & \multirow{1}{*}{} & \multirow{1}{*}{} & \multirow{1}{*}{} & \multirow{1}{*}{} &  \multirow{1}{*}{} \\  \bottomrule
\end{tabular}
\end{table}

\begin{table*}[ht]
\vspace{-0.3cm}
\centering
\caption{Convergence of the CGMY parameter for Ethereum return data}
\label{tabd2}
\resizebox{13cm}{!}{%
\begin{tabular}{ccccccccc}
\hline
\textbf{$Iterations$} & \textbf{$\mu$} & \textbf{$\beta$}  & \textbf{$\alpha$} &\textbf{$\lambda_{+}$} & \textbf{$\lambda_{-}$} & \textbf{$Log(ML)$} & \textbf{$||\frac{dLog(ML)}{dV}||$} & \textbf{$Max Eigen Value$} \\ \hline
1 & -0.48538 & 0.39043616 & 0.95824875 & 0.16671208 & 0.17078617 & -9596.2658 & 1653.57149 & -140.21456 \\
2 & -0.0545131 & 0.40148247 & 0.88205674 & 0.15875317 & 0.18060554 & -9554.6834 & 80.0921993 & -19.637869 \\
3 & -0.1479632 & 0.39049434 & 0.88271998 & 0.15631408 & 0.18704084 & -9553.9065 & 3.84112398 & -44.76325 \\
4 & -0.1465893 & 0.40345482 & 0.88868927 & 0.15450383 & 0.18507239 & -9553.9036 & 0.4436942 & -55.029934 \\
5 & -0.1472464 & 0.39683622 & 0.88667597 & 0.15564017 & 0.18628001 & -9553.9026 & 0.14094418 & -51.274906 \\
6 & -0.1470247 & 0.39907668 & 0.88736036 & 0.15525581 & 0.18587143 & -9553.9025 & 0.05563606 & -52.523819 \\
7 & -0.1471148 & 0.39816841 & 0.88708569 & 0.15541227 & 0.18603772 & -9553.9025 & 0.02143506 & -52.017698 \\
8 & -0.1470898 & 0.39842098 & 0.88716234 & 0.15536883 & 0.18599155 & -9553.9025 & 0.00019543 & -52.158334 \\
9 & -0.14709 & 0.39841855 & 0.88716161 & 0.15536924 & 0.18599199 & -9553.9025 & 1.16E-05 & -52.156981 \\
10 & -0.14709 & 0.39841867 & 0.88716164 & 0.15536922 & 0.18599197 & -9553.9025 & 1.78E-06 & -52.157046 \\
11 & -0.14709 & 0.39841869 & 0.88716165 & 0.15536922 & 0.18599197 & -9553.9025 & 2.71E-07 & -52.157055 \\
12 & -0.14709 & 0.39841869 & 0.88716165 & 0.15536922 & 0.18599197 & -9553.9025 & 4.14E-08 & -52.157057\\ \hline
\end{tabular}
}
\end{table*}
\clearpage
  \newpage
 \setcounter{equation}{0}
\setcounter{table}{0}
 \section{S\&P 500 index:  Bilateral Gamma (BG) distribution ($\beta_{-}=\beta_{+}=0$)} \label{eq:an5}
\begin{equation}
 \begin{aligned}
V(dx) =\left(\alpha_{+} \frac{e^{-\lambda_{+} x}}{x} \boldsymbol{1}_{x>0} + \alpha_{-} \frac{e^{-\lambda_{-} |x|}}{|x|} \boldsymbol{1}_{x<0}\right) dx 
 \end{aligned}
 \end{equation}

 \begin{equation}
\begin{aligned}
\Psi(\xi)=\mu\xi i - \alpha_{+}\log\left(1-\frac{1}{\lambda_{+}}i\xi\right)- \alpha_{-}\log\left(1+\frac{1}{\lambda_{-}}i\xi\right) \label {eq:l27b}
 \end{aligned}
 \end{equation}

  \begin{table}[ht]
\centering
\caption{ BG maximum-likelihood estimation for S\&P 500 return data}
\label{tabe1}
\begin{tabular}{@{} c|c|c|ccccc @{}}
\toprule
\multicolumn{1}{c|}{\textbf{Model}} & \multicolumn{1}{c|}{\textbf{Parameter}} & \multicolumn{1}{c|}{\textbf{Estimate}} & \multicolumn{1}{c|}{\textbf{Std Err}} & \multicolumn{1}{c|}{\textbf{z}} & \multicolumn{1}{c|}{\textbf{$Pr(Z >|z|)$}} & \multicolumn{2}{c}{\textbf{[95\% Conf.Interval]}}  \\ \toprule

\multirow{10}{*}{\textbf{GTS}} & \multirow{1}{*}{\textbf{$\mu$}} &\multirow{1}{*}{-0.031467} & \multirow{1}{*}{(0.010)} & \multirow{1}{*}{-3.07} & \multirow{1}{*}{2.1E-03} & \multirow{1}{*}{-0.052} & \multirow{1}{*}{-0.011}  \\

 & \multirow{1}{*}{\textbf{$\alpha_{+}$}} & \multirow{1}{*}{1.092741} & \multirow{1}{*}{(0.058)} & \multirow{1}{*}{18.98} & \multirow{1}{*}{2.6E-80} & \multirow{1}{*}{0.980} & \multirow{1}{*}{1.206}  \\

 & \multirow{1}{*}{\textbf{$\alpha_{-}$}} & \multirow{1}{*}{0.701784} & \multirow{1}{*}{(0.042)} & \multirow{1}{*}{16.80} & \multirow{1}{*}{2.3E-63} & \multirow{1}{*}{0.620} & \multirow{1}{*}{0.784}  \\

 &\multirow{1}{*}{\textbf{$\lambda_{+}$}} & \multirow{1}{*}{1.539690} & \multirow{1}{*}{(0.064)} & \multirow{1}{*}{22.82} & \multirow{1}{*}{3.1E-115} & \multirow{1}{*}{1.407} & \multirow{1}{*}{1.672}  \\

 & \multirow{1}{*}{\textbf{$\lambda_{-}$}} & \multirow{1}{*}{1.110737} & \multirow{1}{*}{(0.050)} & \multirow{1}{*}{22.07} & \multirow{1}{*}{6.6E-108} & \multirow{1}{*}{1.012} & \multirow{1}{*}{1.209}  \\ \cmidrule(l){2-8}

  & \multirow{1}{*}{\textbf{Log(ML)}} & \multirow{1}{*}{-4925} & \multirow{1}{*}{} & \multirow{1}{*}{} & \multirow{1}{*}{} & \multirow{1}{*}{} & \multirow{1}{*}{}  \\
 & \multirow{1}{*}{\textbf{AIC}} & \multirow{1}{*}{9859} & \multirow{1}{*}{} & \multirow{1}{*}{} & \multirow{1}{*}{} & \multirow{1}{*}{} &  \multirow{1}{*}{}  \\
 & \multirow{1}{*}{\textbf{BIK}} & \multirow{1}{*}{9890} & \multirow{1}{*}{} & \multirow{1}{*}{} & \multirow{1}{*}{} & \multirow{1}{*}{} &  \multirow{1}{*}{} \\  \bottomrule
\end{tabular}
\end{table}

\begin{table*}[ht]
\vspace{-0.3cm}
\centering
\caption{ Convergence of the BG parameter for S\&P 500 return data}
\label{tabe2}
\resizebox{13cm}{!}{%
\begin{tabular}{ccccccccc}
\hline
\textbf{$Iterations$} & \textbf{$\mu$} & \textbf{$\alpha_{+}$} & \textbf{$\alpha_{-}$} & \textbf{$\lambda_{+}$} & \textbf{$\lambda_{-}$} & \textbf{$Log(ML)$} & \textbf{$||\frac{dLog(ML)}{dV}||$} & \textbf{$Max Eigen Value$} \\ \hline
1 & 0 & 0.79242624 & 0.54224981 & 1.27974316 & 0.93713344 & -4951.1439 & 1138.53458 & -265.251 \\
2 & -0.0038447 & 0.93153413 & 0.64461254 & 1.41132138 & 1.05504868 & -4931.7583 & 549.025405 & -171.22804 \\
3 & -0.0103214 & 1.03062846 & 0.70198868 & 1.49426667 & 1.10555794 & -4926.8412 & 286.215345 & -126.18156 \\
4 & -0.0186317 & 1.07922912 & 0.71421996 & 1.53377391 & 1.11392285 & -4925.393 & 135.694287 & -113.12071 \\
5 & -0.0279475 & 1.09450205 & 0.70493092 & 1.54418172 & 1.10795103 & -4924.7065 & 38.0551545 & -116.58686 \\
6 & -0.0313951 & 1.09325663 & 0.70162581 & 1.54038766 & 1.10996346 & -4924.621 & 1.54417452 & -120.06271 \\
7 & -0.0314671 & 1.09274119 & 0.70178365 & 1.53969027 & 1.11073682 & -4924.6205 & 0.02788236 & -120.35435 \\
8 & -0.0314664 & 1.09276971 & 0.70182788 & 1.53971127 & 1.11079928 & -4924.6205 & 0.00198685 & -120.34482 \\
9 & -0.0314663 & 1.09277213 & 0.70183158 & 1.5397131 & 1.11080431 & -4924.6205 & 0.00016039 & -120.34394 \\
10 & -0.0314662 & 1.09277232 & 0.70183188 & 1.53971325 & 1.11080472 & -4924.6205 & 1.29E-05 & -120.34387 \\
11 & -0.0314662 & 1.09277234 & 0.7018319 & 1.53971326 & 1.11080475 & -4924.6205 & 1.04E-06 & -120.34387 \\
12 & -0.0314662 & 1.09277234 & 0.7018319 & 1.53971326 & 1.11080476 & -4924.6205 & 8.43E-08 & -120.34387 \\
13 & -0.0314662 & 1.09277234 & 0.7018319 & 1.53971326 & 1.11080476 & -4924.6205 & 6.80E-09 & -120.34387 \\
14 & -0.0314662 & 1.09277234 & 0.7018319 & 1.53971326 & 1.11080476 & -4924.6205 & 5.63E-10 & -120.34387 \\
15 & -0.0314662 & 1.09277234 & 0.7018319 & 1.53971326 & 1.11080476 & -4924.6205 & 5.73E-11 & -120.34387 \\ \hline
\end{tabular}
}
\end{table*}

\newpage

\setcounter{equation}{0}
\setcounter{table}{0}
\section{SPY ETF:  Bilateral Gamma (BG) distribution ($\beta_{-}=\beta_{+}=0$)}\label{eq:an6}
\begin{equation}
 \begin{aligned}
V(dx) =\left(\alpha_{+} \frac{e^{-\lambda_{+} x}}{x} \boldsymbol{1}_{x>0} + \alpha_{-} \frac{e^{-\lambda_{-} |x|}}{|x|} \boldsymbol{1}_{x<0}\right) dx 
 \end{aligned}
 \end{equation}

 \begin{equation}
\begin{aligned}
\Psi(\xi)=\mu\xi i - \alpha_{+}\log\left(1-\frac{1}{\lambda_{+}}i\xi\right)- \alpha_{-}\log\left(1+\frac{1}{\lambda_{-}}i\xi\right) \label {eq:l27b}
 \end{aligned}
 \end{equation}

  \begin{table}[ht]
\centering
\caption{ BG maximum-likelihood estimation for SPY EFT return data}
\label{tabf1}
\begin{tabular}{@{} c|c|c|ccccc @{}}
\toprule
\multicolumn{1}{c|}{\textbf{Model}} & \multicolumn{1}{c|}{\textbf{Parameter}} & \multicolumn{1}{c|}{\textbf{Estimate}} & \multicolumn{1}{c|}{\textbf{Std Err}} & \multicolumn{1}{c|}{\textbf{z}} & \multicolumn{1}{c|}{\textbf{$Pr(Z >|z|)$}} & \multicolumn{2}{c}{\textbf{[95\% Conf.Interval]}}  \\ \toprule

\multirow{10}{*}{\textbf{GTS}} & \multirow{1}{*}{\textbf{$\mu$}} &\multirow{1}{*}{0.015048} & \multirow{1}{*}{(0.012)} & \multirow{1}{*}{1.28} & \multirow{1}{*}{2.0E-01} & \multirow{1}{*}{-0.008} & \multirow{1}{*}{0.038}  \\

 & \multirow{1}{*}{\textbf{$\alpha_{+}$}} & \multirow{1}{*}{1.068239} & \multirow{1}{*}{(0.067)} & \multirow{1}{*}{16.02} & \multirow{1}{*}{8.6E-58} & \multirow{1}{*}{0.938} & \multirow{1}{*}{1.199}  \\

 & \multirow{1}{*}{\textbf{$\alpha_{-}$}} & \multirow{1}{*}{0.764449} & \multirow{1}{*}{(0.044)} & \multirow{1}{*}{17.33} & \multirow{1}{*}{3.0E-67} & \multirow{1}{*}{0.678} & \multirow{1}{*}{0.851}  \\

 &\multirow{1}{*}{\textbf{$\lambda_{+}$}} & \multirow{1}{*}{1.525718} & \multirow{1}{*}{(0.073)} & \multirow{1}{*}{20.98} & \multirow{1}{*}{1.1E-97} & \multirow{1}{*}{1.383} & \multirow{1}{*}{1.668}  \\

 & \multirow{1}{*}{\textbf{$\lambda_{-}$}} & \multirow{1}{*}{1.156439} & \multirow{1}{*}{(0.052)} & \multirow{1}{*}{22.15} & \multirow{1}{*}{1.1E-108} & \multirow{1}{*}{1.054} & \multirow{1}{*}{1.259}  \\ \cmidrule(l){2-8}

  & \multirow{1}{*}{\textbf{Log(ML)}} & \multirow{1}{*}{-4899} & \multirow{1}{*}{} & \multirow{1}{*}{} & \multirow{1}{*}{} & \multirow{1}{*}{} & \multirow{1}{*}{}  \\
 & \multirow{1}{*}{\textbf{AIC}} & \multirow{1}{*}{9807} & \multirow{1}{*}{} & \multirow{1}{*}{} & \multirow{1}{*}{} & \multirow{1}{*}{} &  \multirow{1}{*}{}  \\
 & \multirow{1}{*}{\textbf{BIK}} & \multirow{1}{*}{9838} & \multirow{1}{*}{} & \multirow{1}{*}{} & \multirow{1}{*}{} & \multirow{1}{*}{} &  \multirow{1}{*}{} \\  \bottomrule
\end{tabular}
\end{table}

\begin{table*}[ht]
\vspace{-0.3cm}
\centering
\caption{Convergence of the BG parameter for SPY EFT return data}
\label{tabf2}
\resizebox{13cm}{!}{%
\begin{tabular}{ccccccccccc}
\hline
\textbf{$Iterations$} & \textbf{$\mu$} & \textbf{$\alpha_{+}$} & \textbf{$\alpha_{-}$} & \textbf{$\lambda_{+}$} & \textbf{$\lambda_{-}$} & \textbf{$Log(ML)$} & \textbf{$||\frac{dLog(ML)}{dV}||$} & \textbf{$Max Eigen Value$} \\ \hline
1 & 0 & 0.78775729 & 0.59711061 & 1.28855513 & 1.01435308 & -4918.7331 & 406.35365 & -252.28104 \\
2 & 0.02867773 & 0.97562263 & 0.67572827 & 1.46822249 & 0.97596762 & -4908.5992 & 226.190986 & -116.11753 \\
3 & 0.02727407 & 1.05127517 & 0.78618306 & 1.51846501 & 1.17883595 & -4899.0798 & 45.2281041 & -96.788275 \\
4 & 0.00834089 & 1.07692251 & 0.75226045 & 1.5348003 & 1.14577232 & -4898.9955 & 131.637516 & -107.77617 \\
5 & 0.01126962 & 1.07242358 & 0.7568497 & 1.53011913 & 1.1494212 & -4898.751 & 48.0005418 & -103.03258 \\
6 & 0.01386478 & 1.06933921 & 0.76167668 & 1.52688303 & 1.15363987 & -4898.6763 & 11.3246873 & -100.1483 \\
7 & 0.01492745 & 1.06823047 & 0.76397541 & 1.52573245 & 1.15588409 & -4898.6693 & 0.98802026 & -99.171136 \\
8 & 0.01504464 & 1.06821119 & 0.76439389 & 1.52569528 & 1.15636567 & -4898.6693 & 0.02529683 & -99.040575 \\
9 & 0.01504742 & 1.06823539 & 0.76444163 & 1.5257152 & 1.15642976 & -4898.6693 & 0.00300624 & -99.030178 \\
10 & 0.01504762 & 1.06823881 & 0.76444764 & 1.52571803 & 1.15643796 & -4898.6693 & 0.00038489 & -99.028926 \\
11 & 0.01504764 & 1.06823925 & 0.76444841 & 1.52571839 & 1.15643901 & -4898.6693 & 4.92E-05 & -99.028765 \\
12 & 0.01504765 & 1.06823931 & 0.76444851 & 1.52571844 & 1.15643915 & -4898.6693 & 6.30E-06 & -99.028745 \\
13 & 0.01504765 & 1.06823932 & 0.76444852 & 1.52571845 & 1.15643917 & -4898.6693 & 1.32E-08 & -99.028742 \\
14 & 0.01504765 & 1.06823932 & 0.76444852 & 1.52571845 & 1.15643917 & -4898.6693 & 1.69E-09 & -99.028742 \\
15 & 0.01504765 & 1.06823932 & 0.76444852 & 1.52571845 & 1.15643917 & -4898.6693 & 2.23E-10 & -99.028742 \\ \hline
\end{tabular}
}
\end{table*}
\end{appendices}

\end{document}